\documentclass{article}

\usepackage{fullpage}
\usepackage{parskip}
\usepackage{amsmath}
\usepackage{amssymb}
\usepackage{amsthm}
\usepackage{cite}
\usepackage{xcolor}
\usepackage[hidelinks]{hyperref}
\usepackage{chngcntr}
\usepackage{tcolorbox}
\usepackage{dirtytalk}
\usepackage{multicol}
\usepackage{tikz}
\usetikzlibrary{patterns}

\title{Quantum gradient estimation of Gevrey functions}
\author{Arjan Cornelissen\footnote{QuSoft / CWI / UvA, Science Park 123, 1098 XG Amsterdam, Netherlands. \texttt{arjan@cwi.nl}}}

\newcommand{\ket}[1]{\ensuremath{\left|#1\right\rangle}}
\newcommand{\bra}[1]{\ensuremath{\left\langle#1\right|}}
\newcommand{\braket}[2]{\ensuremath{\left\langle#1\middle|#2\right\rangle}}
\newcommand{\norm}[1]{\ensuremath{\left\|#1\right\|}}
\newcommand{\R}{\ensuremath{\mathbb{R}}}
\newcommand{\N}{\ensuremath{\mathbb{N}}}
\newcommand{\C}{\ensuremath{\mathbb{C}}}
\renewcommand{\P}{\ensuremath{\mathbb{P}}}
\newcommand{\E}{\ensuremath{\mathbb{E}}}
\renewcommand{\O}{\ensuremath{\mathcal{O}}}
\newcommand{\F}{\ensuremath{\mathcal{F}}}
\newcommand{\G}{\ensuremath{\mathcal{G}}}
\newcommand{\A}{\ensuremath{\mathcal{A}}}
\newcommand{\B}{\ensuremath{\mathcal{B}}}
\renewcommand{\d}{\ensuremath{\mathrm{d}}}
\renewcommand{\vec}[1]{\ensuremath{\mathbf{#1}}}

\DeclareMathOperator{\QFT}{QFT}
\DeclareMathOperator{\round}{round}

\DeclareMathOperator{\Var}{Var}

\newcounter{blockcounter}[section]
\newcounter{definitioncounter}
\newcounter{theoremcounter}
\newcounter{lemmacounter}
\newcounter{algorithmcounter}
\counterwithin{definitioncounter}{section}
\counterwithin{theoremcounter}{section}
\counterwithin{lemmacounter}{section}
\counterwithin{algorithmcounter}{section}

\newenvironment{block}[3]{
	\begin{center}
	\begin{tcolorbox}[standard jigsaw, sharp corners, width=.97\textwidth, left=.5em, top=.1em, bottom = .1em, right=.5em, opacityback=0, boxrule=.03cm]
		\textbf{#1 #2 (#3):}
		
}{
	\end{tcolorbox}
	\end{center}
}
\newenvironment{definition}[1]{
\setcounter{definitioncounter}{\theblockcounter}
\refstepcounter{blockcounter}
\refstepcounter{definitioncounter}
\begin{block}{Definition}{\thedefinitioncounter}{#1}
}{\end{block}}
\newenvironment{theorem}[1]{
\setcounter{theoremcounter}{\theblockcounter}
\refstepcounter{blockcounter}
\refstepcounter{theoremcounter}
\begin{block}{Theorem}{\thetheoremcounter}{#1}
}{\end{block}}
\newenvironment{lemma}[1]{
\setcounter{lemmacounter}{\theblockcounter}
\refstepcounter{blockcounter}
\refstepcounter{lemmacounter}
\begin{block}{Lemma}{\thelemmacounter}{#1}
}{\end{block}}
\newenvironment{algorithm}[1]{
\setcounter{algorithmcounter}{\theblockcounter}
\refstepcounter{blockcounter}
\refstepcounter{algorithmcounter}
\begin{block}{Algorithm}{\thealgorithmcounter}{#1}
}{\end{block}}
\begin{document}
	\maketitle
	
	\section*{Abstract}
	
	Gradient-based numerical methods are ubiquitous in optimization techniques frequently applied in industry to solve practical problems. Often times, evaluating the objective function is a complicated process, so estimating the gradient of a function with as few function evaluations as possible is a natural problem.
	
	We investigate whether quantum computers can perform $\ell^{\infty}$-approximate gradient estimation of multivariate functions $f : \R^d \to \R$ with fewer function evaluations than classically. Following previous work by Jordan \cite{Jordan05} and Gily\'en et al.\ \cite{Gilyen18}, we prove that one can calculate an $\ell^{\infty}$-approximation of the gradient of $f$ with a query complexity that scales sublinearly with $d$ under weaker smoothness conditions than previously considered.
	
	Furthermore, for a particular subset of smoothness conditions, we prove a new lower bound on the query complexity of the gradient estimation problem, proving essential optimality of Gily\'en et al.'s gradient estimation algorithm in a broader range of parameter values, and affirming the validity of their conjecture~\cite{Gilyen18}. Moreover, we improve their lower bound qualitatively by showing that their algorithm is also optimal for functions that satisfy the imposed smoothness conditions globally instead of locally. Finally, we introduce new ideas to prove lower bounds on the query complexity of the $\ell^p$-approximate gradient estimation problem where $p \in [1,\infty)$, and prove that lifting Gily\'en et al.'s algorithm to this domain in the canonical manner is essentially optimal.
	
	\section{Introduction}
	
	Function optimization is a fundamental problem in mathematics and computer science. It finds many real-world applications and is typically used as a tool to tweak continuous parameters to maximize profit or minimize cost. As the field of quantum computing is progressing at a fast pace, the question whether quantum effects can be used to speed up the process of function optimization arises naturally.
	
	There exist many classical algorithms that perform function optimization. One of the most well-known is gradient ascent/descent. The algorithm first makes a random guess in the domain of the objective function, and then iteratively updates this guess in the direction in which the function changes fastest. This direction is determined by the \textit{gradient} of the function, and hence in every iteration this gradient is to be calculated. In this paper, we look at whether this gradient calculation step can be sped up using quantum effects.
	
	The functions that one tries to optimize are typically not given in closed form, and hence calculating the gradient can often not be done using analytical methods. Instead, one usually treats the function as a \textit{black box}, and then resorts to numerical methods that \textit{estimate} the gradient based on several function evaluations. The efficiency of these methods is typically measured in the number of function evaluations required. We use this black box model to evaluate the efficiency of our methods.
	
	To guarantee that the numerical methods employed yield an accurate estimate of the gradient, one often imposes some smoothness conditions on the objective function. Typically, one requires that its higher order (partial) derivatives are bounded or decaying. In this work, we consider an infinite family of smoothness conditions, and prove that the algorithm we construct produces accurate results under these restrictions.
	
	Finally, in the black box model, it is customary to investigate whether the number of queries to the black box can be lower bounded. In the classical setting, it is not hard to show that under any reasonable smoothness restrictions, one needs a number of queries linear in the dimension of the domain of the objective function. We obtain sublinear dependence on the dimension for a large part of the infinite family of smoothness conditions, and prove optimality for a considerable portion thereof.
	
	\subsection{Relation to earlier work on gradient estimation}
	
	This section covers how the results of this paper relate to earlier work. First, we elaborate on the statement of the problem, then discuss the input model, and subsequently consider the smoothness conditions. Finally, we state the results that we obtained, and how they compare with results obtained in previous works. The key ideas of the results mentioned in this section can be found in \autoref{subsec:algorithm_idea} and \autoref{subsec:lower_bound_idea}, and rigorous justification can be found in \autoref{sec:algorithm} and \autoref{sec:lower_bound}.
	
	The problem of gradient estimation was first considered by Jordan \cite{Jordan05}, and subsequently the results were generalized and improved by Gily\'en et al.~\cite{Gilyen18}. More specifically, Gily\'en et al.\ considered the the problem of estimating the gradient of a function $f : \R^d \to \R$ with high probability up to $\varepsilon$-precision coordinate-wise (i.e.\ up to $\ell^{\infty}$-norm). We will be looking at a slightly more general version of the problem, where one attempts to estimate the gradient of a function $f : \R^d \to \R$ with high probability up to $\varepsilon$-precision with respect to the $\ell^p$-norm, where $p \in [1,\infty]$. We refer to this problem as the \textit{gradient estimation problem w.r.t.\ the $\ell^p$-norm}. We can compare our results for $p = \infty$ to those found by Gily\'en et al.
	
	Gily\'en et al.\ motivated encoding the objective function $f : \R^d \to \R$ into a black box in the following manner. Let $G \subseteq \R^d$ be a set of points in the domain and let $\{\ket{\vec{x}} : \vec{x} \in G\}$ form an orthonormal set of states. Then, we assume to have access to the function $f$ via the following quantum operation:
	\begin{equation}
		\label{eq:phase_oracle}
		O_{f,G} : \ket{\vec{x}} \mapsto e^{if(\vec{x})}\ket{\vec{x}}.
	\end{equation}
	The quantum operation $O_{f,G}$ is referred to as a \textit{phase oracle}.\footnote{The formal definition can be found in \autoref{def:phase_oracle}.} Gily\'en et al.~showed that a variety of input models can be converted to this setting with an overhead that is at most logarithmic in the precision. For the details, we refer to \cite{Gilyen18}, especially to Section 4 and Appendix B. In this text, we will restrict our attention to this input model.
	
	Gily\'en et al.~considered the following smoothness condition on $f$.\footnote{This is not the exact smoothness condition that was investigated by Gily\'en et al., but it is easily shown that the results that are obtained using their bound are equal to the ones obtained with this bound, up to constant factors.} For some $c > 0$, $\sigma \in \R$, and all $\vec{x} \in \R^d$, $k \in \N_0$ and multi-indices $\alpha \in [d]^k$:\footnote{We will use the following notational convenience: for any $n \in \N$, $[n] = \{1,2, \dots, n\}$.}
	\begin{equation}
		\label{eq:smoothness_condition}
		|\partial_{\alpha}f(\vec{x})| \leq \frac12c^k(k!)^{\sigma}.
	\end{equation}
	Here we denote $\partial_{\alpha} = \partial_{\alpha_1}\partial_{\alpha_2} \cdots \partial_{\alpha_k}$, i.e., consecutive partial differentiation with respect to the coordinates $\alpha_1, \dots, \alpha_k$. A closely related smoothness condition has been studied before by Gevrey~\cite{Gevrey18}, so we will refer to this smoothness condition as the \textit{Gevrey condition}.
	
	Gily\'en et al.\ arrived at two results. First, they constructed an algorithm that solves the gradient estimation problem w.r.t.\ the $\ell^{\infty}$-norm under the promise that the function satisfies the Gevrey condition for some $\sigma \leq \frac12$. Second, they proved a lower bound on the query complexity of any algorithm that solves the gradient estimation problem w.r.t.\ the $\ell^{\infty}$-norm, whenever one restricts the allowed inputs to all functions that satisfy the Gevrey condition for some $\sigma \geq \frac12$. The results are shown in \autoref{tbl:results}.
	
	\begin{table}[h!]
		\centering
		\begin{tabular}{c|c|c|c|c}
			& \multicolumn{4}{c}{Query complexity to $O_{f,G}$} \\\cline{2-5}
			Smoothness parameter & \multicolumn{2}{c|}{Gily\'en et al.'s results} & \multicolumn{2}{c}{Our results} \\\cline{2-5}
			in the Gevrey condition & Algorithm & Lower bound & Algorithm & Lower bound \\\hline
			$\sigma \in \left[0,\frac12\right)$ & $\widetilde{\O}\left(\frac{c\sqrt{d}}{\varepsilon}\right)$ & $-$ & $\widetilde{O}\left(\frac{cd^{\frac12+\frac1p}}{\varepsilon}\right)$ & $\Omega\left(\frac{cd^{\frac12+\frac1p}}{\varepsilon}\right)$ \\
			$\sigma = \frac12$ & $\widetilde{\O}\left(\frac{c\sqrt{d}}{\varepsilon}\right)$ & $\Omega\left(\frac{c\sqrt{d}}{\varepsilon}\right)$ & $\widetilde{O}\left(\frac{cd^{\frac12+\frac1p}}{\varepsilon}\right)$ & $\Omega\left(\frac{cd^{\frac12+\frac1p}}{\varepsilon}\right)$ \\
			$\sigma \in \left(\frac12,1\right]$ & $-$ & $\Omega\left(\frac{c\sqrt{d}}{\varepsilon}\right)$ & $\widetilde{O}\left(\frac{cd^{\sigma+\frac1p}}{\varepsilon}\right)$ & $\Omega\left(\frac{cd^{\frac12+\frac1p}}{\varepsilon}\right)$
		\end{tabular}
		\caption{Comparison between Gily\'en et al.~\cite{Gilyen18} and our results.}
		\label{tbl:results}
	\end{table}

	Our results are also shown in \autoref{tbl:results}. For the gradient estimation problem w.r.t.\ the $\ell^{\infty}$-norm, we construct an algorithm for functions that satisfy the Gevrey condition with $\frac12 < \sigma \leq 1$ and we prove a query complexity lower bound for functions that satisfy the Gevrey condition with $0 \leq \sigma < \frac12$. Finally, we generalize all these results to estimating the gradient w.r.t.\ the $\ell^p$-norm.
	
	The algorithm we construct to solve the problem stated above is essentially the same as the one employed by Gily\'en et al., with minor tweaking of the parameters, and a slightly more direct proof of the lower bound on the success probability. The key ideas of this algorithm are described in \autoref{subsec:algorithm_idea}, and the algorithm is presented in full detail in \autoref{sec:algorithm}.
	
	We improve the lower bound proof of Gily\'en et al.\ in three ways. First, we use different objective functions, which satisfy the smoothness condition in \autoref{eq:smoothness_condition} with $\sigma = 0$, rather than $\sigma = \frac12$. Second, these new objective functions satisfy the Gevrey condition globally, which provides a qualitatively stronger result. Finally, we show that for any $p \in [1,\infty)$ we can reduce the argument to the case where $p = \infty$. The key details are elaborated upon in \autoref{subsec:lower_bound_idea}, and the full proof is presented in \autoref{sec:lower_bound}.
	
	We remark that Gily\'en et al.'s results show optimality in the case where $\sigma = \frac12$ and $p = \infty$. They also conjectured that their algorithm was optimal in the case where $\sigma \in [0,\frac12)$ and $p = \infty$. We increase the region of optimality to $\sigma \in [0,\frac12]$ and $p \in [1,\infty]$, and hence prove their conjecture.
	
	\subsection{Key ideas for the quantum gradient estimation algorithm}
	\label{subsec:algorithm_idea}
	
	In this section, we cover the key ideas that constitute the quantum gradient estimation algorithm. We start by introducing a very naive gradient estimation method. Then, we will improve this method using some more sophisticated numerical methods. Finally, we show how the quantum Fourier transform can speed up the algorithm even further.
	
	\subsubsection{Naive gradient estimation method}
	
	In this subsection, we analyze the most straightforward gradient estimation method. We first restrict to the one-dimensional case, i.e., $d = 1$, and then generalize to higher dimensions.
	
	Suppose we have a function $f : \R \to \R$, which satisfies the Gevrey condition in \autoref{eq:smoothness_condition} for some $c > 0$ and $\sigma \in \R$. We consider this function to be a black box, i.e., we can only access it by plugging in a $x$ and obtaining $f(x)$. The smoothness condition in the one-dimensional case can be rewritten as follows:
	\begin{equation}
		\label{eq:Gevrey_condition_1d}
		\forall x \in \R, \forall k \in \N_0, \qquad \left|f^{(k)}(x)\right| \leq \frac12c^k(k!)^{\sigma}.
	\end{equation}
	Now, suppose that we want to estimate the derivative of $f$ at $0$ up to precision $\varepsilon > 0$, i.e., we want to find a $g \in \R$ such that $|g - f'(0)| \leq \varepsilon$. One of the easiest methods to obtain estimates of $f'(0)$ is to choose some $r > 0$ and evaluate
	\begin{equation}
		\label{eq:derivative_estimation}
		g = \frac{f(r) - f(0)}{r}.
	\end{equation}
	When does this method yields an $\varepsilon$-approximate estimate of $f'(0)$? In \autoref{fig:jordan_approach}, we have drawn a function $f$ and the line tangent to $f$ at $x = 0$, given by $x \mapsto f(0) + f'(0)x$. Around this tangent line, we have drawn a cone whose sides have slopes that differ by exactly $\varepsilon$ from $f'(0)$.
	
	\begin{figure}[h!]
		\centering
		\begin{tikzpicture}
			\draw[->] (-2,0) -- (2,0) node[right] {$x$};
			\draw[->] (0,0) node[below] {$0$} -- (0,3) node[above] {$y$};
			\path[fill=black!20] (-2,.1) -- (2,2.9) -- (2,2.1) -- (-2,.9);
			\draw[domain=-2:2,thick] plot ({\x},{.5*sin(deg(\x))+1.5}) node[right] {$f(x)$};
			\draw[dashed] (-2,.5) -- (2,2.5) node[right] {$f(0) + f'(0)x$};
			\draw[dotted] (-1.25,3) -- (-1.25,0) node[below] {$-r$};
			\draw[dotted] (1.25,3) -- (1.25,0) node[below] {$r$};
		\end{tikzpicture}
		\caption{Whenever $f$ remains within the shaded cone at $x = r$, we can guarantee that the estimate of the derivative is $\varepsilon$-precise.}
		\label{fig:jordan_approach}
	\end{figure}
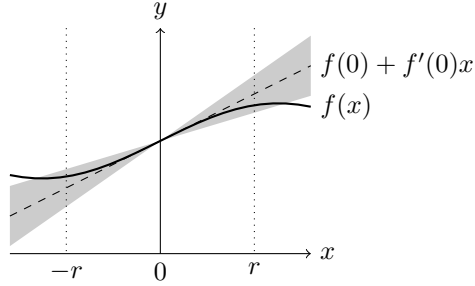

	For the derivative estimation method outlined in \autoref{eq:derivative_estimation} to yield an $\varepsilon$-precise estimate of $f'(0)$, we must choose $r > 0$ such that we can guarantee that $(r,f(r))$ is located in the cone. In other words, we must guarantee that
	\[|f(r) - (f(0) + f'(0)r)| \leq \varepsilon r.\]
	Let's bound the left-hand side using Taylor's theorem and the smoothness condition in \autoref{eq:Gevrey_condition_1d}. Observe that there exists a $\xi \in [0,r]$ such that
	\[|f(r) - (f(0) + f'(0)r)| \leq \frac12 |f''(\xi)| r^2 \leq \frac12 \cdot \left(\frac12 c^2 2^{\sigma}\right) \cdot r^2 = \frac{c^22^{\sigma}r}{4} \cdot r,\]
	so we require
	\[\frac{c^22^{\sigma}r}{4} \leq \varepsilon \qquad \Leftrightarrow \qquad r \leq \frac{4\varepsilon}{c^22^\sigma} = \Theta\left(\frac{\varepsilon}{c^2}\right).\]
	Let's choose $r = \Theta(\varepsilon/c^2)$. In \autoref{eq:derivative_estimation}, we divide by $r$, which means that our function evaluations must be at least $\Theta(\varepsilon^2/c^2)$-precise to ensure that we can calculate $g$ up to $\Theta(\varepsilon)$ precision. Using the phase estimation algorithm to perform $\Theta(\varepsilon^2/c^2)$-precise function evaluations, we must perform $\Theta(c^2/\varepsilon^2)$ queries to the phase oracle of $f$.
	
	Finally, if $d \in \N$ and $f : \R^d \to \R$ satisfies \autoref{eq:smoothness_condition} for some $c > 0$ and $\sigma \in \R$, then we can perform the above method in each dimension separately. So, we can estimate the gradient of $f$ evaluated at $\vec{0}$, i.e., $\nabla f(\vec{0})$, $\varepsilon$-precise coordinate-wise with high probability using $\Theta(c^2d/\varepsilon^2)$ queries to the phase oracle of $f$.
	
	This trivial approach has query complexity $\Theta(c^2d/\varepsilon^2)$. In the next subsection we will make a start with improving it.
	
	\subsubsection{Improvement using function smoothing}
	
	In this subsection, we use some more sophisticated numerical methods to improve the quantum gradient estimation algorithm outlined in the previous subsection. To that end, we again restrict our attention to the one-dimensional case first, and then generalize to higher dimensions.
	
	Suppose we are in the one-dimensional setting, i.e., we have a function $f : \R \to \R$ which for some $c > 0$ and $\sigma \in \R$ satisfies the one-dimensional Gevrey condition in \autoref{eq:Gevrey_condition_1d}.	The main problem with the above method was that we had to choose $r$ very small to ensure that $f$ remained in the shaded cone of \autoref{fig:jordan_approach}. The main idea in this section is to modify the function $f$ so that it stays in the cone for longer, and hence so that we can choose larger values for $r$.
	
	Whenever $\sigma \leq 1$, we can write $f$ in terms of its Taylor series:
	\[f(x) = \sum_{k=0}^{\infty} \frac{f^{(k)}(0)}{k!}x^k.\]
	If $\sigma < 1$, then this series is guaranteed to converge, i.e., the above relation holds for all $x \in \R$, and if $\sigma = 1$, then it converges at least for all $x$ in the interval $(-1/c,1/c)$.
	
	The key idea is that we can use linear combinations of $f$ to cancel the lowest order Taylor terms. To that end, we choose $m \in \N$ arbitrarily, and define, for all real finite sequences $a = (a_\ell)_{\ell=-m}^m$:
	\[f_{a,2m} : \R \to \R, \qquad f_{a,2m}(x) = \sum_{\ell=-m}^m a_{\ell}f(\ell x).\]
	We can now plug in the Taylor series and observe that for all $x \in (-1/(cm),1/(cm))$:\footnote{The Taylor series always converges absolutely on the interior of its region of convergence, so we can rearrange terms in any way we like. This justifies the exchange of the summation signs.}
	\[f_{a,2m}(x) = \sum_{\ell=-m}^m a_{\ell} \sum_{k=0}^{\infty} \frac{f^{(k)}(0)}{k!} (\ell x)^k = \sum_{k=0}^{\infty} \frac{f^{(k)}(0)}{k!} x^k \cdot \sum_{\ell=-m}^m a_{\ell}\ell^k.\]
	Now, we want to choose $a$ such that $x \mapsto f_{a,2m}(x)$ is close to $x \mapsto f(0) + f'(0)x$, i.e., that it stays in the shaded cone in \autoref{fig:jordan_approach} as long as possible. To that end, we require that as many of the lowest order Taylor terms as possible vanish, except for the constant and linear one. In other words, we require that for all $k \in \{0, \dots, 2m\}$ (here we use the convention that $0^0 = 1$)
	\begin{equation}
		\label{eq:coefficient_relation}
		\sum_{\ell=-m}^m a_{\ell}\ell^k = \begin{cases}
			1, & \text{if } k \in \{0,1\}, \\
			0, & \text{otherwise}.
		\end{cases}
	\end{equation}
	We denote the solution by $a^{(2m)}$, which can be given in closed form\footnote{For the exact values, see \autoref{def:coefficients}.}. We abbreviate the resulting function $f_{a^{(2m)},2m}$ to $f_{(2m)}$. These functions $f_{(2m)}$, we refer to as \textit{smoothings of $f$}. In \autoref{fig:smoothing}, we plot some smoothings of $f(x) = \sin(x)$. One can see that the region of approximate linearity is enlarged when $m$ is increased.
	
	\begin{figure}[h!]
		\centering
		\includegraphics[width=.5\textwidth]{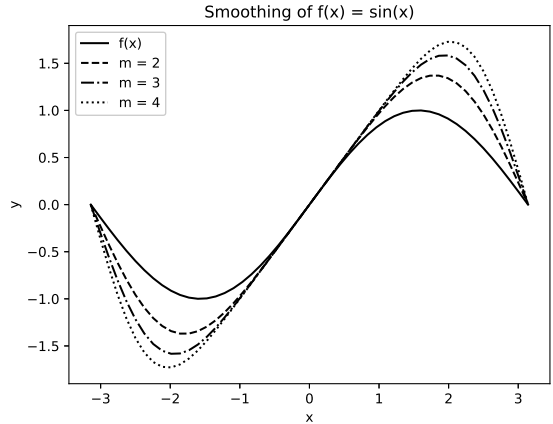}
		\caption{Smoothings of $f(x) = \sin(x)$.}
		\label{fig:smoothing}
	\end{figure}
	
	For $m = 1$, one can easily verify that $a^{(2)} = (-1/2,1,1/2)$ satisfies \autoref{eq:coefficient_relation}. Hence, when one applies \autoref{eq:derivative_estimation} to the function $f_{(2)}$ instead of $f$, one obtains
	\[g = \frac{f_{(2)}(r) - f_{(2)}(0)}{r} = \frac{-\frac12f(-r) + f(0) + \frac12f(r) - f(0)}{r} = \frac{f(r) - f(-r)}{2r}.\]
	So, by estimating the derivative of the smoothing $f_{(2)}$ of $f$, we recover the simple central difference scheme. Similarly, if we increase $m$, we recover the higher order central difference schemes.
	
	In the main body of this text, we quantify how much bigger this region of approximate linearity becomes upon increasing $m$. We obtain:\footnote{The multivariate version of this statement is proven in \autoref{lem:justification_linearity}.}
	\[|f_{(2m)}(r) - (f(0) + f'(0)r)| = \O(r^{2m+1}).\]
	Gily\'en et al.~\cite{Gilyen18} have shown that one can implement a phase oracle that accesses $f_{(2m)}$ using just $\widetilde{\O}(m)$ queries to the phase oracle that accesses $f$. Finally, we find that if we let $m$ scale logarithmically in $c/\varepsilon$, then it suffices to choose $r = \widetilde{O}(1/c)$ to ensure that $f_{(2m)}(r)$ is within the shaded cone in \autoref{fig:jordan_approach}.\footnote{The $1/c$ scaling is present in the definition of $r$ in \autoref{alg:QGE}.} Following the same arguments as outlined in the previous subsubsection, we find that we can estimate the derivative up to precision $\Theta(\varepsilon)$ using $\widetilde{\Theta}(\varepsilon/c)$-precise evaluations of $f$. Hence, the query complexity becomes $\widetilde{\Theta}(c/\varepsilon)$. In the multidimensional case, we can still do this procedure in each dimension separately, and hence the resulting query complexity becomes $\widetilde{\Theta}(cd/\varepsilon)$.
	
	So far, apart from the phase estimation procedure to obtain binary function evaluations from phase oracles, we have not yet used any techniques that are inherently quantum. We have, however, reduced the query complexity from quadratic in $c/\varepsilon$ to linear, but the dependence on $d$ remained unaffected. In the next subsubsection, we will investigate how we can use quantum effects to reduce the query complexity dependence on the dimension $d$.
	
	\subsubsection{Improvement using quantum Fourier transform}
	
	In this subsection, we elaborate on how the quantum Fourier transform can be used to speed up the methods described in the previous subsections. To that end, we first revise the definition of the $n$-qubit quantum Fourier transform, where $n \in \N$ and $j \in \{-2^{n-1}, \dots, 2^{n-1} - 1\}$:\footnote{We assume that the computational basis states $\ket{j}$ are labeled by signed $n$-bit integers $j \in \{-2^{n-1}, \dots, 2^{n-1}-1\}$.}
	\[\QFT_{2^n}\ket{j} = \frac{1}{\sqrt{2^n}} \sum_{k=-2^{n-1}}^{2^{n-1}-1} e^{\frac{2\pi i}{2^n} \cdot jk} \ket{k}.\]
	If we apply the quantum Fourier transform to a computational basis state $\ket{j}$, we obtain a state in which the complex angle of the amplitude of the $k$-th computational basis state depends linearly on $k$. Moreover, the slope of this linear dependence is proportional to $j$.
	
	The key idea is that this effect can be inverted. If we have a uniform superposition of computational basis states $\ket{k}$ with phases $e^{iak}$, for some $a \in \R$ and for each $k \in \{-2^{n-1}, \dots, 2^{n-1}-1\}$, we can employ the \textit{inverse quantum Fourier transform} to obtain an estimate of the real parameter $a$. In other words, the inverse quantum Fourier transform allows for recovering the slope of the phase as a function of $k$. We obtain the following relation, where the approximation symbol is justified by the robustness of the quantum Fourier transform (see for instance \cite{Nielsen00}, Equation 5.34). For all $n \in \N$ and $a \in (-2\pi/3,2\pi/3)$,
	\[\QFT^{\dagger}_{2^n} \left[\frac{1}{\sqrt{2^n}} \sum_{k=-2^{n-1}}^{2^{n-1}-1} e^{iak}\ket{k}\right] \approx \ket{\text{round}\left(\frac{2^na}{2\pi}\right)}.\]
	This idea generalizes well to higher dimensions. Suppose we have $d \in \N$ registers in a product state, each of which is in a state that has a linearly varying phase. Then, we can apply the inverse quantum Fourier transform on each of the registers individually, and recover each of the slopes. For all $n \in \N$ and vectors $\vec{a} \in (-2\pi/3,2\pi/3)^d$,
	\begin{align*}
		\left(\QFT_{2^n}^{\dagger}\right)^{\otimes d} \left[\frac{1}{\sqrt{2^{nd}}} \sum_{\vec{k} \in \{-2^{n-1}, \dots, 2^{n-1}-1\}^d} e^{i\vec{a} \cdot \vec{k}}\ket{\vec{k}}\right] &= \bigotimes_{j=1}^d \QFT_{2^n}^{\dagger} \left[\frac{1}{\sqrt{2^n}} \sum_{k_j=-2^{n-1}}^{2^{n-1}} e^{ia_jk_j}\ket{k_j}\right] \\
		&\approx \bigotimes_{j=1}^d \ket{\round\left(\frac{2^na_j}{2\pi}\right)} = \ket{\round\left(\frac{2^n}{2\pi}\vec{a}\right)}.
	\end{align*}
	The above relation motivates a surprisingly simple quantum algorithm that estimates the gradient. We define a uniform grid centered around the origin, with side length $r > 0$. The points of this grid are denoted by $\vec{x_k}$,\footnote{The addition of $\vec{\frac12}$ is to make sure that the grid is centered around the origin, as this maps elements from the set $\{-2^{n-1}, \dots, 2^{n-1}-1\}^d$ to the set $\{-2^{n-1}+\frac12, \dots, 2^{n-1}-\frac12\}^d$.}
	\[\forall \vec{k} \in \{-2^{n-1}, \dots, 2^{n-1}-1\}^d, \qquad \vec{x}_{\vec{k}} = \frac{r}{2^n}\left(\vec{k} + \vec{\frac12}\right),\]
	and the collection of all these points, we denote by $G$:
	\[G = \{\vec{x}_{\vec{k}} : \vec{k} \in \{-2^{n-1}, \dots, 2^{n-1}-1\}^d\} \subseteq \left[-\frac{r}{2},\frac{r}{2}\right]^d.\]
	We present a graphical depiction of the set $G$ in \autoref{fig:grid}.
	
	\begin{figure}[h!]
		\centering
		\begin{tikzpicture}[element/.style={draw,circle,fill,inner sep=.1em}]
			\draw[->] (-1.5,0) -- (1.5,0) node[right] {$x$};
			\draw[->] (0,-1.5) -- (0,1.5) node[above] {$y$};
			\foreach \x in {-.875,-.625,...,.875}
				\foreach \y in {-.875,-.625,...,.875}
					\node[element] at (\x,\y) {};
			\draw[dashed] (-1,-1) -- (1,-1) -- (1,1) -- (-1,1) -- cycle;
			\draw[<->] (-1,-1.1) -- node[below right] {$r$} (1,-1.1);
		\end{tikzpicture}
		\caption{Graphical depiction of the grid employed in the gradient estimation algorithm where $d = 2$ and $n = 3$. The dots denote the elements of $G$. The side length of the grid is $r$, and it is placed symmetrically around the origin. The number of points along each direction is $2^n$.}
		\label{fig:grid}
	\end{figure}
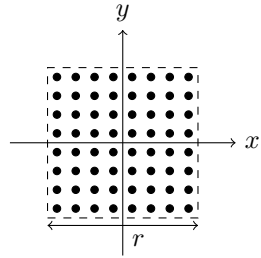
	
	Moreover, we associate computational basis states to the elements in this grid:
	\[\forall \vec{k} \in \{-2^{n-1}, \dots, 2^{n-1}-1\}^d, \qquad \ket{\vec{x}_\vec{k}} = \ket{\vec{k}} = \bigotimes_{j=1}^d \ket{k_j}.\]
	The algorithm uses ideas from the previous subsection. We define a smoothing of $f : \R^d \to \R$, similarly as in the previous section. For all $m \in \N$:
	\[f_{(2m)} : \R^d \to \R, \qquad f_{(2m)}(\vec{x}) = \sum_{\ell=-m}^m a_{\ell}^{(2m)}f(\ell\vec{x})\]
	Similarly as in the previous section, $f_{(2m)}$ is close to $f(\vec{0}) + \nabla f(\vec{0}) \cdot \vec{x}$. Moreover, Gily\'en et al.~\cite{Gilyen18} proved that a phase oracle $O_{f_{(2m)},G}$ can be implemented using $\widetilde{\O}(m)$ queries to $O_{f,G}$.
	
	We propose the following algorithm to estimate $\nabla f(\vec{0})$ up to $\varepsilon$-precision coordinate-wise. The parameters $S$, $m$ and $r$ will be chosen later.
	\begin{enumerate}
		\setlength\itemsep{-.6em}
		\item Prepare a uniform superposition over the grid $G$ with side length $r > 0$.
		\item Apply the phase oracle $O_{f_{(2m)},G}$ a total of $S \in \N$ times.
		\item Apply the inverse quantum Fourier transform on each register individually.
		\item Measure in the computational basis and denote the resulting vector $\vec{h} \in \{-2^{n-1}, \dots, 2^{n-1}-1\}^d$.
		\item Calculate
		\[\vec{g} = \frac{2\pi}{Sr}\vec{h}.\]
	\end{enumerate}
	After step 1, we have the following state:
	\[\ket{\psi_1} = \frac{1}{\sqrt{2^{nd}}} \sum_{\vec{x} \in G} \ket{\vec{x}}.\]
	After step 2:
	\[\ket{\psi_2} = \frac{1}{\sqrt{2^{nd}}} \sum_{\vec{x} \in G} e^{iSf_{(2m)}(\vec{x})}\ket{\vec{x}}.\]
	If $f_{(2m)}$ is close to linear, then $f_{(2m)}(\vec{x}) \approx f(\vec{0}) + \nabla f(\vec{0}) \cdot \vec{x}$. So, we obtain (throwing away an unimportant constant phase factor):
	\begin{equation}
		\label{eq:approximate_linearity}
		\ket{\psi_2} \approx \frac{1}{\sqrt{2^{nd}}}\sum_{\vec{x} \in G} e^{iS\nabla f(\vec{0}) \cdot \vec{x}} \ket{\vec{x}} = \frac{1}{\sqrt{2^{nd}}} \sum_{\vec{k} \in \{-2^{n-1}, \dots, 2^{n-1}-1\}^d} e^{\frac{iSr}{2^n}\nabla f(\vec{0}) \cdot \vec{k}} \ket{\vec{k}}.
	\end{equation}
	Applying the inverse quantum Fourier transform yields, approximately:
	\[\ket{\psi_3} \approx \ket{\round\left(\frac{Sr}{2\pi}\nabla f(\vec{0})\right)}.\]
	Upon measuring, we obtain $\vec{h} \approx Sr/(2\pi)\nabla f(\vec{0})$. Hence, we output $\vec{g} = 2\pi/(Sr)\vec{h} \approx \nabla f(\vec{0})$.
	
	The number of queries to the phase oracle $O_{f,G}$ in the above algorithm is given by $\widetilde{O}(mS)$. As we do not perform some procedure individually for each coordinate but rather have one procedure that determines the entries of the gradient simultaneously, there is a possibility that, after choosing the appropriate parameters $S$, $m$ and $r$, the query complexity of this algorithm scales sublinearly in $d$. To that end, we investigate what values of these parameters will ensure that the estimate of the gradient is sufficiently close to the actual value with high probability.
	
	As $\vec{h}$ is a vector of integers, it can differ from $Sr/(2\pi)\nabla f(\vec{0})$ by at least $1/2$ coordinate-wise. Hence, $\vec{g}$ can differ from $\nabla f(\vec{0})$ by at most $\pi/(Sr)$. If we want to approximate $\nabla f(\vec{0})$ up to precision $\varepsilon$, we must ensure that $\pi/(Sr) \leq \varepsilon$, i.e., $S \geq \pi/(\varepsilon r)$. But $S$ is proportional to the query complexity of the algorithm, so to minimize it, we want to choose $r$ as big as possible under the restriction that the approximate linearity used in \autoref{eq:approximate_linearity} is justified. Note that the region of approximate linearity must now contain the entire $d$-dimensional grid $G$ on which the function is evaluated, instead of the one-dimensional interval that we considered in the previous subsection.
	
	The key technique that is used in evaluating how far $f_{(2m)}$ is from being linear, is the method of bounding the \textit{second moments of higher order bounded tensors}, as first described by Gily\'en et al.~\cite{Gilyen18}. Intuitively, one can imagine that the function is most likely to be far from linear in the corners of the grid $G$, as these are furthest from the origin. However, this method exploits the fact that if $f$ is sufficiently smooth, it cannot be far from linear in all corners of $G$ at the same time. One can see this by looking at all second order terms in the two-dimensional case: $(x,y) \mapsto x^2$ and $(x,y) \mapsto y^2$ are positive everywhere, hence also in the corners of any grid $G$. However, $(x,y) \mapsto xy$ is only positive in two of the four corners, and is negative in the other two corners, so it can only amplify the deviation from linear in half of the corners, and will cancel this deviation in the other two corners. It is this effect that is very carefully exploited in the higher dimensional and higher order case, using the method proposed by Gily\'en et al.
	
	To justify the approximation symbol in \autoref{eq:approximate_linearity}, we show that it suffices to choose:\footnote{\autoref{lem:justification_linearity} justifies these choices.}
	\[m = \begin{cases}
		\Theta\left(\log(\frac{c\sqrt{d}}{\varepsilon})\right), & \text{if } \sigma \leq \frac12, \\
		\Theta\left(\log(\frac{cd^{\sigma}}{\varepsilon})\right), & \text{if } \sigma \in \left(\frac12,1\right],
	\end{cases} \qquad r = \begin{cases}
		\widetilde{\Theta}\left(\frac{1}{c\sqrt{d}}\right), & \text{if } \sigma \leq \frac12, \\
		\widetilde{\Theta}\left(\frac{1}{cd^{\sigma}}\right), & \text{if } \sigma \in \left(\frac12,1\right].
	\end{cases}\]
	Choosing $S = \Theta(1/(r\varepsilon))$, the query complexity becomes
	\begin{equation}
		\label{eq:query_complexities}
		\widetilde{\O}(mS) = \begin{cases}
			\widetilde{\O}\left(\frac{c\sqrt{d}}{\varepsilon}\right), & \text{if } \sigma \leq \frac12, \\
			\widetilde{\O}\left(\frac{cd^{\sigma}}{\varepsilon}\right), & \text{if } \sigma \in \left(\frac12,1\right].
		\end{cases}
	\end{equation}
	Note that for $\sigma < 1$, we obtain an improvement over the query complexity achieved in the previous subsection. In the case where $\sigma \leq \frac12$, we even obtain a quadratic speed-up in $d$.
	
	As a final note, if we want to estimate the gradient $\varepsilon$-precisely w.r.t.\ the $\ell^p$-norm, for some $p \in [1,\infty]$, then we can simply run this algorithm with the accuracy parameter $\varepsilon' = \varepsilon/d^{1/p}$. Plugging $\varepsilon'$ in the query complexities in \autoref{eq:query_complexities}, we obtain the query complexities in \autoref{tbl:results}.
	
	This completes the informal description of the quantum gradient estimation algorithm that we constructed. All the missing details can be found in \autoref{sec:algorithm}, and the exact statement of the algorithm, with the precise choice of all the parameters, can be found in \autoref{alg:QGE}.
	
	\subsection{Key ideas for the lower bound of quantum gradient estimation}
	\label{subsec:lower_bound_idea}
	
	In this section, we elaborate on the key ideas that improve on Gily\'en et al.'s proof of the lower bound on the query complexity of the gradient estimation problem. There are three main improvements, each of which we cover individually in \autoref{subsec:instance_selection}, \autoref{subsec:claw_selection} and \autoref{subsec:median_trick_substitution}.
	
	To describe how one proves lower bounds on the query complexity of the gradient estimation problem, let's first consider a toy example. Suppose we take the functions $f_0(x) = 0$ and $f_{\varepsilon}(x) = 2\varepsilon xe^{-\frac12x^2}$, where $\varepsilon > 0$ is some small positive number. These functions are close to each other:
	\[\norm{f_{\varepsilon} - f_0}_{\infty} = \sup_{x \in \R} |f_{\varepsilon}(x) - f_0(x)| = f_{\varepsilon}(1) = \frac{2\varepsilon}{\sqrt{e}}.\]
	However, their derivatives are not equal:
	\[f'(0) = 0 \qquad \text{and} \qquad f'_{\varepsilon}(0) = 2\varepsilon.\]
	Moreover, as the derivatives differ by $2\varepsilon$, any algorithm that finds approximations of the derivative with precision $\varepsilon$ must yield different outputs when run on these two instances.
	
	Because the function values are close, the corresponding phase oracles $O_{f_0,G}$ and $O_{f_\varepsilon,G}$, as introduced in \autoref{eq:phase_oracle}, act in an almost identical manner. However, any algorithm that determines the derivative up to precision $\varepsilon$ must be able to determine whether it is querying $O_{f_0,G}$ or $O_{f_{\varepsilon},G}$. Intuitively, we argue that the algorithm must therefore query these oracles many times to obtain the required \textit{discriminating power} and be able to distinguish between these two instances.
	
	One can come up with more sophisticated methods if one considers more than just two functions that the algorithm must distinguish. One of these methods is the hybrid method, and this is the method that Gily\'en et al.\ employed to prove their lower bound on the query complexity of the gradient estimation problem.
	
	In the hybrid method, one considers one \textit{central instance} of the problem. Next, one picks $N \in \N$ \textit{peripheral instances}, each of which any algorithm that solves the problem must distinguish from the central instance. Pictorially, one can think of a \textit{claw}, in which the center of the claw is formed by the central instance and the endpoints of the legs of the claw form the peripheral instances. This claw is displayed in \autoref{fig:claw}.\footnote{For those that are familiar with the quantum adversary method, the hybrid method is just the quantum adversary method without negative weights, where there is just one instance on one side of the relation.}
	
	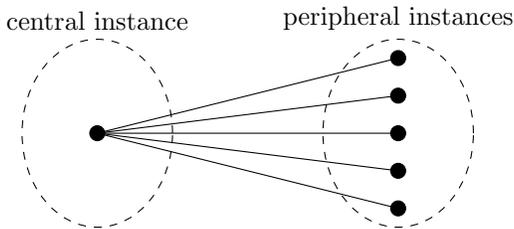
\begin{figure}[h!]
		\centering
		\begin{tikzpicture}[vertex/.style={circle, draw, fill=black, inner sep = .2em},xscale=4,yscale=.5]
			\node[vertex] at (0,0) {};
			\foreach \x in {-2,...,2}
				\draw (0,0) -- (1,\x) node[vertex] {};
			\draw[dashed] (0,0) ellipse (.25 and 2.5);
			\node[above] at (0,2.5) {central instance};
			\draw[dashed] (1,0) ellipse (.25 and 2.5);
			\node[above] at (1,2.5) {peripheral instances};
		\end{tikzpicture}
		\caption{This is a claw, as employed by the hybrid method. The single instance on the left side is referred to as the central instance, and the instances on the right side are referred to as the peripheral instances. The solid lines pair up instances that any algorithm should be able to distinguish from one another. These lines are referred to as the legs of the claw.}
		\label{fig:claw}
	\end{figure}
	
	\subsubsection{Instance selection}
	\label{subsec:instance_selection}
	
	The first improvement over the lower bound of Gily\'en et al.\ lies in the choice of peripheral instances in the hybrid method. We first describe Gily\'en et al.'s choices for the instances in the hybrid method, and then we elaborate on which ones we chose.
	
	The central instance that Gily\'en et al.\ use in the hybrid method is the following function $f_0$:
	\begin{equation}
		\label{eq:Gilyen_central_instance}
		f_0 : \R^d \to \R, \qquad f_0(\vec{x}) = 0.
	\end{equation}
	Furthermore, they choose the following $d$ peripheral instances. For some $c > 0$ and all $j \in [d]$:
	\begin{equation}
		\label{eq:Gilyen_peripheral_instances}
		f_j : \R^d \to \R, \qquad f_j(\vec{x}) = 2\varepsilon x_je^{-\frac12 c^2\norm{\vec{x}}^2}.
	\end{equation}
	These functions satisfy the Gevrey smoothness condition, \autoref{eq:smoothness_condition}, at $\vec{x} = \vec{0}$, with $c > 0$ and $\sigma \geq \frac12$. However, whether these functions satisfy the same conditions in an open region around $\vec{0}$ is not clear.
	
	The gradients of these functions are as follows:
	\[\nabla f_0(\vec{0}) = \vec{0} \qquad \text{and} \qquad \forall j \in [d], \qquad \nabla f_j(\vec{0}) = 2\varepsilon\vec{e}_j.\]
	Any quantum algorithm that estimates gradients up to precision $\varepsilon$ w.r.t.\ the $\ell^{\infty}$-norm must be able to distinguish between the central instance and any of the peripheral instances. Using these choices for central and peripheral instances, Gily\'en et al.\ obtained the lower bound results shown in \autoref{tbl:results}.
	
	Let's picture the gradients of these functions as vertices in the \textit{gradient space} $\R^d$. Moreover, let's connect the vertex corresponding to the central instance with the vertices corresponding to the peripheral instances, just like in \autoref{fig:claw}. The resulting picture is shown in \autoref{fig:gradient_claw}. Observe that the vertices form a claw, where any pair of legs is orthogonal to one another. Moreover, the length of the legs is $2\varepsilon$.
	
	\begin{figure}[h!]
		\centering
		\begin{tikzpicture}[x = {(-0.5cm,-0.5cm)},
							y = {(0.9659cm,-0.25882cm)},
							z = {(0cm,1cm)},
							scale = 2,
							vertex/.style={circle, draw, fill=black, inner sep = .2em}]
			\draw[thin,->] (-1,0,0) -- (1.5,0,0) node[below] {$x$};
			\draw[thin,->] (0,-1,0) -- (0,1.5,0) node[right] {$y$};
			\draw[thin,->] (0,0,-1) -- (0,0,1.5) node[above] {$z$};
			\node[vertex] (v0) at (0,0,0) {};
			\node[vertex] (v1) at (1,0,0) {};
			\node[vertex] (v2) at (0,1,0) {};
			\node[vertex] (v3) at (0,0,1) {};
			\draw[ultra thick] (v0) node[left] {$\nabla f_0(\vec{0})$} -- (v1) node[left] {$\nabla f_1(\vec{0})$};
			\draw[ultra thick] (v0) -- (v2) node[above] {$\nabla f_2(\vec{0})$};
			\draw[ultra thick] (v0) -- (v3) node[right] {$\nabla f_3(\vec{0})$};
		\end{tikzpicture}
		\caption{Graphical depiction of the gradients of the instances defined in \autoref{eq:Gilyen_central_instance} and \autoref{eq:Gilyen_peripheral_instances}, with $d = 3$. The central instance, $f_0$, is in the center of the claw, and the three peripheral instances are at the endpoints of the legs of the claw.}
		\label{fig:gradient_claw}
	\end{figure}
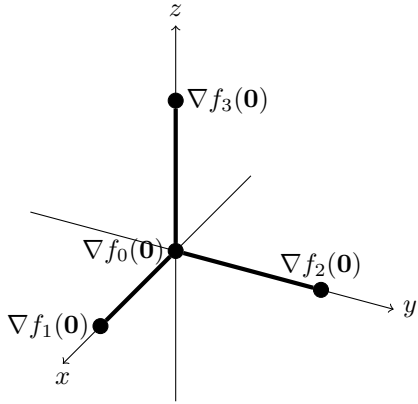

	Our first key observation is as follows: the peripheral instances can be changed, without influencing their gradients. The exact same claw as displayed in \autoref{fig:gradient_claw} can be obtained using more regular instances, i.e., functions that satisfy \autoref{eq:smoothness_condition} with $\sigma \geq 0$ rather than just $\sigma \geq \frac12$. Specifically, the peripheral instances that we select are, for all $j \in [d]$:
	\begin{equation}
		\label{eq:peripheral_instances}
		f_j : \R^d \to \R, \qquad f_j(\vec{x}) = \frac{2\varepsilon}{c}\sin(cx_j)\prod_{\underset{k\neq j}{k=1}}^d \cos(cx_k).
	\end{equation}
	We show that these functions satisfy the Gevrey smoothness condition for $\sigma = 0$ on all of $\R^d$, instead of just the point $\vec{x} = \vec{0}$. This makes the result a lot more useful in practical settings, as we now know that even if we consider objective functions that satisfy the smoothness condition on an open domain with $\sigma \geq 0$, we cannot obtain a gradient estimation algorithm with query complexity smaller than the lower bounds shown in \autoref{tbl:results}. This was not at all obvious from Gily\'en et al.'s results, hence we improve on them not only quantitatively, but also qualitatively.
	
	Using Gily\'en et al.'s hybrid method argument with the peripheral instances in \autoref{eq:peripheral_instances} allows us to conclude our results in \autoref{tbl:results} in the case where $\sigma \geq 0$ and $p = \infty$.
	
	\subsubsection{Claw selection in Hamming cube}
	\label{subsec:claw_selection}
	
	To prove the results for $p \in [1,\infty)$ shown in \autoref{tbl:results}, we need some new ideas. In particular, there are two fundamental observations that we need to develop a better lower bound for $p < \infty$.
	
	Let's revisit the claw shown in \autoref{fig:gradient_claw}. Intuitively, the shorter the legs, the closer the gradients are, so the more similar the objective functions can be, and hence the higher the query complexity we obtain. Hence, we want to minimize the length of the legs of the claw in the gradient space to maximize the lower bound on the query complexity of gradient estimation, but we cannot make the legs too short, because then we the algorithm would not be able to distinguish between the instances. The main idea is to construct a claw that has legs shorter than $2\varepsilon$, for which we can prove that the algorithm can distinguish the peripheral instances from the central instance.
	
	To that end, let's set $p = 1$ for the time being. In \autoref{fig:Hamming_cube}, we draw the unit ball in gradient space with respect to the $\ell^1$-norm. Alongside unit ball we have also drawn the Hamming cube, centered around the origin, with radius $2/d$. The vertices of the Hamming cube are gradients of particular instances of the gradient estimation problem. The precise definition of these functions is not relevant for now, but they can be found in \autoref{def:test_functions} and their properties are proven in \autoref{lem:test_functions_smoothness} and \autoref{lem:gradient_test_functions}.\footnote{In the precise definition, we use a Hamming cube of radius $73/d$, instead of $2/d$. This implies that interesting effects do not happen until $d > 73$, though, but this is very hard to visualize. Hence, for conceptual simplicity, we explain the method using a Hamming cube with radius $2/d$.}
	
	\begin{figure}[h!]
		\centering
		\begin{tikzpicture}[x = {(-0.5cm,-0.5cm)},
							y = {(0.9659cm,-0.25882cm)},
							z = {(0cm,1cm)},
							scale = .5,
							vertex/.style={circle, draw, fill=black, inner sep = .2em},
							face/.style={fill=black!20}]
			\draw (-6,0,0) -- (-5,0,0);
			\draw[->] (0,-6,0) -- (0,6,0) node[right] {$y$};
			\draw[->] (0,0,-6) -- (0,0,6) node[above] {$z$};
			\draw[thick] (3,-3,-3) -- (-3,-3,-3); % Bottom left ridge
			\draw (-3,3,-3) -- (-3,-3,-3) -- (-3,-3,3); % Back ridges
			\fill[black!10, opacity=.8] (3,-3,-3) -- (3,-3,0) -- (3,-2,0) -- (3,0,-2) -- (3,0,-3) -- cycle; % Front face, bottom left part
			\fill[black!10, opacity=.8] (-3,3,-3) -- (-3,3,0) -- (-2,3,0) -- (0,3,-2) -- (0,3,-3) -- cycle; % Right face, bottom back part
			\fill[black!10, opacity=.8] (-3,-3,3) -- (0,-3,3) -- (0,-2,3) -- (-2,0,3) -- (-3,0,3) -- cycle; % Top face, back left part
			\draw[thick] (-3,3,0) -- (-3,3,-3); % Back right ridge, bottom part
			\draw[face] (-5,0,0) -- (0,-5,0) -- (0,0,-5) -- cycle;
			\draw[face] (-5,0,0) -- (0,5,0) -- (0,0,-5) -- cycle;
			\draw[face] (5,0,0) -- (0,-5,0) -- (0,0,-5) -- cycle;
			\draw[face] (5,0,0) -- (0,5,0) -- (0,0,-5) -- cycle;
			\draw[face] (-5,0,0) -- (0,-5,0) -- (0,0,5) -- cycle;
			\draw[face] (-5,0,0) -- (0,5,0) -- (0,0,5) -- cycle;
			\draw[face] (5,0,0) -- (0,-5,0) -- (0,0,5) -- cycle;
			\draw[face] (5,0,0) -- (0,5,0) -- (0,0,5) -- cycle;
			\fill[black!10, opacity=.8] (3,3,3) -- (3,3,0) -- (3,2,0) -- (3,0,2) -- (3,0,3) -- cycle;
			\fill[black!10, opacity=.8] (3,3,-3) -- (3,3,0) -- (3,2,0) -- (3,0,-2) -- (3,0,-3) -- cycle;
			\fill[black!10, opacity=.8] (3,-3,3) -- (3,-3,0) -- (3,-2,0) -- (3,0,2) -- (3,0,3) -- cycle;
			\fill[black!10, opacity=.8] (3,3,3) -- (3,3,0) -- (2,3,0) -- (0,3,2) -- (0,3,3) -- cycle;
			\fill[black!10, opacity=.8] (3,3,-3) -- (3,3,0) -- (2,3,0) -- (0,3,-2) -- (0,3,-3) -- cycle;
			\fill[black!10, opacity=.8] (-3,3,3) -- (-3,3,0) -- (-2,3,0) -- (0,3,2) -- (0,3,3) -- cycle;
			\fill[black!10, opacity=.8] (3,3,3) -- (0,3,3) -- (0,2,3) -- (2,0,3) -- (3,0,3) -- cycle;
			\fill[black!10, opacity=.8] (-3,3,3) -- (0,3,3) -- (0,2,3) -- (-2,0,3) -- (-3,0,3) -- cycle;
			\fill[black!10, opacity=.8] (3,-3,3) -- (0,-3,3) -- (0,-2,3) -- (2,0,3) -- (3,0,3) -- cycle;
			\draw[thick] (3,3,3) -- (3,3,-3) -- (3,-3,-3) -- (3,-3,3) -- cycle; % Front ridges
			\draw[thick] (3,3,3) -- (-3,3,3); %Right top ridge
			\draw[thick] (3,3,-3) -- (-3,3,-3); % Right bottom ridge
			\draw[thick] (3,-3,3) -- (-3,-3,3); % Top left ridge
			\draw[thick] (-3,3,3) -- (-3,-3,3); % Top back ridge
			\draw[thick] (-3,3,3) -- (-3,3,0); % Back right ridge, top part
			\draw[->] (5,0,0) -- (6,0,0) node[below] {$x$};
			\node[vertex] (111) at (3,3,3) {};
			\node[below left] at (111) {$\nabla f_{(1,1,1)}(\vec{0})$};
			\node[vertex] (011) at (-3,3,3) {};
			\node[above right] at (011) {$\nabla f_{(-1,1,1)}(\vec{0})$};
			\node[vertex] (101) at (3,-3,3) {};
			\node[left] at (101) {$\nabla f_{(1,-1,1)}(\vec{0})$};
			\node[vertex] (001) at (-3,-3,3) {};
			\node[above left] at (001) {$\nabla f_{(-1,-1,1)}(\vec{0})$};
			\node[vertex] (110) at (3,3,-3) {};
			\node[below right] at (110) {$\nabla f_{(1,1,-1)}(\vec{0})$};
			\node[vertex] (010) at (-3,3,-3) {};
			\node[below right] at (010) {$\nabla f_{(-1,1,-1)}(\vec{0})$};
			\node[vertex] (100) at (3,-3,-3) {};
			\node[below left] at (100) {$\nabla f_{(1,-1,-1)}(\vec{0})$};
		\end{tikzpicture}
		\caption{Two shapes are shown. The non-transparent object is the unit ball w.r.t.\ the $\ell^1$-norm. The transparent cube is the Hamming cube with side length $4/3$. All vertices of the Hamming cube are outside the $\ell^1$ unit ball. However, the edges are shorter than $2$, which would be the length of the edges of the $\ell^{\infty}$ unit ball.}
		\label{fig:Hamming_cube}
	\end{figure}
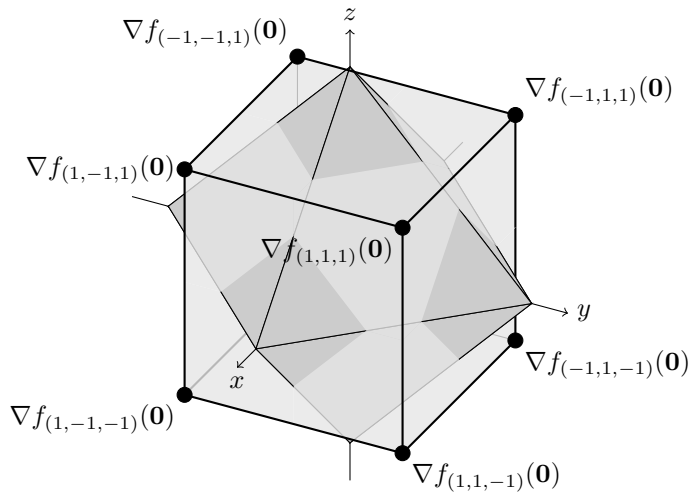

	The first key observation we make is that for all $d \in \N$, all vertices of the Hamming cube are outside the unit ball, even though the side length of the Hamming cube becomes ever smaller as $d$ increases. If we had drawn the unit ball with respect to the $\ell^{\infty}$-norm instead of the $\ell^1$-norm, then the Hamming cube would have been completely contained in this unit ball for $d > 2$.
	
	The second key observation is that every vertex of the Hamming cube can be considered as the center of a claw similar to the one shown in \autoref{fig:gradient_claw}, which provides us with $2^d$ different locations where we could perform the hybrid method. Moreover,the legs of this claw become ever shorter when $d$ increases, meaning that if we can show that any algorithm that solves the gradient estimation problem can distinguish between neighboring vertices on this Hamming cube, then we can prove a better lower bound than in the $\ell^{\infty}$-case.
	
	The main idea of the improved lower bound is a bit more complicated that what we mentioned in the previous paragraph. We show that there is at least one vertex in this Hamming cube, which any algorithm that solves the gradient estimation problem must distinguish from at least a quarter of its neighbors. With this vertex as the central instance, and the $d/4$ neighbors as peripheral instances, we perform the hybrid method. This argument then gives the lower bound results displayed in \autoref{tbl:results} for $p < \infty$. The precise details of the arguments mentioned in this section are presented in \autoref{lem:vertex_approximation}, \autoref{lem:edge_separation} and \autoref{thm:lower_bound_proof}.
	
	\subsubsection{Modified ``median trick"}
	\label{subsec:median_trick_substitution}
	
	One technical issue arises with the proof method presented above. Proving that there exists a vertex that is distinguished from at least a constant fraction of its neighbors only seems to work when we require the success probability of the algorithm to be pretty high (e.g., at least $17/18$). Hence, we need some extra work to prove that there are no quantum gradient estimation algorithms that achieve a slightly smaller success probability with a significantly smaller query complexity.
	
	To that end, we modify the median trick proposed by Gily\'en et al. This trick works as follows. Suppose we have a quantum algorithm $\A$ that estimates the gradient of a function up to precision $\varepsilon$ w.r.t.\ the $\ell^{\infty}$-norm, with probability at least $P > \frac12$. Suppose we run this algorithm several times and obtain estimates in the gradient space as displayed in \autoref{fig:median_trick}.
	
	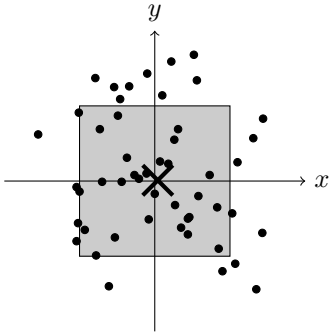
\begin{figure}[h!]
		\centering
		\begin{tikzpicture}[sample/.style={draw,circle,fill,inner sep = .1em}]
			\draw[fill=black!20] (1,1) -- (1,-1) -- (-1,-1) -- (-1,1) -- cycle;
			\node[sample] at (0.56,1.34) {};
			\node[sample] at (0.22,1.59) {};
			\node[sample] at (1.44,0.83) {};
			\node[sample] at (-0.93,-0.65) {};
			\node[sample] at (0.44,-0.5) {};
			\node[sample] at (-0.08,-0.51) {};
			\node[sample] at (-0.7,-0.01) {};
			\node[sample] at (0.73,0.08) {};
			\node[sample] at (1.31,0.57) {};
			\node[sample] at (-0.0,-0.17) {};
			\node[sample] at (-0.53,-0.75) {};
			\node[sample] at (-1.55,0.62) {};
			\node[sample] at (-1.04,-0.08) {};
			\node[sample] at (0.1,1.14) {};
			\node[sample] at (1.43,-0.69) {};
			\node[sample] at (-0.73,0.69) {};
			\node[sample] at (-0.46,1.09) {};
			\node[sample] at (-0.78,-0.99) {};
			\node[sample] at (-1.02,-0.56) {};
			\node[sample] at (-0.49,0.87) {};
			\node[sample] at (0.44,-0.71) {};
			\node[sample] at (1.35,-1.44) {};
			\node[sample] at (1.1,0.25) {};
			\node[sample] at (0.83,-0.35) {};
			\node[sample] at (-0.54,1.25) {};
			\node[sample] at (-0.1,1.43) {};
			\node[sample] at (-0.11,0.1) {};
			\node[sample] at (0.58,-0.2) {};
			\node[sample] at (0.07,0.26) {};
			\node[sample] at (0.26,0.55) {};
			\node[sample] at (-0.61,-1.4) {};
			\node[sample] at (0.31,0.69) {};
			\node[sample] at (0.27,-0.32) {};
			\node[sample] at (-0.27,0.08) {};
			\node[sample] at (-0.37,0.31) {};
			\node[sample] at (0.52,1.68) {};
			\node[sample] at (1.07,-1.1) {};
			\node[sample] at (0.9,-1.2) {};
			\node[sample] at (-1.0,-0.14) {};
			\node[sample] at (-0.44,-0.01) {};
			\node[sample] at (0.85,-0.9) {};
			\node[sample] at (0.18,0.23) {};
			\node[sample] at (0.46,-0.48) {};
			\node[sample] at (-1.04,-0.8) {};
			\node[sample] at (1.03,-0.43) {};
			\node[sample] at (0.35,-0.62) {};
			\node[sample] at (-0.21,0.03) {};
			\node[sample] at (-0.34,1.26) {};
			\node[sample] at (-1.01,0.91) {};
			\node[sample] at (-0.79,1.37) {};
			\node (med) at (0.04, 0.01) {};
			\draw[ultra thick] (med) ++(-.2,.2) -- ++(.4,-.4);
			\draw[ultra thick] (med) ++(-.2,-.2) -- ++(.4,.4);
			\draw[->] (-2,0) -- (2,0) node[right] {$x$};
			\draw[->] (0,-2) -- (0,2) node[above] {$y$};
		\end{tikzpicture}
		\caption{$50$ samples of a probability distribution that hits the gray square with probability $2/3$. The center of the $\times$-symbol is the coordinate-wise median of all these samples. The probability of this symbol not being in the gray square decreases exponentially with the number of samples.}
		\label{fig:median_trick}
	\end{figure}

	One can now build a quantum algorithm $\B$ with a much higher success probability than $\A$, by taking the coordinate-wise median of the results from $\A$. The resulting vector in $\R^d$ can be shown to be $\varepsilon$-close to the true gradient with probability at least $1 - de^{-2N(P-1/2)^2}$, where $N$ is the number of samples. Choosing $N$ logarithmic in $d/\delta$ is sufficient to obtain a success probability of $1 - \delta$, and hence we can use the median trick to boost the success probability of a gradient estimation algorithm to arbitrary height without incurring more than a logarithmic overhead on the query complexity of the algorithm.
	
	The median trick as presented above only works when estimating the gradient in the $\ell^{\infty}$-norm, though. Consider the samples in \autoref{fig:ell1_median_trick}. Suppose that they were obtained by some quantum algorithm $\A$ that estimates the gradient to precision $\varepsilon$ in the $\ell^1$-norm. In this example, all points except one lie in the successful region around $\nabla f(\vec{0})$, i.e., all samples but one are within $\ell^1$-distance $\varepsilon$ of the gradient that is to be estimated. However, if we take the coordinate-wise median of these samples, the resulting vector is located outside the $\ell^1$-ball with radius $\varepsilon$ around the true gradient. Hence, the median trick can make matters worse when trying to estimate the gradient accurately w.r.t.\ $\ell^1$-distance.
	
	\begin{figure}[h!]
		\centering
		\begin{tikzpicture}[sample/.style={draw,circle,fill,inner sep=.1em}, scale=1.5]
			\draw[fill=black!20] (1,0) -- (0,1) -- (-1,0) -- (0,-1) -- cycle;
			\node[sample] at (0.77,0.09) {};
			\node[sample] at (0.9,-0.05) {};
			\node[sample] at (0.85,-0.05) {};
			\node[sample] at (0.81,-0.03) {};
			\node[sample] at (0.7,-0.01) {};	
			\node[sample] at (-0.02,0.74) {};
			\node[sample] at (0.06,0.8) {};
			\node[sample] at (0.02,0.83) {};
			\node[sample] at (-0.05,0.77) {};
			\node[sample] at (0.05,0.93) {};
			\node[sample, inner sep = .2em] at (.6,.6) {};
			\node (med) at (.6,.6) {};
			\draw[ultra thick] (med) ++(-.2,.2) -- ++(.4,-.4);
			\draw[ultra thick] (med) ++(-.2,-.2) -- ++(.4,.4);
			\draw[dashed] (-1.5,.65) -- (1.5,.65);
			\draw[dashed] (.65,-1.5) -- (.65,1.5);
			\draw[->] (-1.5,0) -- (1.5,0) node[right] {$x$};
			\draw[->] (0,-1.5) -- (0,1.5) node[above] {$y$};
		\end{tikzpicture}
		\caption{The gray diamond denotes the unit ball w.r.t.\ the $\ell^1$-norm. Out of $11$ samples that are drawn, $5$ are concentrated at the top of the figure, $5$ are concentrated at the right, and one is at the cross. is the coordinate-wise mean of all the samples. Even though all but one of the samples are within the $\ell^1$-unit ball, the coordinate-wise mean is not. Hence, taking the coordinate-wise mean may make matters worse when approximating vectors w.r.t.\ the $\ell^1$-distance.}
		\label{fig:ell1_median_trick}
	\end{figure}
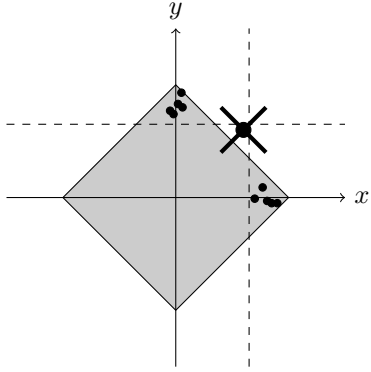

	Generally speaking, for any $n \in \N$, if the success probability is merely lower bounded by $\frac12 + \frac{1}{2n}$, then the median trick can yield a result that is off by $n\varepsilon$. To circumvent the additional error introduced by the median trick, we substitute it altogether, and introduce a new idea to boost the success probability of any $\ell^1$-approximate gradient estimation algorithm that succeeds with a probability strictly larger than $\frac12$.
	
	Consider the following setting. We draw balls of radius $\varepsilon$ around all samples obtained from running the $\varepsilon$-precise $\ell^1$-approximate gradient estimation algorithm $\A$ that worked with success probability $P > \frac12$. See \autoref{fig:ell1_balls_trick}.
	
	\begin{figure}[h!]
		\centering
		\begin{tikzpicture}[sample/.style={draw,fill,circle,inner sep=.1em}]
			\draw[fill=black!20] (1,0) -- (0,1) -- (-1,0) -- (0,-1) -- cycle;
			\draw[dotted] (2,0) -- (0,2) -- (-2,0) -- (0,-2) -- cycle;
			\draw[->] (-2.5,0) -- (2.5,0) node[right] {$x$};
			\draw[->] (0,-2.5) -- (0,2.5) node[above] {$y$};
			\node[sample, inner sep = .2em] at (0,0) {};
			\foreach \x in {(.2,.3), (.8,.1), (.1,.65),(-1.3,-.6),(1.2,-1.2)}
			{
				\node[sample] at \x {};
				\draw[dashed] \x ++(1,0) -- ++(-1,1) -- ++(-1,-1) -- ++(1,-1) -- cycle;
			}
			\draw[pattern=horizontal lines] (.6,.9) -- (0.98,0.52) -- (0.18,-0.28) -- (-.2,.1) -- cycle;
		\end{tikzpicture}
		\caption{We want to approximate the big dot at the origin accurately w.r.t.\ the $\ell^1$-norm. We have an algorithm that terminates successfully if it outputs an estimate that is in the gray region. Suppose we obtained $5$ samples from this algorithm, represented by the smaller black dots. We draw around these $5$ samples shapes with dashed borders identical to the gray region. The hatched region is the intersection of at least half of these and we return an arbitrary point in this region. The main constituents of the proof for boosting the success probability is showing that, with high probability, this hatched region is non-empty, it contains the origin, and is contained within the dotted region, which represented the $\ell^1$-ball around the origin with twice the radius as the gray region.}
		\label{fig:ell1_balls_trick}
	\end{figure}
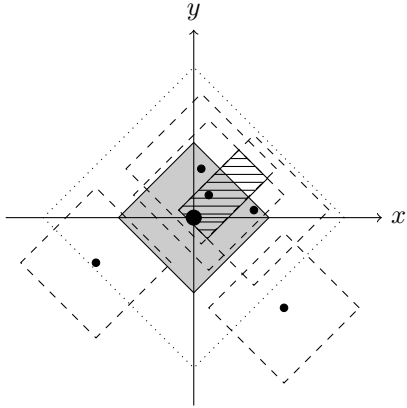

	As long as we have enough samples, we can expect that more than half of them are within $\ell^1$-distance $\varepsilon$ from the true gradient. Hence, with high probability, the gradient is contained in the intersection of at least half of the $\ell^1$-balls with radius $\varepsilon$ that we drew around the samples. This region is the hatched region in \autoref{fig:ell1_balls_trick}. Moreover, this intersection cannot have a diameter that is bigger than $2\varepsilon$, and so if we return any point in this region, we will return a $2\varepsilon$-approximation of the true gradient.
	
	Similar to the median trick in the $\ell^{\infty}$-setting, we find that for the above method to work with success probability $1 - \delta$, we need a number of samples that scales logarithmically in $d/\delta$. So, we have devised a way to boost the success probability of $\ell^1$-approximate gradient estimation algorithms, showing that there cannot exist a quantum algorithm that solves the $\ell^1$-approximate gradient estimation problem significantly faster once we lower the success probability slightly from $17/18$. This argument is presented in full detail in \autoref{thm:general_lower_bound_proof}.
	
	This concludes our high level description of the lower bound. All the details can be found in \autoref{sec:lower_bound}.
	
	\subsection{Applications}
	
	In this subsection, we briefly elaborate on the applications that we envision for the algorithm developed in this paper. As gradient estimation algorithms are frequently used subroutines in classical computations, we expect that their quantum counterpart will find many applications as well, hence we do not expect that the list of applications presented in this subsection is exhaustive.
	
	\subsubsection{Speeding up classical gradient descent methods}
	
	We first explain how our results can be used if we know how to classically evaluate the objective function $f : \R^d \to \R$. As this function can be evaluated using a classical circuit, there also exists a quantum circuit, $B_f$, that acts in the following manner,
	\[B_f : \ket{\vec{x}}\ket{0} \mapsto \ket{\vec{x}}\ket{f(\vec{x})},\]
	i.e., the function value of $f$ at $\vec{x}$ is evaluated and returned in a binary representation in the last register. Using the \textit{phase kickback trick}, one can now construct a phase oracle $O_f$, which acts as follows,
	\[O_f : \ket{\vec{x}} \mapsto e^{if(\vec{x})}\ket{\vec{x}},\]
	using just one call to the binary circuit $B_f$. Hence, if one has access to a classical circuit that evaluates $f$, then one can use our results to perform gradient estimation on a quantum computer. If in addition, $f$ satisfies the Gevrey smoothness condition for $\sigma < 1$, then one readily obtains a speed-up over classical gradient estimation routines.
	
	This can for instance be useful in algorithms that use gradient estimation as a subroutine. The simplest example of such an algorithm is gradient descent, which attempts to find the minimum of the objective function by updating the guess in every iteration in the direction opposite to the gradient. Hence, in every iteration the gradient has to be estimated, so under appropriate smoothness conditions the quantum gradient estimation algorithm as presented in this paper can speed up every iteration individually.
	
	\subsubsection{Optimizing the success probability of a variational quantum circuit}
	
	Our results can also be used in settings that are more inherently quantum, for instance when we are using variational quantum circuits. Such a circuit consists of a fixed number of gates, but the action of some of them is influenced by a global parameter vector $\theta \in \R^d$. In \autoref{fig:variational_circuit}, we present an example of a variational circuit where $d = 3$.
	
	\begin{figure}[h!]
		\centering
		\input{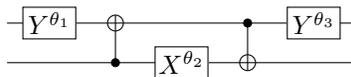}
		\caption{An example of a variational circuit acting on two qubits where the parameter vector $\theta$ is three-dimensional.}
		\label{fig:variational_circuit}
	\end{figure}
	
	Suppose that one has a variational quantum circuit $C(\theta)$, where $\theta \in \R^d$, which acts on $n \in \N$ qubits as
	\begin{equation}
		\label{eq:variational_circuit}
		C(\theta) : \ket{0}^{\otimes n} \mapsto \sqrt{p_1(\theta)}\ket{\psi(\theta)}\ket{1} + \sqrt{1 - p_1(\theta)}\ket{\phi(\theta)}\ket{0},
	\end{equation}
	where $\ket{\psi(\theta)}$ and $\ket{\phi(\theta)}$ are unknown $(n-1)$-qubit states. One can think of the final qubit as an indicator whether the circuit succeeded or failed in performing a task. In this setting, one wants to maximize the value of $p_1(\theta)$. Given controlled access to the following circuit,
	\[C : \ket{\theta}\ket{0}^{\otimes n} \mapsto \ket{\theta}\left(C(\theta)\ket{0}^{\otimes n}\right),\]
	Gily\'en et al.\ have shown, in Theorem 14 in \cite{Gilyen18}, that one can construct the phase oracle
	\[O_{p_1} : \ket{\vec{\theta}} \mapsto e^{ip_1(\theta)}\ket{\vec{\theta}}\]
	using a number of controlled calls to $C$ that is just logarithmic in the precision of the phase oracle. Using the quantum gradient estimation algorithm we can now perform gradient ascent to find the value $\theta \in \R^d$ that maximizes the success probability of the circuit $C(\theta)$.
	
	Gily\'en et al.\ already explained in Section 7 of \cite{Gilyen18} how this construction can be used to speed up variational quantum eigensolvers as proposed in \cite{Kandala17}, QAOA \cite{Farhi14}, and quantum auto-encoders \cite{Romero17}. We want to add that this particular construction might also yield speed-ups in the field of quantum reinforcement learning. In Chapter 6 of \cite{Cornelissen18}, we develop a quantum circuit $C(\theta)$ as in \autoref{eq:variational_circuit}, where $p_1(\theta)$ is approximately proportional to the value function, and $\theta \in \R^d$ are the parameters of the policy. Whether this construction yields any speed-ups depends on the smoothness of the value function, which is something we plan on investigating in the future.
	
	\subsubsection{Analog computation}
	
	Finally, we want to stress that the methods that were used to develop this gradient estimation algorithm are also inherently interesting. To illustrate, suppose that for some $f : \R \to \R$, $O_f$ is a quantum oracle that acts in the following manner,
	\[O_f : \ket{x} \mapsto e^{if(x)}\ket{x}.\]
	We easily observe that composing two such oracles allows for creating an oracle that adds the two functions,
	\[O_{f_1 + f_2} = O_{f_1}O_{f_2} : \ket{x} \mapsto e^{i(f_1(x) + f_2(x))}\ket{x}.\]
	Similarly, suppose that for some function $f : \R \to [0,1]$, we have the following oracle,
	\[U_f : \ket{x}\ket{0} \mapsto \ket{x}\left(\sqrt{f(x)}\ket{0} + \sqrt{1-f(x)}\ket{1}\right)\]
	Then, by composing two such oracles, where it is understood that both act on the index register, but the $U_{f_1}$ acts on the first auxiliary qubit and $U_{f_2}$ on the second, we construct a product oracle:
	\[U_{f_1f_2} = U_{f_1}U_{f_2} : \ket{x}\ket{00} \mapsto \ket{x}\left(\sqrt{f_1(x)f_2(x)}\ket{00} + \ket{\text{states orthogonal to 00}}\right)\]
	Moreover, Gily\'en et al.\ \cite{Gilyen18} have shown how one can interconvert between these types of oracles, using only logarithmic overhead in the precision. Hence, if a function $f : \R \to \R$ can be calculated using just addition and multiplication, then one can construct an oracle evaluating $f$ without ever performing any classical arithmetic circuit. As in the process we never store any digital representation of the function values, we refer to this type of computation as \textit{analog computation}.
	
	Within the construction of the quantum gradient estimation algorithm, we use this technique of analog computation to efficiently build up the oracle that evaluates the smoothings of the objective function, e.g., in \autoref{thm:smoothing_phase_oracle}. We also use it in Chapter 6 in \cite{Cornelissen18} to build up the oracle that evaluates the value function in reinforcement learning. Using block-encodings, one can also perform this type of analog computations in a more general linear algebra setting, as discovered by Gily\'en et al.\ in \cite{Gilyen18b}. We expect that these techniques will find many more applications.
	
	\subsection{Paper outline}
	
	This paper contains two main results: an improved quantum gradient estimation algorithm, described in \autoref{sec:algorithm}, and an improved lower bound on the query complexity of the gradient estimation problem, presented in \autoref{sec:lower_bound}. We organize the paper so that the reader can read either of these results without having to understand the other.
	
	More specifically, after the introduction we proceed with the preliminaries in \autoref{sec:preliminaries}. \autoref{sec:algorithm} is devoted to the new quantum gradient estimation algorithm, whereas \autoref{sec:lower_bound} is dedicated to proving the new lower bound on the quantum gradient estimation problem. These two sections have no cross-references, so they can both be read independently. After that, there is a generic \autoref{sec:conclusion}, which elaborates on the current state of research, and lists some interesting topics for further research. Finally, \autoref{sec:misc_results} lists some results that are used in the proofs of this text, but are not explicitly proven here. Pointers are provided to where these results are proven.
	
	\section{Preliminaries}
	\label{sec:preliminaries}
	
	In this section, we provide rigorous definitions of the mathematical objects used in subsequent sections, and we elaborate on some of their properties. Specifically, in \autoref{subsec:notation} we elaborate on the notation, in \autoref{subsec:gevrey_functions} we formally introduce the smoothness condition imposed on our functions, in \autoref{subsec:phase_oracles} we formally introduce the input oracle model, and in \autoref{subsec:formal_problem_statement} we provide a formal statement of the problem.
	
	\subsection{Notation}
	\label{subsec:notation}
	
	In this subsection, we formally introduce the terminology that we use throughout the remainder of this text. We use the convention that $\N$ contains all \textit{positive} integers and $\N_0$ contains all \textit{non-negative} integers. For all $d \in \N$, we define $[d] = \{1, 2, \dots, d\}$. We write elements from $\R^d$, where $d \in \N$, in boldface, e.g., $\vec{x} \in \R^d$. When numbers appear in bold, like $\vec{0}$ or $\vec{1}$, we denote the vector in $\R^d$ with all entries equal to this number.
	
	Let $d \in \N$ and $f : \R^d \to \R$ be smooth, i.e., all (higher order) partial derivatives of $f$ exist. For all $j \in [d]$ and $\vec{x} \in \R^d$, we let $\partial_jf(\vec{x})$ denote the partial derivative of $f$ with respect to the $j$th coordinate, evaluated at $\vec{x}$. Let $k \in \N$ and $\alpha = (\alpha_1, \dots, \alpha_k) \in [d]^k$. We define $\partial_{\alpha}f(\vec{x}) = \partial_{\alpha_1} \cdots \partial_{\alpha_k}f(\vec{x})$. In particular, note that $\partial_jf(\vec{x}) = \partial_{(j)}f(\vec{x})$. Furthermore, we let $\vec{x}^{\alpha} = x_{\alpha_1}x_{\alpha_2} \cdots x_{\alpha_k}$, i.e., the product of the entries of $\vec{x}$ specified by the multi-index $\alpha$.
	
	Let $n \in \N$. Throughout this text, we identify $n$-qubit states with unit vectors in $\C^{2^n}$ and denote them with a ket-symbol $\ket{\cdot}$. We use this symbol for unit vectors only.
	
	Finally, unless stated otherwise, we use $\log$ when referring to $\log_2$, and $\log_e$ is denoted by $\ln$.
	
	\subsection{Gevrey functions}
	\label{subsec:gevrey_functions}
	
	In this subsection, we formally introduce the class of functions that satisfy the smoothness condition that we impose. A similar class of functions was first considered by Gevrey \cite{Gevrey18}, so we refer to it as the Gevrey class.
	
	\begin{definition}{Gevrey functions}
		\label{def:gevrey_functions}
		Let $d \in \N$, $\sigma \in \R$, $c > 0$, $\Omega \subseteq \R^d$ open and $f : \R^d \to \R$. We say that $f$ is a \textit{Gevrey function on $\Omega$ with parameters $c$ and $\sigma$} if $f$ is smooth, i.e., all its (higher order) partial derivatives exist, and the following upper bound on its partial derivatives is satisfied for all $\vec{x} \in \Omega$, $k \in \N_0$ and $\alpha \in [d]^k$:
		\[|\partial_{\alpha}f(\vec{x})| \leq \frac12c^k(k!)^{\sigma}.\]
		The collection of all Gevrey functions $f : \R^d \to \R$ on $\Omega$ with parameters $c$ and $\sigma$ is referred to as the \textit{Gevrey class with parameters $c$ and $\sigma$}, and is denoted by $\G_{d,c,\sigma,\Omega}$.
	\end{definition}

	There are a few properties of Gevrey functions that are immediately clear. We list a few of them in the theorem below.
	
	\begin{theorem}{Properties of Gevrey functions}
		\label{thm:gevrey_functions_properties}
		Let $d \in \N$. $\sigma \in \R$, $c > 0$ and $\Omega \subseteq \R^d$ open. The following properties of Gevrey functions hold:
		\begin{enumerate}
			\item Whenever $0 < c_1 < c_2$, we have $\G_{d,c_1,\sigma,\Omega} \subseteq \G_{d,c_2,\sigma,\Omega}$.
			\item Whenever $\sigma_1 < \sigma_2$, we have $\G_{d,c,\sigma_1,\Omega} \subseteq \G_{d,c,\sigma_2,\Omega}$.
			\item Whenever $\Omega_1 \subseteq \Omega_2 \subseteq \R^d$ are both open, we have:
			\[f \in \G_{d,c,\sigma,\Omega_2} \qquad \Rightarrow \qquad f|_{\Omega_1} \in \G_{d,c,\sigma,\Omega_1}\]
		\end{enumerate}
	\end{theorem}

	\begin{proof}
		We simply cover every property individually.
		\begin{enumerate}
			\setlength\itemsep{-.5em}
			\item Let $f \in \G_{d,c_1,\sigma,\Omega}$, $\vec{x} \in \Omega$, $k \in \N_0$ and $\alpha \in [d]^k$ arbitrarily. We find $|\partial_{\alpha}f(\vec{x})| \leq \frac12c_1^k(k!)^{\sigma} \leq \frac12c_2^k(k!)^{\sigma}$, and so $f \in \G_{d,c_2,\sigma,\Omega}$. As this holds for any $f \in \G_{d,c_1,\sigma,\Omega}$, we find $\G_{d,c_1,\sigma,\Omega} \subseteq \G_{d,c_2,\sigma,\Omega}$.
			\item Let $f \in \G_{d,c,\sigma_1,\Omega}$, $\vec{x} \in \Omega$, $k \in \N_0$ and $\alpha \in [d]^k$ arbitrarily. We find $|\partial_{\alpha}f(\vec{x})| \leq \frac12c^k(k!)^{\sigma_1} \leq \frac12c^k(k!)^{\sigma_2}$, and so $f \in \G_{d,c,\sigma_2,\Omega}$. As this holds for any $f \in \G_{d,c,\sigma_1,\Omega}$, we find $\G_{d,c,\sigma_1,\Omega} \subseteq \G_{d,c,\sigma_2,\Omega}$.
			\item Let $f \in \G_{d,c,\sigma,\Omega_2}$, $\vec{x} \in \Omega_1$, $k \in \N_0$ and $\alpha \in [d]^k$ arbitrarily. We find: $|\partial_{\alpha}(f|_{\Omega_1})(\vec{x})| = |\partial_{\alpha}f(\vec{x})| \leq \frac12c^k(k!)^{\sigma}$ and so $f|_{\Omega_1} \in \G_{d,c,\sigma,\Omega_1}$. As this holds for any $f \in \G_{d,c,\sigma,\Omega_2}$, the statement follows.
		\end{enumerate}
		This completes the proof.
	\end{proof}

	It is not easy to obtain a very good understanding of how the different Gevrey classes relate. For instance, it is not clear to the author whether $\G_{1,1,1,\R} \setminus \G_{1,1,0,\R}$ is non-empty. Investigating this would be an interesting topic of further research. The best we can do at this point to help the reader develop some intuition for some concrete the Gevrey classes is to list some functions, and indicate to which classes they belong. Some example functions can be found in \autoref{tbl:Gevrey_functions}.
	
	\begin{table}[h!]
		\centering
		\begin{tabular}{l|c|c|c|c|c}
			Function & $\G_{1,\cdot,0,(-1,1)}$ & $\G_{1,\cdot,0,(-1,\infty)}$ & $\G_{1,\cdot,0,\R}$ & $\G_{1,\cdot,\frac12,\R}$ & $\G_{1,\cdot,1,\R}$ \\\hline
			$x \mapsto \frac12\sin(cx)$ & $c$ & $c$ & $c$ & $c$ & $c$ \\
			$x \mapsto \frac12\cos(cx)$ & $c$ & $c$ & $c$ & $c$ & $c$ \\
			$x \mapsto \frac12\exp(-c(x+1))$ & $c$ & $c$ & & & \\\hline
			$x \mapsto \frac12\exp(-\frac12(cx)^2)$ & & & & $c^*$ & $c^*$ \\\hline
			$x \mapsto \frac{1}{2(1+(cx)^2)}$ & & & & & $c^*$ \\
			$x \mapsto \frac12\arctan(cx)$ & & & & & $c^*$ \\
			$x \mapsto \frac{1}{2(1 + e^{-cx})}$ & & & & & $c^*$
		\end{tabular}
		\parbox{.85\textwidth}{\small * As these values are mainly included for illustrative purposes, we have not tried hard to find formal proofs to justify these entries. However, they are strongly supported by numerical evidence.}
		\caption{Table of Gevrey functions. The entries in the table indicate the minimal value that one can plug into the second argument in the subscript of $\G$, given the other arguments as displayed in the top row, such that the function in the leftmost column is a member of the specified Gevrey class. An empty cell means that no such value exists.}
		\label{tbl:Gevrey_functions}
	\end{table}

	Gevrey classes are in general not closed under composition, but some interesting results about composite Gevrey functions can be obtained whenever $\sigma \geq 1$. The interested reader is referred to \cite{Cornelissen18}, Theorem 6.4.2.
	
	\subsection{(Fractional) phase oracles}
	\label{subsec:phase_oracles}
	
	In this subsection, we introduce the input model to the gradient estimation problem, that is, we define how we assume to have access to the objective function $f$ whose gradient we want to estimate. We first introduce the concept of phase oracles in the definition below.
	
	\begin{definition}{Phase oracle}
		\label{def:phase_oracle}
		Let $d,n \in \N$, $f : \R^d \to \R$, $G \subseteq \R^d$, with $|G| \leq 2^{nd}$, and let $\{\ket{\vec{x}} : \vec{x} \in G\}$ be an orthonormal set of $nd$-qubit states. A \textit{phase oracle evaluating $f$ on $G$}, denoted by $O_{f,G}$ is an operator on the $nd$-qubit state space that acts as follows:
		\[O_{f,G} : \ket{\vec{x}} \mapsto e^{if(\vec{x})}\ket{\vec{x}}.\]
		For simplicity, we assume that it acts as the identity operator on the orthogonal complement of the space spanned by the states $\ket{\vec{x}}$ where $\vec{x} \in G$. Moreover, we introduce shorthand notation for the \textit{controlled phase oracle}, which performs the following action on one more qubit which we refer to as the \textit{control qubit}:
		\[C(O_{f,G}) = \ket{0}\bra{0} \otimes I_{2^n} + \ket{1}\bra{1} \otimes O_{f,G}.\]
	\end{definition}

	Gily\'en et al.~\cite{Gilyen18} motivated the use of this model in their paper. They also considered different input models, namely the \textit{probability oracle}, which appears naturally in the context of quantum variational circuits \cite{Kandala17}, and the \textit{binary oracle}, which arises naturally when emulating function evaluations performed with a classical circuit. Both of these input models can be efficiently converted to the phase oracle model, i.e., phase oracles can be constructed from probability and binary oracles with at most polylogarithmic overhead. Moreover, Gily\'en et al.\ showed that phase oracles can be efficiently converted to probability oracles as well. The details can be found in \cite{Gilyen18}, Chapter 4 and Appendix B, or in \cite{Cornelissen18}, Circuit 6.2.5 and Circuit 6.4.10.
	
	It is interesting to note that the phase oracle is in a sense \textit{analog}, meaning that it does not require a binary representation of the function value. As a consequence, one cannot recover one function value using just one call to the phase oracle, like one could classically. In that sense, the input model is weaker than the classical model, which makes it more surprising that speed-ups over the classical setting can be attained.

	The fact that this input model is analog inherently complicates performing calculations. For instance, suppose we want to multiply our function value with a number $\xi$ that is between $-1$ and $1$. In the digital setting, one could first obtain the function value as a bitstring, and subsequently implement the multiplication by $\xi$ by manipulating it. In the analog model, though, it is not at all clear how one would achieve this, but it is clear that different methods are required. Recently, some very promising techniques were developed by Gily\'en et al.\ \cite{Gilyen18b}. We use and further elaborate on some of these results in \autoref{subsec:fractional_phase_oracles} and \autoref{subsec:central_difference_schemes}.
	
	\subsection{Formal problem statement}
	\label{subsec:formal_problem_statement}
	
	Now that we have established notational conventions and considered the smoothness conditions and input model, we have covered all necessary prerequisites to formally introduce the problem. First, we formally define what a quantum gradient estimation algorithm is.
	
	\begin{definition}{Quantum gradient estimation algorithms}
		\label{def:quantum_gradient_estimation_algorithms}
		Let $d \in \N$, $\varepsilon > 0$, $p \in [1,\infty]$, $P \in [0,1]$, $G \subseteq \Omega \subseteq \R^d$, with $\Omega$ open, and $\F$ a class of smooth functions mapping $\Omega$ into $\R$. Let $\A$ be a quantum algorithm that can access a function $f \in \F$ through queries to the controlled phase oracle $C(O_{f,G})$, and outputs a vector in $\R^d$. We define the following objects:
		\begin{enumerate}
			\item $\A(f)$ is the random variable with values in $\R^d$ describing the outcome of the algorithm.
			\item $T_{\A}(f)$ is the random variable denoting the total number of queries to the controlled phase oracle that the algorithm performs when it is supplied with an oracle for $f \in \F$.
			\item $T_{\A}$ is the infimum of all $M \geq 0$ such that almost surely, i.e., with probability $1$:
			\[\max_{f \in \F} T_{\A}(f) \leq M.\]
			The value $T_{\A}$ is referred to as the \textit{worst-case query complexity of $\A$ to the phase oracle}.
			\item $P_{\A}(f)$ is the probability that the algorithm, upon receiving input $f \in \F$, outputs a vector $\vec{g} \in \R^d$ that satisfies
			\[\norm{\vec{g} - \nabla f(\vec{0})}_p \leq \varepsilon.\]
		\end{enumerate}
		We say that $\A$ is an \textit{$\varepsilon$-precise $\ell^p$-approximate quantum gradient estimation algorithm for $\F$ on $G$ with success probability lower bounded by $P$} if, for every input function $f \in \F$, we have $P_{\A}(f) \geq P$.
	\end{definition}

	A few remarks about the above definition are in place. First, we will allow \textit{quantum algorithms} to have classical postprocessing steps. This means that after the outcome of a quantum measurement is obtained, the algorithm can modify the resulting bit string before returning the output.
	
	Secondly, note that we defined $T_{\A}(f)$ to be a random variable. The reason for this is that the number of oracle calls might not be constant across several runs of the algorithm. For instance, one can perform some intermediate measurement and based on the outcome of this measurement decide whether the algorithm should terminate or not. As the measurement outcomes can be described by random variables, defining $T_{\A}(f)$ to be a random variable as well comes about naturally.
	
	Finally, note that the success probability $P_{\A}(f)$ need not be constant across different inputs. This is why we require the function $P_{\A} : \F \to [0,1]$ to be lower bounded by $P$ globally. Hence, we want the \textit{worst-case success probability} to be at least $P$. We have not looked into quantum gradient estimation algorithms where the \textit{average-case success probability} is lower bounded by some constant, but this would be an interesting topic of further research.
	
	Now, we are ready to introduce the formal description of the gradient estimation problem.
	
	\begin{definition}{Quantum gradient estimation of Gevrey functions problem}
		\label{def:quantum_gradient_estimation_problem}
		Let $d \in \N$, $\varepsilon, c > 0$, $\sigma \in \R$, $p \in [1,\infty]$, $P \in [0,1]$, $G \subseteq \Omega \subseteq \R^d$, with $\Omega$ open. The \textit{$\varepsilon$-precise $\ell^p$-approximate quantum gradient estimation problem of $\G_{d,c,\sigma,\Omega}$ on $G$ with probability lower bounded by $P$} is the following question: \textit{What is the $\varepsilon$-precise $\ell^p$-approximate quantum gradient estimation algorithm $\A$ for $\G_{d,c,\sigma,\Omega}$ on $G$ with success probability lower bounded by $P$ that minimizes $T_{\A}$?} We refer to the optimal value of $T_{\A}$ as the \textit{query complexity of this problem}.
	\end{definition}
	
	Just like we only considered the worst-case success probability, we also solely consider the worst-case query complexity. It would be interesting to investigate if one can achieve fundamentally different results by minimizing the \textit{average-case} or \textit{expected query complexity} instead, so this would be an interesting topic of further research too.
	
	This completes our discussion of the preliminaries. The reader can now proceed with \autoref{sec:algorithm}, where explicit constructions of quantum gradient estimation algorithms are considered. Alternatively, the reader can skip ahead to \autoref{sec:lower_bound}, where lower bounds on the worst-case query complexity of quantum gradient estimation algorithms are proved.
	
	\section{Generalization of Gily\'en et al.'s quantum gradient estimation algorithm}
	\label{sec:algorithm}
	
	In this chapter, we describe the quantum gradient estimation algorithm that we constructed. In \autoref{subsec:fractional_phase_oracles}, we start by revising some of the techniques developed by Gily\'en et al.\ in \cite{Gilyen18b}. Afterwards, in \autoref{subsec:central_difference_schemes}, we introduce the numerical methods that we employ in the quantum gradient estimation algorithm, and prove some of their properties. Then, we state the algorithm in \autoref{subsec:algorithm}. There are two main things that need to be proved. First, we prove in \autoref{subsec:query_complexity} that the query complexity is as claimed. Subsequently,  in \autoref{subsec:success_probability}, we prove that the success probability of the algorithm is lower bounded by $2/3$.
	
	\subsection{Fractional phase oracles}
	\label{subsec:fractional_phase_oracles}
	
	In \autoref{subsec:phase_oracles}, we briefly mentioned that performing calculations with analog oracles requires fundamentally different methods than the ones employed to perform digital calculations. In this section we elaborate on one such method, the fractional phase oracle.
	
	\begin{definition}{Fractional phase oracle}
		\label{def:fractional_phase_oracle}
		Let $d,n \in \N$, $f : \R^d \to \R$, $G \subseteq \R^d$, with $|G| \leq 2^{nd}$, $\{\ket{\vec{x}} : \vec{x} \in G\}$ an orthonormal set of $nd$-qubit states, and $-1 < \xi < 1$. Then the \textit{fractional phase oracle evaluating $f$ on $G$ with power $\xi$}, denoted by $O_{f,G}^{\xi}$, is an operator acting on the $n$-qubit state space in the following manner:
		\[O_{f,G}^{\xi} : \ket{\vec{x}} \mapsto e^{i\xi f(\vec{x})}\ket{\vec{x}}.\]
		Again, for simplicity, we assume that this operator acts as the identity operator on the orthogonal complement of the subspace spanned by the vectors $\ket{\vec{x}}$ where $\vec{x} \in G$.
	\end{definition}
	
	The immediate question that arises is whether one can implement a fractional phase oracle from \autoref{def:fractional_phase_oracle} using a few queries to the normal phase oracle, see \autoref{def:phase_oracle}. A construction that achieves this was first introduced by Gily\'en et al. \cite{Gilyen18}, using techniques from \cite{Gilyen18b}. In the theorem below, we restate their result for completeness. The precise construction can be found in \cite{Gilyen18b}, or alternatively in \cite{Cornelissen18}.
	
	\begin{theorem}{Fractional phase oracles (Gily\'en et al.'s construction)}
		\label{thm:fraction_phase_oracle}
		Let $d,n \in \N$, $f : \R^d \to \R$ such that $\norm{f}_{\infty} \leq \frac12$, $G \subseteq \R^d$, with $|G| \leq 2^{nd}$, and $\{\ket{\vec{x}} : \vec{x} \in G\}$ be an orthonormal set of $nd$-qubit states. Let $-1 < \xi < 1$ and $\delta > 0$. Then, one can construct a unitary operator $U$ acting on $nd$ qubits as well as $N' = \Theta(nd)$ ancillary qubits using $\widetilde{\O}\left(\log\left(\frac{1}{\delta}\right)\right)$ queries to a controlled version of $O_{f,G}$, such that
		\[\norm{\left(\bra{0}^{\otimes N'} \otimes I_{2^{nd}}\right)U\left(\ket{0}^{\otimes N'} \otimes I_{2^{nd}}\right) - O_{f,G}^{\xi}} \leq \delta\]
		Using Gily\'en et al.'s terminology as introduced in \cite{Gilyen18b}, Definition 43, $U$ is a $(1,N',\delta)$-block-encoding of $O_{f,G}^{\xi}$.
	\end{theorem}
	
	\begin{proof}
		The implementation of fractional phase oracles is shortly discussed in \cite{Gilyen18} at the beginning of Section~4.3. The theorem they are referring to is Corollary~72 in \cite{Gilyen18b}. Essentially the same statement is proven in Circuit~4.2.18 in \cite{Cornelissen18}, using similar but subtly different techniques compared to those used in \cite{Gilyen18b}.
	\end{proof}
	
	Consider the setting as described in \autoref{thm:fraction_phase_oracle}. If one prepares an auxiliary register of $N - nd$ qubits in the all-zeros state and applies $U$ to all $N$ qubits, then it approximately performs $O_{f,G}^{\xi}$ on the last $nd$ qubits, and simultaneously approximately returns the auxiliary qubits to the all-zeros state. This is different from the usual way in which auxiliary registers are used, because normally it is guaranteed that the auxiliary register is returned to the all-zeros state exactly. However, the operator $U$ will make sure that the cumulative amplitude in the subspace orthogonal to the all-zeros state of the auxiliary register is at most $\delta$, which is sufficient for our purposes.
	
	\subsection{Central difference schemes}
	\label{subsec:central_difference_schemes}
	
	In this subsection, we elaborate on the numerical methods that are employed in the quantum gradient estimation algorithm.
	
	The numerical methods employed in the algorithm described in \autoref{sec:algorithm} are commonly referred to as \textit{central difference schemes}. We first introduce the coefficients in \autoref{def:coefficients} and subsequently prove the properties we need in \autoref{thm:central_difference_scheme_properties}. A more comprehensive derivation of these coefficients may be achieved using Taylor series and explicit inverses of Vandermonde matrices, which under the hood use the techniques that we employ in \autoref{thm:central_difference_scheme_properties}. For a digression in this area, we refer the reader to standard texts on numerical differentiation.
	
	% Do not put anything in between the footnote and the box below.
	\addtocounter{footnote}{1}
	\begin{definition}{Central difference scheme\footnotemark[\thefootnote]}
		\label{def:coefficients}
		Let $m \in \N$. For all $\ell \in \{-m,-m+1, \dots, m-1,m\}$, let
		\[a_{\ell}^{(2m)} = \begin{cases}
		1, & \text{if } \ell = 0, \\
		\frac{(-1)^{\ell+1}(m!)^2}{\ell(m+\ell)!(m-\ell)!}, & \text{otherwise}.
		\end{cases}\]
	\end{definition}
	\footnotetext{\label{Gilyendiff}We digress slightly from the definition used in \cite{Gilyen18}. Whenever $\ell \neq 0$, the definitions agree, however when $\ell = 0$, Gily\'en et al.\ define $a_{\ell}^{(2m)}$ to be $0$, whereas we define it to be $1$. The discrepancy originates from the fact that Gily\'en et al. attempt to approximate the function $f(\vec{x}) = \nabla f(\vec{0}) \cdot \vec{x}$, whereas we attempt to approximate the function $f(\vec{x}) = f(\vec{0}) + \nabla f(\vec{0}) \cdot \vec{x}$.}
	% End of the box/footnote block
	
	The key elementary properties of these coefficients are stated in the following theorem.
	
	\begin{theorem}{Properties of the central difference scheme}
		\label{thm:central_difference_scheme_properties}
		Let $2 \leq m \in \N$. Then, for all $k \in \{0,1, \dots, 2m\}$,
		\[\sum_{\ell=-m}^m a_{\ell}^{(2m)} \ell^k = \begin{cases}
		1, & \text{if } k \in \{0,1\}, \\
		0, & \text{otherwise}.
		\end{cases}\]
		We also find that for all $\ell \in \{-m, \dots, m\} \setminus \{0\}$,
		\[\left|a_{\ell}^{(2m)}\right| < \frac{1}{|\ell|}.\]
		Furthermore, for all integer $k \geq 2m+1$,
		\[\left|\sum_{\ell=-m}^m a_{\ell}^{(2m)} \ell^k\right| \leq 2m^k.\]
	\end{theorem}
	
	\begin{proof}
		We begin by proving the first equality. Our proof is based on Lagrange's interpolation formula. Let $k \in \{0,1, \dots, 2m\}$ and observe that for all $x \in \R$,
		\[\sum_{\ell=-m}^m \ell^k \prod_{\underset{j \neq \ell}{j = -m}}^m \frac{x - j}{\ell - j} = x^k.\]
		Differentiating both sides with respect to $x$ yields:
		\[\sum_{\ell=-m}^m \ell^k \sum_{\underset{n \neq \ell}{n = -m}}^m \frac{1}{\ell - n} \prod_{\underset{j \not\in \{n,\ell\}}{j=-m}}^m \frac{x - j}{\ell - j} = kx^{k-1}.\]
		Plugging in $x = 0$ yields
		\[\sum_{\ell=-m}^m \ell^k \sum_{\underset{n \neq \ell}{n=-m}}^m \frac{1}{\ell - n} \prod_{\underset{j \not\in \{n,\ell\}}{j=-m}}^m \frac{-j}{\ell - j} = \delta_{k1},\]
		where $\delta$ is the Kronecker delta. We can rewrite the left hand side as follows:
		\[\sum_{\ell=-m}^m \ell^k \cdot \left(-\sum_{\underset{n \neq \ell}{n=-m}}^m \frac{1}{n - \ell} \prod_{\underset{j \not\in \{n,\ell\}}{j=-m}}^m \frac{j}{j - \ell}\right) = \sum_{\ell=-m}^m \ell^k \cdot \left(-\sum_{\underset{n \neq \ell}{n=-m}}^m \frac{\displaystyle \prod_{\underset{j \not\in \{n,\ell\}}{j=-m}}^m j}{\displaystyle \prod_{\underset{j \neq \ell}{j=-m}}^m (j - \ell)}\right) = \delta_{k1}.\]
		Now, we rewrite the expression in the parentheses. For all $\ell \in \{-m, \dots, m\} \setminus \{0\}$, all terms where $n \neq 0$ drop out as there is a factor of $0$ present, so we obtain:
		\[-\sum_{\underset{n \neq \ell}{n = -m}}^m \frac{\displaystyle \prod_{\underset{j \not\in \{n,\ell\}}{j=-m}}^m j}{\displaystyle \prod_{\underset{j \neq \ell}{j=-m}}^m (j - \ell)} = -\frac{\displaystyle \prod_{\underset{j \neq 0,\ell}{j=-m}}^m j}{\displaystyle \prod_{\underset{j \neq \ell}{j=-m}}^m (j-\ell)} = -\frac{(-1)^mm! \cdot m! \cdot \frac{1}{\ell}}{(-1)^{m+\ell}(m+\ell)! \cdot (m-\ell)!} = \frac{(-1)^{\ell+1}(m!)^2}{\ell(m+\ell)!(m-\ell)!} = a_{\ell}^{(2m)}.\]
		On the other hand, when $\ell = 0$, we obtain
		\[-\sum_{\underset{n \neq \ell}{n=-m}}^m \frac{\displaystyle \prod_{\underset{j \not\in \{n,\ell\}}{j=-m}}^m j}{\displaystyle \prod_{\underset{j \neq \ell}{j=-m}}^m (j - \ell)} = -\sum_{\underset{n \neq 0}{n=-m}}^m \frac{\displaystyle \prod_{\underset{j \not\in \{n,\ell\}}{j=-m}}^m j}{\displaystyle \prod_{\underset{j \neq \ell}{j=-m}}^m j} = \sum_{\underset{n \neq 0}{n=-m}}^m \frac1n = 0.\]
		Hence, we find that
		\[\sum_{\ell=-m}^m a_{\ell}^{(2m)} \ell^k = \delta_{k0} + \sum_{\underset{\ell \neq 0}{\ell=-m}}^m a_{\ell}^{(2m)}\ell^k = \delta_{k0} + \delta_{k1} = \begin{cases}
		1, & \text{if } k \in \{0,1\}, \\
		0, & \text{otherwise}.
		\end{cases}\]
		This completes the proof of the first equality. For the first inequality, let $\ell \in \{-m, \dots, m\} \setminus \{0\}$. We observe that
		\[\left|a_{\ell}^{(2m)}\right| = \frac{(m!)^2}{|\ell|(m+|\ell|)!(m-|\ell|)!} = \frac{1}{|\ell|} \prod_{j=1}^{|\ell|} \frac{m-|\ell|+j}{m+j} < \frac{1}{|\ell|}.\]
		This completes the proof of the second statement. Finally, let $k$ be an integer such that $k \geq 2m+1$. Then, we find that
		\begin{align*}
		\left|\sum_{\ell=-m}^m a_{\ell}^{(2m)} \ell^k\right| &\leq \sum_{\underset{\ell\neq0}{\ell=-m}}^m \left|a_{\ell}^{(2m)}\right| \cdot |\ell|^k \leq \sum_{\ell=-m}^m |\ell|^{k-1} = 2\sum_{\ell=1}^m \ell^{k-1} \leq 2 \int_1^m \ell^{k-1} \;\d\ell + 2m^{k-1} \\
		&= 2\left[\frac{\ell^k}{k}\right]_1^m + 2m^{k-1} = \frac{2m^k}{k} - \frac{2}{k} + 2m^{k-1} \leq \frac{2m^k}{2m} + 2m^{k-1} = 3m^{k-1} \leq 2m^k.
		\end{align*}
		In the final inequality we used that $m \geq 2$. This completes the proof.
	\end{proof}
	
	Using the coefficients from \autoref{def:coefficients}, we can now formally define smoothings of functions.
	
	\begin{definition}{Smoothings of functions}
		Let $f : \R^d \to \R$ and $m \in \N$. We let the \textit{$2m$-th order smoothing of $f$} be the function $f_{(2m)} : \R^d \to \R$, defined as follows:
		\[f_{(2m)}(\vec{x}) = \sum_{\ell=-m}^m a_{\ell}^{(2m)} f(\ell\vec{x}).\]
	\end{definition}
	
	There is a very nice intuitive way to think about the smoothings of some function $f$. Suppose that a rope lying on a table and you want to straighten it out. One way to achieve this is to grab the rope on both ends, and pull. At the points where you grab the rope, far from the middle, you might distort the linearity of the rope to get a good grip on it. However, in the middle, the tension in the rope straightens it out. The same thing happens with the smoothings of $f$. They might be far from linear at points far from the origin, but around the origin they become approximately linear over a longer interval. This intuitive picture is quantified in \autoref{lem:justification_linearity} where we calculate how well $f$ approximates its linearization around the origin.
	
	The final question that remains is how one can implement a phase oracle to smoothings of $f$, when one has access to the phase oracle of $f$. The following theorem elaborates on how one can achieve this.
	
	\begin{theorem}{Implementation of phase oracle evaluating smoothings of functions}
		\label{thm:smoothing_phase_oracle}
		Let $d,n,n' \in \N$, $f : \R^d \to \R$, $m \in \N$, $G \subseteq \R^d$, with $|G| \leq 2^{nd}$, and $\{\ket{\vec{x}}_G : \vec{x} \in G\}$ be an orthonormal set of $nd$-qubit states. Define
		\[\overline{G} = \{\ell\vec{x} : \vec{x} \in G, \ell \in \{-m, \dots, m\}\},\]
		and let $\{\ket{\vec{x}}_{\overline{G}} : \vec{x} \in \overline{G}\}$ be an orthogonal set of $n'd$-qubit states.
		Let $\delta > 0$. Then, using $\widetilde{\O}\left(m\log\left(\frac{m}{\delta}\right)\right)$ queries to a phase oracle $O_{f,\overline{G}}$ evaluating $f$ on $\overline{G}$, we can implement an operator $U$ acting on $nd$ qubits as well we $N' = \Theta(nd)$ qubits, such that
		\[\norm{\left(\bra{0}^{\otimes N'} \otimes I_{2^{nd}}\right) U \left(\ket{0}^{\otimes N'} \otimes I_{2^{nd}}\right) - O_{f_{(2m)},G}} \leq \delta,\]
		where $O_{f_{(2m)},G}$ is a phase oracle evaluating the $2m$-th order smoothing of $f$, $f_{(2m)}$, on $G$. Using the terminology introduced by Gily\'en et al.\ in \cite{Gilyen18b}, Definition 43, $U$ is a $(1,N-nd,\delta)$-block-encoding of $O_{f_{(2m)},G}$.
	\end{theorem}
	
	\begin{proof}
		Let $\ell \in \{-m, \dots, m\}$. The first ingredient we need is a multiplication circuit, $M_{\ell}$, that will perform the following mapping:
		\[M_{\ell} : \ket{\vec{x}}_G\ket{0}^{\otimes(n'd)} \mapsto \ket{\vec{x}}_G\ket{\ell\vec{x}}_{\overline{G}}.\]
		How exactly one constructs this circuit is at this point not relevant, as one would first have to describe how the states $\vec{x}_G$ and $\ket{\vec{x}}_{\overline{G}}$, for $\vec{x} \in G$ or $\vec{x} \in \overline{G}$, respectively, look like. The point is that the circuit $M_{\ell}$ can be implemented without querying the phase oracle $O_{f,\overline{G}}$. For a more thorough description of what this circuit \textit{could} look like, see Circuit 4.3.7 in \cite{Cornelissen18}.
		
		Now, note that
		\[\left(O_{f_{(2m)},G} \otimes I_{2^{n'd}}\right)\ket{\vec{x}}_G\ket{0}^{\otimes(n'd)} = \left(\prod_{\ell=-m}^m M_{\ell}^{\dagger}\left(I_{2^{nd}} \otimes O_{f,\overline{G}}^{a_{\ell}^{(2m)}}\right)M_{\ell}\right)\ket{\vec{x}}_G\ket{0}^{\otimes(n'd)}.\]
		Indeed, for any $\vec{x} \in G$ and $\ell \in \{-m, \dots, m\}$, we find that
		\begin{align*}
		M_{\ell}^{\dagger}\left(I_{2^{nd}} \otimes O_{f,\overline{G}}^{a_{\ell}^{(2m)}}\right)M_{\ell}\ket{\vec{x}}_G\ket{0}^{\otimes(n'd)} &= M_{\ell}^{\dagger}\left(I_{2^{nd}} \otimes O_{f,\overline{G}}^{a_{\ell}^{(2m)}}\right)\ket{\vec{x}}_G\ket{\ell\vec{x}}_{\overline{G}} = e^{ia_{\ell}^{(2m)}f(\ell\vec{x})}M_{\ell}^{\dagger}\ket{\vec{x}}_G\ket{\ell\vec{x}}_{\overline{G}} \\
		&= e^{ia_{\ell}^{(2m)}f(\ell\vec{x})}\ket{\vec{x}}_G\ket{0}^{\otimes(n'd)}.
		\end{align*}
		Hence,
		\[\left[\prod_{\ell=-m}^m M_{\ell}^{\dagger}\left(I_{2^{nd}} \otimes O_{f,\overline{G}}^{a_{\ell}^{(2m)}}\right)M_{\ell}\right]\ket{\vec{x}}_G\ket{0}^{\otimes(n'd)} = \prod_{\ell=-m}^m e^{ia_{\ell}^{(2m)}f(\ell\vec{x})}\ket{\vec{x}}_G\ket{0}^{\otimes(n'd)} = e^{if_{(2m)}(\vec{x})}\ket{\vec{x}}_G\ket{0}^{\otimes(n'd)}.\]
		Furthermore, recall from \autoref{thm:central_difference_scheme_properties} that, for all non-zero $\ell$, the coefficient $a_{\ell}^{(2m)}$ is strictly smaller than $1$. Moreover, $a_0^{(2m)} = 1$. Thus we can implement the phase oracle evaluating $f_{(2m)}$ on $G$ with $2m$ calls to fractional phase oracles and one call to the normal phase oracle of $f$ on $\overline{G}$. If we implement the fractional phase oracles with precision $\delta/(2m)$, we obtain $O_{f_{(2m)},G}$ up to precision $\delta$ in the operator norm, and according to \autoref{thm:fraction_phase_oracle} the total number of queries to $O_{f,\overline{G}}$ is
		\[2m \cdot \widetilde{\O}\left(\log\left(\frac{2m}{\delta}\right)\right) + 1 = \widetilde{\O}\left(m\log\left(\frac{m}{\delta}\right)\right).\]
		This completes the proof.
	\end{proof}
	
	This completes our discussion on the numerical methods employed in the quantum gradient estimation algorithm. In the next section, we elaborate on other results needed in the remainder of this text.
	
	\subsection{Quantum gradient estimation algorithm}
	\label{subsec:algorithm}
	
	In this section, we arrive at the first main result of this paper, \autoref{alg:QGE}. We make a distinction between parameters, for which the user can choose suitable values depending on the application, and derived constants, which are calculated from these parameters. We provide explicit formulas for all derived constants, such that one could in principle implement this algorithm without going through the proofs presented in subsequent sections.
	
	% Don't put the following footnote somewhere else in the text,
	% as the counter is used in the box below.
	\addtocounter{footnote}{1}
	\begin{algorithm}{Quantum gradient estimation}
		\label{alg:QGE}
		\textbf{Description:} Given phase oracle access to a function $f \in \G_{d,c,\sigma,[-mr,mr]^d}$, this algorithm calculates an $\varepsilon$-precise $\ell^p$-approximate estimate of $\nabla f(\vec{0})$ with success probability at least $2/3$.
		
		\textbf{Parameters:}
		\begin{enumerate}
			\item $\sigma \in [1/2,1]$: the first parameter that determines the smoothness of the objective function.
			\item $c > 0$: the second parameter that determines the smoothness of the objective function.
			\item $p \in [1,\infty]$: the norm w.r.t.\ which we measure the precision of our estimate of the gradient.
			\item $d \in \N$: the dimension of the domain of the objective function.
			\item $\varepsilon \in (0,c)$: the precision with which we want to determine the gradient.
		\end{enumerate}
	
		\textbf{Derived constants:}
		\begin{enumerate}
			\item $\varepsilon' = \frac{\varepsilon}{d^{\frac1p}}$.
			\item $m = \max\left\{\left\lceil \log\left(\frac{cd^{\sigma}}{\varepsilon'}\right)\right\rceil, 2\right\}$.
			\item $r = \frac{2^{\sigma}}{2emcd^{\sigma}} \cdot \left(\frac{2^{\sigma}\varepsilon'}{272\pi emcd^{\sigma}}\right)^{\frac{1}{2m}}$.
			\item $S = \left\lceil \frac{8\pi}{r\varepsilon'} \right\rceil$.
			\item $n = \left\lceil \log\left(\frac{12c}{\varepsilon'}\right)\right\rceil$.
			\item $\vec{x}_{\vec{k}} = \frac{r}{2^n}\left(\vec{k} + \vec{\frac12}\right)$, for all $\vec{k} \in \{-2^{n-1}, \dots, 2^{n-1}-1\}^d$.
			\item $G = \{\vec{x}_{\vec{k}} : \vec{k} \in \{-2^{n-1}, \dots, 2^{n-1}-1\}^d\}$.
			\item $\overline{G} = \{\ell\vec{x} : \vec{x} \in G, \ell \in \{-m, \dots, m\}\}$.
			\item $N = \left\lceil18\log(3d)\right\rceil$.
		\end{enumerate}
		
		\textbf{Input oracle:} A controlled phase oracle $C(O_{f,\overline{G}})$, i.e., a quantum operation that performs the following action: $C(O_{f,\overline{G}}) = \ket{0}\bra{0} \otimes I + \ket{1}\bra{1} \otimes O_{f,\overline{G}}$.
		
		\textbf{Output:} A vector $\vec{v} \in \R^d$ that satisfies $\norm{\vec{v} - \nabla f(\vec{0})}_p \leq \varepsilon$.
		
		\textbf{Worst-case success probability (\autoref{thm:success_probability}):} $2/3$.

		\textbf{Worst-case query complexity (\autoref{thm:query_complexity}):}
		\begin{equation}
			\label{eq:query_complexity}
			\widetilde{\O}\left(\frac{cd^{\sigma+\frac1p}}{\varepsilon}\right).
		\end{equation}
		
		\textbf{Algorithm:}
		\begin{enumerate}
			\item Repeat $N$ times the following quantum circuit on $nd + N'$ qubits:
			\begin{enumerate}
				\item Construct a uniform superposition over all points in $G$ in a register with $nd$ qubits.
				\item Apply $O_{f_{(2m)},G}$ a total of $S$ times on this register, using the construction outlined in \autoref{thm:smoothing_phase_oracle} with precision $\delta = 1/(12\sqrt{2}S)$. This requires the use of $N' = \Theta(nd)$ auxiliary qubits.
				\item Apply the inverse quantum Fourier transform in each of the $d$ directions separately. That is, implement the unitary $\widetilde{\QFT}^{\dagger}$, which acts as follows for all $\vec{k} \in \{-2^{n-1}, \dots, 2^{n-1}-1\}^d$,\footnotemark[\thefootnote]
				\begin{equation}
					\widetilde{\QFT}^{\dagger} : \ket{\vec{x}_{\vec{k}}}_G \mapsto \frac{1}{\sqrt{2^{nd}}}\sum_{\vec{h} \in \{-2^{n-1}, \dots, 2^{n-1}-1\}^d} e^{-\frac{2\pi i\vec{k} \cdot \vec{h}}{2^n}}\ket{\vec{h}}.
					\label{eq:inverseQFT}
				\end{equation}
				\item Measure in the computational basis and denote the result by $\vec{h} \in \{-2^{n-1}, \dots, 2^{n-1}-1\}^d$.
				\item Calculate $\vec{g} = \frac{2\pi}{Sr}\vec{h}$.
			\end{enumerate}
			\item Let $\vec{v}$ be the coordinate-wise mean of the $N$ vectors $\vec{g}$ that were obtained in the previous step.
			\item Return $\vec{v}$.
		\end{enumerate}
	\end{algorithm}
	\footnotetext{How one implements this unitary is not of particular relevance at this stage, since this can be implemented without making any oracle calls. The implementation depends on how the states $\ket{\vec{x}}_G$, where $\vec{x} \in G$, are embedded in the $nd$-qubit state space. If one chooses to embed the state $\ket{\vec{x}_{\vec{k}}}_G$ as $\ket{\vec{k}} = \ket{k_1} \cdots \ket{k_d}$, for all $\vec{k} \in \{-2^{n-1}, \dots, 2^{n-1}-1\}^d$, then this operation reduces to $d$ parallel application of the $n$-qubit Fourier transform, which takes $\widetilde{\O}\left(nd\right)$ gates.}

	Two claims in the box above need a proof. First, we prove that the query complexity is as claimed in \autoref{subsec:query_complexity}. Afterwards, in \autoref{subsec:success_probability}, we prove that the success probability is indeed lower bounded by $2/3$, see \autoref{thm:success_probability}.
	
	\subsection{Query complexity}
	\label{subsec:query_complexity}
	
	From \autoref{alg:QGE}, it is not directly clear that the number of calls to $O_{f,G}$ is indeed asymptotically given by the formula in \autoref{eq:query_complexity}. We prove this below.
	
	\begin{theorem}{Query complexity of \autoref{alg:QGE}}
		\label{thm:query_complexity}
		\autoref{alg:QGE} has query complexity:
		\[\widetilde{\O}\left(\frac{cd^{\sigma+\frac1p}}{\varepsilon}\right).\]
	\end{theorem}

	\begin{proof}
		Note from \autoref{thm:smoothing_phase_oracle} that the number of queries to $O_{f,\overline{G}}$ for implementing $O_{f_{(2m)},G}$ scales as $\widetilde{\O}(m)$. Hence, the total query complexity of \autoref{alg:QGE} to $O_{f,\overline{G}}$ is $\widetilde{\O}(mNS)$. As in \autoref{alg:QGE} we let $m$ and $N$ scale logarithmically in all parameters, they are absorbed in the tilde, and so we just have to show that
		\[S = \widetilde{O}\left(\frac{cd^{\sigma}}{\varepsilon'}\right).\]
		Using $\varepsilon' \leq \varepsilon \leq c$, $m \geq 2$, $d \geq 1$ and $\sigma \in [\frac12,1]$, observe that
		\begin{equation}
			\label{eq:Sgeq1}
			\frac{8\pi}{r\varepsilon'} = \frac{8\pi}{\varepsilon'} \cdot \frac{2emcd^{\sigma}}{2^{\sigma}} \cdot \left(\frac{2^{\sigma}\varepsilon'}{272\pi emcd^{\sigma}}\right)^{-\frac{1}{2m}} \geq \frac{16\pi mcd^{\sigma}}{2^{\sigma}\varepsilon'} \geq 8\pi cmd^{\sigma} \geq 1.
		\end{equation}
		Hence, using $\lceil x \rceil \leq 2x$ whenever $x \geq 1$,
		\[\frac{8\pi}{r\varepsilon'} \leq S = \left\lceil \frac{8\pi}{r\varepsilon'} \right\rceil \leq 2 \cdot \frac{8\pi}{r\varepsilon'} \qquad \Rightarrow \qquad S = \Theta\left(\frac{1}{r\varepsilon'}\right).\]
		It remains to show that
		\[\frac1r = \widetilde{\O}\left(cd^{\sigma}\right).\]
		To that end, we have
		\[\frac{1}{r} = \frac{2emcd^{\sigma}}{2^{\sigma}} \cdot \left(\frac{2^{\sigma}\varepsilon'}{272\pi emcd^{\sigma}}\right)^{-\frac{1}{2m}} = \frac{2em}{2^{\sigma}} \cdot cd^{\sigma} \cdot \left(\frac{272\pi emcd^{\sigma}}{2^{\sigma}\varepsilon'}\right)^{\frac{1}{2m}},\]
		so we just have to show that
		\[\left(\frac{272\pi emcd^{\sigma}}{2^{\sigma}\varepsilon'}\right)^{\frac{1}{2m}} = \widetilde{\O}(1).\]
		We find that
		\[\lim_{m \to \infty} m^{\frac{1}{2m}} = \sqrt{\lim_{m \to \infty} m^{\frac1m}} = \sqrt{1} = 1 \qquad \text{and} \qquad \lim_{m \to \infty} \left(\frac{272\pi e}{2^{\sigma}}\right)^{\frac{1}{2m}} = \left(\frac{272\pi e}{2^{\sigma}}\right)^0 = 1,\]
		so it remains to show that
		\[\left(\frac{cd^{\sigma}}{\varepsilon'}\right)^{\frac{1}{2m}} = \widetilde{\O}\left(1\right).\]
		To that end, observe that
		\[\left(\frac{cd^{\sigma}}{\varepsilon'}\right)^{\frac{1}{2m}} = 2^{\frac{\log\left(\frac{cd^{\sigma}}{\varepsilon'}\right)}{2\left\lceil\log\left( \frac{cd^{\sigma}}{\varepsilon'} \right)\right\rceil}} \leq 2^{\frac12} = \O(1).\]
		This completes the proof.
	\end{proof}

	We have now proved one of the two claims that in \autoref{alg:QGE}. The second claim is the topic of the next section.

	\subsection{Success probability}
	\label{subsec:success_probability}
	
	The main result of this section is \autoref{thm:success_probability}, in which we prove that the success probability of \autoref{alg:QGE} is indeed lower bounded by $2/3$, as claimed in the box. The proof is rather long, so we have divided it into more manageable chunks.
	
	The core lemma used in the proof is the \textit{method of bounding the second moments of higher order bounded tensors}, as introduced by Gily\'en et al.\ in \cite{Gilyen18}, Lemma 36. We generalize their result by noting that one can also obtain a trivial bound for the same expression, which can be geometrically averaged with the Gily\'en bound to obtain a non-trivial result.
	
	\begin{lemma}{Method of bounding the second moments of higher order bounded tensors}
		\label{lem:bounding_second_moment_tensors}
		Let $d \in \N$ and let $x_1, \dots, x_d$ be independent identically distributed random variables over the interval $[-\frac12,\frac12]$, such that for all $j \in [d]$,
		\[\E\left[x_j\right] = 0.\]
		Then, for all $k \in \N$ and $q \in [0,1]$,
		\[\E\left[\left(\sum_{j=1}^d x_j\right)^{2k}\right] \leq \left[2\left(\frac{d}{2}\right)^kk!\right]^q \cdot \left[\left(\frac{d}{2}\right)^{2k}\right]^{1-q}.\]
	\end{lemma}

	\begin{proof}
		We prove the $q = 0$ and $q = 1$ cases. The result then follows from geometric averaging.
		
		First, we focus on the $q = 0$ case. Note that all $x_j$'s are bounded in absolute value by $\frac12$. Hence,
		\[\E\left[\left(\sum_{j=1}^d x_j\right)^{2k}\right] \leq \E\left[\left(\sum_{j=1}^d |x_j|\right)^{2k}\right] \leq \E\left[\left(d \cdot \frac12\right)^{2k}\right] = \left(\frac{d}{2}\right)^{2k},\]
		which completes the proof for the case $q = 0$.
		
		That leaves the $q = 1$ case. To that end, observe (this idea is due to Gily\'en et al.) that
		\[\E\left[\left(\sum_{j=1}^d x_j\right)^{2k}\right] = \int_0^\infty \P\left[\left(\sum_{j=1}^d x_j\right)^{2k} \geq t\right] \;\d t = \int_0^\infty \P\left[\left|\sum_{j=1}^d x_j\right| \geq t^{\frac{1}{2k}}\right] \;\d t \leq 2\int_0^{\infty} e^{-\frac{2t^{\frac1k}}{d}} \;\d t.\]
		In the last inequality, we used Hoeffding's inequality (\autoref{thm:statistics}). Performing variable substitution
		\[y = \frac{2t^{\frac1k}}{d} \qquad \Leftrightarrow \qquad t = \left(\frac{dy}{2}\right)^k \qquad \Rightarrow \qquad \d t = \left(\frac{d}{2}\right)^kky^{k-1}\;\d y,\]
		we obtain
		\[\E\left[\left(\sum_{j=1}^d x_j\right)^{2k}\right] \leq 2\left(\frac{d}{2}\right)^k \int_0^{\infty}ky^{k-1}e^{-y} \;\d y = 2\left(\frac{d}{2}\right)^kk\Gamma(k) = 2\left(\frac{d}{2}\right)^kk!.\]
		This completes the proof of the $q = 1$ case, and hence finishes the entire proof.
	\end{proof}

	Note that when $k$ is fixed and $d$ tends to infinity, then the $q = 1$ bound is tighter, whereas when $d$ is fixed and $k$ tends to infinity, the $q = 0$ bound is tighter.
	
	Next, we show how the result from \autoref{lem:bounding_second_moment_tensors} can be used to show that the smoothing of $f$ is on average close to its linearization. Here we significantly clean up the derivation compared to the one used in \cite{Gilyen18}, Theorems 24 and 25, as we substitute the need for Chebyshev's inequality and a small region on which the bound might not hold with a clever use of the relation $\Var[X] = \E[X^2] - \E[X]^2$.

	\begin{lemma}{Justification of approximate linearity}
		\label{lem:justification_linearity}
		Let $G$ and $S$ be as in \autoref{alg:QGE}, and let $f \in \G_{d,c,\sigma,[-mr,mr]^d}$. Let $U(G)$ be the uniform distribution over $G$. Then:
		\[\underset{\vec{x} \sim U(G)}{\E}\left[\left|f_{(2m)}(\vec{x}) - f(\vec{0}) - \nabla f(\vec{0}) \cdot \vec{x}\right|^2\right] \leq \frac{1}{144S^2}\]
	\end{lemma}

	\begin{proof}
		Let $\vec{x} \in G$ be arbitrary and we define the function $f_{\vec{x}} : t \mapsto f(t\vec{x})$. We find, for all $k \in \N_0$ and all $t \in [-m,m]$:
		\[f_{\vec{x}}^{(k)}(t) = \frac{\d^kf_{\vec{x}}}{\d t^k}(t) = \sum_{\alpha \in [d]^k} (\partial_{\alpha}f)(t\vec{x}) \cdot \vec{x}^{\alpha}.\]
		Hence,
		\[\left|f_{\vec{x}}^{(k)}(t)\right| \leq \sum_{\alpha \in [d]^k} \left|(\partial_{\alpha}f)(t\vec{x})\right| \cdot \norm{\vec{x}}_{\infty}^k \leq \frac12 (rcd)^k (k!)^{\sigma}.\]
		If $\sigma = 1$,
		\[r = \frac{2^{\sigma}}{2emcd^{\sigma}} \cdot \left(\frac{2^{\sigma}\varepsilon'}{272\pi emcd^{\sigma} }\right)^{\frac{1}{2m}} = \frac{1}{emcd} \cdot \left(\frac{\varepsilon'}{136\pi emcd}\right)^{\frac{1}{2m}} < \frac{1}{mcd},\]
		so the Taylor series of $f_{\vec{x}}$ converges at least on $[-m,m]$. We find, for all $t \in [-m,m]$,
		\[f(t\vec{x}) = f_{\vec{x}}(t) = \sum_{k=0}^{\infty} \frac{f_{\vec{x}}^{(k)}(0)}{k!} t^k = \sum_{k=0}^{\infty} \frac{t^k}{k!} \sum_{\alpha \in [d]^k} (\partial_{\alpha}f)(\vec{0}) \cdot \vec{x}^{\alpha},\]
		and hence
		\[f_{(2m)}(\vec{x}) = \sum_{\ell=-m}^m a_{\ell}^{(2m)}f(\ell\vec{x}) = \sum_{\ell=-m}^m a_{\ell}^{(2m)} \sum_{k=0}^{\infty} \frac{\ell^k}{k!} \sum_{\alpha \in [d]^k} (\partial_{\alpha}f)(\vec{0}) \cdot \vec{x}^{\alpha} = \sum_{k=0}^{\infty} \frac{1}{k!} \sum_{\alpha \in [d]^k} (\partial_{\alpha}f)(\vec{0}) \cdot \vec{x}^{\alpha} \sum_{\ell=-m}^m a_{\ell}^{(2m)}\ell^k.\]
		Using \autoref{thm:central_difference_scheme_properties},
		\[f_{(2m)}(\vec{x}) = f(\vec{0}) + \nabla f(\vec{0}) \cdot \vec{x} + \sum_{k=2m+1}^{\infty} \frac{1}{k!}\sum_{\alpha \in [d]^k} (\partial_{\alpha}f)(\vec{0}) \cdot \vec{x}^{\alpha} \cdot \sum_{\ell=-m}^m a_{\ell}^{(2m)} \ell^k.\]
		Thus, we obtain:
		\begin{align*}
			\left|f_{(2m)}(\vec{x}) - f(\vec{0}) - \nabla f(\vec{0}) \cdot \vec{x}\right| &\leq \sum_{k=2m+1}^{\infty} \frac{1}{k!} \left|\sum_{\alpha \in [d]^k} (\partial_{\alpha}f)(\vec{0}) \cdot \vec{x}^{\alpha}\right| \cdot \left|\sum_{\ell=-m}^m a_{\ell}^{(2m)} \ell^k\right| \\
			&\leq 2\sum_{k=2m+1}^{\infty} \frac{m^k}{k!} \left|\sum_{\alpha \in [d]^k} (\partial_{\alpha}f)(\vec{0}) \cdot \vec{x}^{\alpha}\right|, \\
		\end{align*}
		where the last inequality follows from $m \geq 2$ and \autoref{thm:central_difference_scheme_properties}. By squaring both sides, we obtain:
		\begin{align*}
			\left|f_{(2m)}(\vec{x}) - f(\vec{0}) - \nabla f(\vec{0}) \cdot \vec{x}\right|^2 &\leq \left(2\sum_{k=2m+1}^{\infty} \frac{m^{k-1}}{k!} \left|\sum_{\alpha \in [d]^k} (\partial_{\alpha}f)(\vec{0}) \cdot \vec{x}^{\alpha}\right|\right)^2 \\
			&\leq 4\sum_{k,\ell = 2m+1}^{\infty} \frac{m^{k+\ell}}{k!\ell!} \left|\sum_{\alpha \in [d]^k} \sum_{\beta \in [d]^{\ell}} (\partial_{\alpha}f)(\vec{0}) \cdot (\partial_{\beta}f)(\vec{0}) \cdot \vec{x}^{\alpha} \cdot \vec{x}^{\beta}\right| \\
			&= \sum_{k,\ell=2m+1}^{\infty} \frac{(mcr)^{k+\ell}}{(k!\ell!)^{1-\sigma}} \left|\sum_{\alpha \in [d]^k} \sum_{\beta \in [d]^{\ell}} \frac{(\partial_{\alpha}f)(\vec{0})}{\frac12c^k(k!)^{\sigma}} \cdot \frac{(\partial_{\beta}f)(\vec{0})}{\frac12c^\ell(\ell!)^{\sigma}} \cdot \frac{\vec{x}^{\alpha}}{r^k} \cdot \frac{\vec{x}^{\beta}}{r^{\ell}}\right|.
		\end{align*}
		Now, we define some abbreviations. First of all, for any $\alpha \in [d]^k$ and $\beta \in [d]^{\ell}$, we define:
		\[H_{\alpha,\beta} = \frac{(\partial_{\alpha}f)(\vec{0})}{\frac12c^k(k!)^{\sigma}} \cdot \frac{(\partial_{\beta}f)(\vec{0})}{\frac12c^{\ell}(\ell!)^{\sigma}}.\]
		Note that $|H_{\alpha,\beta}| \leq 1$ for all $\alpha \in [d]^k$ and $\beta \in [d]^{\ell}$, because $f \in \G_{d,c,\sigma,[-mr,mr]^d}$. Moreover, $G/r \subseteq [-1/2,1/2]^d$, thus
		\begin{align}
			\underset{\vec{x} \sim U(G)}{\E} \left[\left|f_{(2m)}(\vec{x}) - f(\vec{0}) - \nabla f(\vec{0}) \cdot \vec{x}\right|^2\right] &\leq \underset{\vec{x} \sim U\left(\frac{G}{r}\right)}{\E} \left[\sum_{k,\ell=2m+1}^{\infty} \frac{(mcr)^{k+\ell}}{(k!\ell!)^{1-\sigma}} \left|\sum_{\alpha \in [d]^k} \sum_{\beta \in [d]^{\ell}} H_{\alpha,\beta} \vec{x}^{\alpha} \vec{x}^{\beta}\right|\right] \nonumber \\
			&= \sum_{k,\ell=2m+1}^{\infty} \frac{(mcr)^{k+\ell}}{(k!\ell!)^{1-\sigma}} \underset{\vec{x} \sim U\left(\frac{G}{r}\right)}{\E} \left[\left|\sum_{\alpha \in [d]^k} \sum_{\beta \in [d]^{\ell}} H_{\alpha,\beta} \vec{x}^{\alpha} \vec{x}^{\beta}\right|\right]. \label{eq:linearity_bound}
		\end{align}
		Now, we focus on the innermost expectation. Using $\E[X^2] = \E[X]^2 + \Var[X] \geq \E[X]^2$, we obtain:
		\begin{align}
			\left[\underset{\vec{x} \sim U\left(\frac{G}{r}\right)}{\E} \left|\sum_{\alpha \in [d]^k} \sum_{\beta \in [d]^{\ell}} H_{\alpha,\beta} \vec{x}^{\alpha} \vec{x}^{\beta}\right|\right]^2 &\leq \underset{\vec{x} \sim U\left(\frac{G}{r}\right)}{\E} \left[\left(\sum_{\alpha \in [d]^k} \sum_{\beta \in [d]^{\ell}} H_{\alpha,\beta} \vec{x}^{\alpha} \vec{x}^{\beta}\right)^2\right] \nonumber \\
			&= \underset{\vec{x} \sim U\left(\frac{G}{r}\right)}{\E} \left[\sum_{\alpha \in [d]^k} \sum_{\beta \in [d]^{\ell}}  \sum_{\gamma \in [d]^k} \sum_{\delta \in [d]^{\ell}} H_{\alpha,\beta} H_{\gamma,\delta} \vec{x}^{\alpha} \vec{x}^{\beta} \vec{x}^{\gamma} \vec{x}^{\delta}\right] \nonumber \\
			&= \sum_{\alpha \in [d]^k} \sum_{\beta \in [d]^{\ell}}  \sum_{\gamma \in [d]^k} \sum_{\delta \in [d]^{\ell}} H_{\alpha,\beta} H_{\gamma,\delta} \underset{\vec{x} \sim U\left(\frac{G}{r}\right)}{\E} \left[\vec{x}^{\alpha} \vec{x}^{\beta} \vec{x}^{\gamma} \vec{x}^{\delta}\right] \nonumber \\
			& \leq \sum_{\alpha \in [d]^k} \sum_{\beta \in [d]^{\ell}}  \sum_{\gamma \in [d]^k} \sum_{\delta \in [d]^{\ell}} \underset{\vec{x} \sim U\left(\frac{G}{r}\right)}{\E} \left[\vec{x}^{\alpha} \vec{x}^{\beta} \vec{x}^{\gamma} \vec{x}^{\delta}\right] \label{eq:second_order_moment} \\
			&= \underset{\vec{x} \sim U\left(\frac{G}{r}\right)}{\E} \left[ \sum_{\alpha \in [d]^k} \sum_{\beta \in [d]^{\ell}}  \sum_{\gamma \in [d]^k} \sum_{\delta \in [d]^{\ell}} \vec{x}^{\alpha} \vec{x}^{\beta} \vec{x}^{\gamma} \vec{x}^{\delta}\right] \nonumber \\
			&= \underset{\vec{x} \sim U\left(\frac{G}{r}\right)}{\E} \left[\sum_{\alpha \in [d]^{2(k+\ell)}} \vec{x}^{\alpha}\right] = \underset{\vec{x} \sim U\left(\frac{G}{r}\right)}{\E} \left[\left(\sum_{j=1}^d x_j\right)^{2(k+\ell)}\right]. \nonumber
		\end{align}
		The inequality in \autoref{eq:second_order_moment} needs some justification. Observe that for all $\alpha,\gamma \in [d]^k$ and $\beta,\delta \in [d]^{\ell}$, the expression $\vec{x}^{\alpha} \vec{x}^{\beta} \vec{x}^{\gamma} \vec{x}^{\delta}$ is a product $x_1^{\zeta_1} \cdots x_d^{\zeta_d}$ with non-negative integer values for $\zeta_1, \dots, \zeta_d$. The shape of the grid ensures that the variables $x_1, \dots, x_d$ are independent, and hence the expectation of the product is a product of the expectations. As the grid is placed symmetrically around the origin, for all $j \in [d]$ the expectation of any odd power of $x_j$ vanishes. On the other hand, the expectation of any even power of $x_j$ is clearly positive. Thus:
		\[\underset{\vec{x} \sim U\left(\frac{G}{r}\right)}{\E} \left[\vec{x}^{\alpha} \vec{x}^{\beta} \vec{x}^{\gamma} \vec{x}^{\delta}\right] = \underset{\vec{x} \sim U\left(\frac{G}{r}\right)}{\E} \left[x_1^{\zeta} \cdots x_d^{\zeta_d}\right] = \prod_{j=1}^d \underset{x_j \sim U\left(\frac{\{-2^{n-1}+\frac12, \dots, 2^{n-1}-\frac12\}}{2^n}\right)}{\E} \left[x_j^{\zeta_j}\right] \geq 0.\]
		This justifies the inequality in \autoref{eq:second_order_moment}. Now, we use \autoref{lem:bounding_second_moment_tensors} with $q = 2(1-\sigma)$ to obtain:
		\begin{align}
			\left[\underset{\vec{x} \sim U\left(\frac{G}{r}\right)}{\E} \left|\sum_{\alpha \in [d]^k} \sum_{\beta \in [d]^{\ell}} H_{\alpha,\beta} \vec{x}^{\alpha} \vec{x}^{\beta}\right|\right]^2 &\leq \underset{\vec{x} \sim U\left(\frac{G}{r}\right)}{\E} \left[\left(\sum_{j=1}^d x_j\right)^{2(k+\ell)}\right] \leq \left[2\left(\frac{d}{2}\right)^{k+\ell}(k+\ell)!\right]^q \cdot \left[\left(\frac{d}{2}\right)^{2(k+\ell)}\right]^{1-q} \nonumber \\
			&= 2^q \left(\frac{d}{2}\right)^{(k+\ell)(q + 2(1-q))} \cdot \left[(k+\ell)!\right]^q \nonumber \\
			&= 2^{2(1-\sigma)} \left(\frac{d}{2}\right)^{(k+\ell) \cdot (2-2(1-\sigma))} \cdot \left[(k+\ell)!\right]^{2(1-\sigma)} \nonumber \\
			&= 2^{2(1-\sigma)} \cdot \left(\frac{d}{2}\right)^{(k+\ell) \cdot 2\sigma} \cdot \left[(k+\ell)!\right]^{2(1-\sigma)}. \label{eq:second_order_moment_bound}
		\end{align}
		Taking a square root and substituting \autoref{eq:second_order_moment_bound} into \autoref{eq:linearity_bound} yields:
		\begin{align*}
			\underset{\vec{x} \sim U(G)}{\E} \left[\left|f_{(2m)}(\vec{x}) - f(\vec{0}) - \nabla f(\vec{0}) \cdot \vec{x}\right|^2\right] &\leq \sum_{k,\ell=2m+1}^{\infty} \frac{(mcr)^{k+\ell}}{(k!\ell!)^{1-\sigma}} \cdot 2^{1-\sigma} \cdot \left(\frac{d}{2}\right)^{(k+\ell)\sigma} \cdot \left[(k+\ell)!\right]^{1-\sigma} \\
			&\leq 2^{1-\sigma} \sum_{k,\ell=2m+1}^{\infty} \left(\frac{mcrd^{\sigma}}{2^{\sigma}}\right)^{k+\ell} \cdot \left(\frac{(k+\ell)!}{k!\ell!}\right)^{1-\sigma}.
		\end{align*}
		Using Stirling's approximation, see \autoref{thm:stirling} in \autoref{sec:misc_results},
		\[\frac{(k+\ell)!}{k!\ell!} \leq \frac{(k+\ell)^{k+\ell+\frac12}e^{-k-\ell}e}{k^{k+\frac12} \ell^{\ell+\frac12}e^{-k}e^{-\ell}2\pi} = \left(\frac{k+\ell}{k}\right)^k \cdot \left(\frac{k+\ell}{\ell}\right)^{\ell} \cdot \sqrt{\frac{k+\ell}{k\ell}} \cdot \frac{e}{2\pi} \leq \left(1 + \frac{\ell}{k}\right)^k \left(1 + \frac{k}{\ell}\right)^{\ell} \leq e^{k+\ell}.\]
		We end up with
		\[\underset{\vec{x} \sim U(G)}{\E} \left[\left|f_{(2m)}(\vec{x}) - f(\vec{0}) - \nabla f(\vec{0}) \cdot \vec{x}\right|^2\right] \leq 2^{1-\sigma}\sum_{k,\ell=2m+1}^{\infty} \left(\frac{emrcd^{\sigma}}{2^{\sigma}}\right)^{k+\ell} = 2^{1-\sigma}\left[\sum_{k=2m+1}^{\infty} \left(\frac{emrcd^{\sigma}}{2^{\sigma}}\right)^k\right]^2.\]
		Plugging in $r$ and using that $\varepsilon' \leq \varepsilon < c$ yields
		\[\frac{emrcd^{\sigma}}{2^{\sigma}} = \frac{emcd^{\sigma}}{2^{\sigma}} \cdot \frac{2^{\sigma}}{2emcd^{\sigma}} \left(\frac{2^{\sigma}\varepsilon'}{272\pi emcd^{\sigma}}\right)^{\frac{1}{2m}} \leq \frac12.\]
		Using that $S \geq 1$, as proven in \autoref{eq:Sgeq1}, we find that
		\[\sum_{k=2m+1}^{\infty} \left(\frac{emrcd^{\sigma}}{2^{\sigma}}\right)^k \leq \frac{emrcd^{\sigma}}{2^{\sigma}} \cdot \left(\frac{emrcd^{\sigma}}{2^{\sigma}}\right)^{2m} \cdot \sum_{k=0}^{\infty} \left(\frac12\right)^k = \frac{2emrcd^{\sigma}}{2^{\sigma} \cdot 2^{2m}} \cdot \frac{2^{\sigma}\varepsilon'}{272\pi emcd^{\sigma}} = \frac{r\varepsilon'}{272\pi \cdot 2^{2m}} \leq \frac{1}{17S}.\]
		Putting it all together, we find
		\[\underset{\vec{x} \sim U(G)}{\E} \left[\left|f_{(2m)}(\vec{x}) - f(\vec{0}) - \nabla f(\vec{0}) \cdot \vec{x}\right|^2\right] \leq 2^{1-\sigma} \left[\sum_{k=2m+1}^{\infty} \left(\frac{emrcd^{\sigma}}{2^{\sigma}}\right)^k\right]^2 \leq \frac{2}{(17S)^2} < \frac{1}{144S^2}.\]
		This completes the proof.
	\end{proof}

	Now that we have shown that the smoothing of $f$ is close to linear, we can use this to deduce that the state that we ideally would like to obtain from our algorithm is sufficiently well approximated. This is the objective of \autoref{lem:two_states_are_close}. It is an adapted version of the last part of Lemma 20 in \cite{Gilyen18}.

	\begin{lemma}{Norm error induced by non-linearity}
		\label{lem:two_states_are_close}
		Let $\ket{\psi}$ be the $(nd + N')$-qubit state directly after step (b) in \autoref{alg:QGE} and let
		\[\ket{\widetilde{\psi}} = \frac{1}{\sqrt{2^{nd}}} \sum_{\vec{k} \in \{-2^{n-1}, \dots, 2^{n-1}-1\}^d} e^{i\frac{Sr}{2^n}\nabla f(\vec{0}) \cdot \vec{k}} \ket{\vec{x}_{\vec{k}}}_G\]
		be an $nd$-qubit state. Then
		\[\min_{\phi \in \R}\norm{e^{i\phi}\ket{\psi} - \ket{\widetilde{\psi}} \otimes \ket{0}^{\otimes N'}} \leq \frac16.\]
	\end{lemma}
	
	\begin{proof}
		Define
		\[\ket{\chi} = \frac{1}{\sqrt{2^{nd}}} \sum_{\vec{x} \in G} e^{iSf_{(2m)}(\vec{x})} \ket{\vec{x}}.\]
		Recall that in step (b) of \autoref{alg:QGE}, we implemented the phase oracle evaluating $f_{(2m)}$ on $G$ up to precision $\delta = 1/(12\sqrt{2}S)$ in the operator norm with an operator $U$. Via an inductive argument on $S$ where in every step of the induction we use the triangle inequality, we obtain that
		\[\norm{\ket{\psi} - \ket{\chi} \otimes \ket{0}^{\otimes N'}} \leq \sqrt{2}S\norm{\left(I_{2^{nd}} \otimes \bra{0}^{\otimes N'}\right) U \left(I_{2^{nd}} \otimes \ket{0}^{\otimes N'}\right) - O_{f_{(2m)},G}} \leq \sqrt{2} \cdot S \cdot \frac{1}{12\sqrt{2}S} = \frac{1}{12}.\]
		It remains to prove that there is a $\phi' \in \R$ such that $\norm{e^{i\phi'}\ket{\chi} - \ket{\widetilde{\psi}}} \leq \frac{1}{12}$, as then
		\[\min_{\phi \in \R} \norm{e^{i\phi}\ket{\psi} \otimes \ket{0}^{\otimes N'} - \ket{\widetilde{\psi}} \otimes \ket{0}^{\otimes N'}} \leq \norm{e^{i\phi'}\ket{\psi} - e^{i\phi'}\ket{\chi}} + \norm{e^{i\phi'}\ket{\chi} - \ket{\widetilde{\psi}}} \leq \frac{1}{12} + \frac{1}{12} = \frac16.\]
		Choose $\phi' = -S(r/2^n \cdot \nabla f(\vec{0}) \cdot \vec{1/2} + f(\vec{0}))$. Then,
		\begin{align*}
		\norm{e^{i\phi'}\ket{\chi} - \ket{\widetilde{\psi}}}^2 &= \frac{1}{2^{nd}}\sum_{\vec{k} \in \{-2^{n-1}, \dots, 2^{n-1}-1\}^d} \left|e^{i(Sf_{(2m)}(\vec{x}_{\vec{k}}) + \phi')} - e^{i\frac{Sr}{2^n}\nabla f(\vec{0}) \cdot \vec{k}}\right|^2 \\
		&= \frac{1}{2^{nd}}\sum_{\vec{k} \in \{-2^{n-1}, \dots, 2^{n-1}-1\}^d} \left|e^{i(Sf_{(2m)}(\vec{x}_{\vec{k}}) - \frac{Sr}{2^n} \cdot \nabla f(\vec{0}) \cdot \vec{\frac12} - Sf(\vec{0}) - \frac{Sr}{2^n}\nabla f(\vec{0}) \cdot \vec{k})} - 1\right|^2 \\
		&\leq \frac{1}{2^{nd}}\sum_{\vec{k} \in \{-2^{n-1}, \dots, 2^{n-1}-1\}^d} \left|Sf_{(2m)}(\vec{x}_{\vec{k}}) - \frac{Sr}{2^n} \cdot \nabla f(\vec{0}) \cdot \vec{\frac12} - Sf(\vec{0}) - \frac{Sr}{2^n}\nabla f(\vec{0}) \cdot \vec{k}\right|^2 \\
		&= \frac{S^2}{2^{nd}}\sum_{\vec{k} \in \{-2^{n-1}, \dots, 2^{n-1}-1\}^d} \left|f_{(2m)}(\vec{x}_{\vec{k}}) - f(\vec{0}) - \nabla f(\vec{0}) \cdot \frac{r}{2^n}\left(\vec{k} + \vec{\frac12}\right)\right|^2 \\
		&= \frac{S^2}{2^{nd}}\sum_{\vec{k} \in \{-2^{n-1}, \dots, 2^{n-1}-1\}^d} \left|f_{(2m)}(\vec{x}_{\vec{k}}) - f(\vec{0}) - \nabla f(\vec{0}) \cdot \vec{x}_{\vec{k}}\right|^2 \\
		&= \frac{S^2}{2^{nd}} \sum_{\vec{x} \in G} \left|f_{(2m)}(\vec{x}) - f(\vec{0}) - \nabla f(\vec{0}) \cdot \vec{x}\right|^2 \\
		&= S^2 \underset{\vec{x} \in U(G)}{\E} \left[(f_{(2m)}(\vec{x}) - f(\vec{0}) - \nabla f(\vec{0}) \cdot \vec{x})^2\right] \\
		&\leq S^2 \cdot \frac{1}{144S^2} = \frac{1}{144},
		\end{align*}
		where we used \autoref{lem:justification_linearity} in the last line. This completes the proof.
	\end{proof}
	
	Now that we have shown that the smoothing of $f$ is sufficiently close to linear that the resulting state is not too far off from the state that we would have if $f$ were linear, we can move on and show that the algorithm recovers the slope of linear functions. To that end, recall a well-known robustness result of the quantum Fourier transform.
	
	\begin{theorem}{Robustness of the quantum Fourier transform}
		\label{thm:robustness_QFT}
		Let $n \geq 4$, $a \in [-2\pi/3,2\pi/3]$ and
		\[\ket{\phi} = \QFT_{2^n}^{\dagger} \left(\frac{1}{\sqrt{2^n}} \sum_{k = -2^{n-1}}^{2^{n-1}-1} e^{iak} \ket{k}\right).\]
		Suppose we perform a computational basis measurement on $\ket{\phi}$ and denote the outcome by $b \in \{-2^{n-1}, \dots, 2^{n-1}-1\}$. Then
		\[\P\left[\left|b-\frac{2^na}{2\pi}\right| \leq 4\right] \geq \frac56.\]
	\end{theorem}
	\begin{proof}
		See \cite{Nielsen00}, Equation 5.34.
	\end{proof}
	
	The next step is to show that every coordinate is independently approximated well with high probability. This is the object of \autoref{lem:inner_loop}, which is an adapted version of Theorem 21 in \cite{Gilyen18}.
	
	\begin{lemma}{Inner loop of step 1 of \autoref{alg:QGE}}
		\label{lem:inner_loop}
		Let $j \in [d]$ be arbitrary, and let $\vec{g} \in \R^d$ be a vector produced by the inner loop of step 1 of \autoref{alg:QGE}. Then
		\[\P\left[|g_j - \nabla f(\vec{0})_j| \leq \varepsilon'\right] \geq \frac23.\]
	\end{lemma}

	\begin{proof}
		If we apply the operator $\widetilde{\QFT}^{\dagger}$ from \autoref{eq:inverseQFT} to the state $\ket{\widetilde{\psi}}$, defined in \autoref{lem:two_states_are_close}, we obtain
		\begin{align*}
			\widetilde{\QFT}^{\dagger}\ket{\widetilde{\psi}} &= \frac{1}{\sqrt{2^{nd}}}\sum_{\vec{k} \in \{-2^{n-1}, \cdots, 2^{n-1}-1\}} e^{i\frac{Sr}{2^n}\nabla f(\vec{0}) \cdot \vec{k}} \QFT^{\dagger}\ket{\vec{x}_{\vec{k}}}_G \\
			&= \frac{1}{2^{nd}} \sum_{\vec{k} \in \{-2^{n-1}, \dots, 2^{n-1}-1\}^d} e^{i\frac{Sr}{2^n} \nabla f(\vec{0}) \cdot \vec{k}} \cdot \sum_{\vec{h} \in \{-2^{n-1}, \dots, 2^{n-1}-1\}^d} e^{-\frac{2\pi i\vec{k} \cdot \vec{h}}{2^n}} \ket{\vec{h}}.
		\end{align*}
		Using that $\ket{\vec{h}} = \ket{h_1} \cdots \ket{h_d}$, we observe that this state is a product state:
		\begin{align*}
			\widetilde{\QFT}^{\dagger}\ket{\widetilde{\psi}} &= \bigotimes_{j=1}^d \frac{1}{2^n} \sum_{k_j = -2^{n-1}}^{2^{n-1}-1} e^{i\frac{Sr}{2^n} \nabla f(\vec{0})_jk_j} \cdot \sum_{h_j = -2^{n-1}}^{2^{n-1}-1} e^{-\frac{2\pi ih_jk_j}{2^n}}\ket{h_j} \\
			&= \bigotimes_{j=1}^d \QFT_{2^n}^{\dagger} \left(\frac{1}{\sqrt{2^n}} \sum_{k_j = 2^{n-1}}^{2^{n-1}-1} e^{i\frac{Sr}{2^n}\nabla f(\vec{0})_j k_j} \ket{k_j}\right).
		\end{align*}
		Suppose we measure the $j$th register of this state in the computational basis, and denote the measurement outcome by $\widetilde{h}_j$. Furthermore, recall that we defined $\vec{g} = (2\pi)/(Sr)\vec{h}$ in \autoref{alg:QGE}. From the robustness of measurements, see \autoref{thm:measurement_robustness}, we obtain that
		\begin{align*}
			\P\left(\left|g_j - \nabla f(\vec{0})_j\right| \leq \varepsilon'\right) &= \P\left(\left|h_j - \frac{Sr}{2\pi}\nabla f(\vec{0})_j\right| \leq \frac{Sr\varepsilon'}{2\pi}\right) \\
			&\geq \P\left(\left|\widetilde{h}_j - \frac{Sr}{2\pi}\nabla f(\vec{0})_j\right| \leq \frac{Sr\varepsilon'}{2\pi}\right) - \min_{\phi \in \R}\norm{e^{i\phi}\left(\widetilde{\QFT}^{\dagger} \otimes I_{2^{N'}}\right)\ket{\psi} - \widetilde{\QFT}^{\dagger}\ket{\widetilde{\psi}} \otimes \ket{0}^{\otimes N'}}.
		\end{align*}
		Note that
		\[\frac{Sr\varepsilon'}{2\pi} \geq \frac{8\pi}{r\varepsilon'} \cdot \frac{r\varepsilon'}{2\pi} = 4,\]
		and using \autoref{eq:Sgeq1} and the Gevrey condition
		\[\norm{\frac{Sr}{2^n}\nabla f(\vec{0})}_{\infty} \leq 2 \cdot \frac{8\pi}{r\varepsilon'} \cdot \frac{r}{2^n} \norm{\nabla f(\vec{0})}_{\infty} \leq 2 \cdot \frac{8\pi}{2^n\varepsilon'} \cdot \frac12c \leq \frac{8\pi c}{2^{\log\left(\frac{12c}{\varepsilon'}\right)}\varepsilon'} = \frac{8c\pi\varepsilon'}{12c\varepsilon'} = \frac{2\pi}{3}.\]
		We employ \autoref{lem:two_states_are_close} and \autoref{thm:robustness_QFT} to obtain
		\begin{align*}
			\P\left(\left|g_j - \nabla f(\vec{0})_j\right| \leq \varepsilon'\right) &\geq \P\left(\left|\widetilde{h}_j - \frac{Sr}{2\pi}\nabla f(\vec{0})_j\right| \leq 4\right) - \min_{\phi \in \R}\norm{e^{i\phi}\ket{\psi} - \ket{\widetilde{\psi}} \otimes \ket{0}^{\otimes N'}} \\
			&\geq \frac56 - \frac16 = \frac23.
		\end{align*}
		This completes the proof.
	\end{proof}

	Finally, we show that only a logarithmic number of repetitions of the quantum routine allow for estimating the gradient sufficiently accurately with the required success probability. This is achieved in \autoref{thm:success_probability}, and also forms the culmination of the argument.

	\begin{theorem}{Success probability of \autoref{alg:QGE}}
		\label{thm:success_probability}
		\autoref{alg:QGE} succeeds with probability at least $2/3$.
	\end{theorem}

	\begin{proof}
		Let $j \in [d]$ and let $\vec{v}$ be a vector produced by \autoref{alg:QGE}. As $v_j$ is defined to be the median of the $g_j$'s obtained over consecutive independent runs of the loop in step 1 of the algorithm, we observe that $|v_j - \nabla f(\vec{0})_j| > \varepsilon'$ only if at least half of the $g_j$'s produced by step 1 satisfy $|g_j - \nabla f(\vec{0})_j| > \varepsilon'$.
		
		For all $k \in [N]$, let $B_k$ be the Bernoulli random variable that is $1$ if on the $k$th run of the loop in step 1 of the algorithm, $|g_j - \nabla f(\vec{0})_j| \leq \varepsilon'$. We find by \autoref{lem:inner_loop} that
		\[P = \P\left(B_k = 1\right) \geq \frac23 \qquad \text{and} \qquad \E\left[\sum_{k=1}^N B_k\right] = NP.\]
		Using Hoeffding's inequality, see \autoref{thm:statistics}, we find that
		\begin{align*}
			\P\left[\sum_{k=1}^N B_k \leq \frac{N}{2}\right] &= \P\left[\sum_{k=1}^N B_k - NP \leq N\left(\frac12 - P\right)\right] \leq e^{-\frac{2N^2\left(P-\frac12\right)^2}{N}} \leq e^{-\frac{N}{18}}.
		\end{align*}
		Hence,
		\[\P\left[\left|v_j - \nabla f(\vec{0})_j\right| > \varepsilon'\right] \leq \P\left[\sum_{k=1}^N B_k \leq \frac{N}{2}\right] \leq e^{-\frac{N}{18}} \leq e^{-\frac{1}{18} \cdot 18 \cdot \log(3d)} = \frac{1}{3d}.\]
		This relation holds regardless of our choice of $j$. Hence, using the union bound, see \autoref{thm:union_bound},
		\[\P\left[\norm{\vec{v} - \nabla f(\vec{0})}_{\infty} \leq \varepsilon'\right] \geq 1 - \sum_{j=1}^d \P\left[\left|v_j - \nabla f(\vec{0})_j\right| > \varepsilon'\right] \geq 1 - d \cdot \frac{1}{3d} = \frac23.\]
		Since for any $\vec{x} \in \R^d$,
		\[\norm{\vec{x}}_p = \left(\sum_{j=1}^d |x_j|^p\right)^{\frac1p} \leq d^{\frac1p}\norm{\vec{x}}_{\infty},\]
		we conclude that
		\[\P\left[\norm{\vec{v} - \nabla f(\vec{0})}_p \leq \varepsilon\right] \geq \P\left[d^{\frac1p}\norm{\vec{v} - \nabla f(\vec{0})}_{\infty} \leq \varepsilon\right] = \P\left[\norm{\vec{v} - \nabla f(\vec{0})}_{\infty} \leq \varepsilon'\right] \geq \frac23.\]
		This completes the proof.
	\end{proof}

	We have now proven that \autoref{alg:QGE} is an $\varepsilon$-precise $\ell^p$-approximate quantum gradient estimation algorithm for $\G_{d,c,\sigma,\R^d}$ on $G$ with success probability lower bounded by $2/3$, as defined in \autoref{def:quantum_gradient_estimation_algorithms}. Note that if we choose $N = \lceil 18\log(d/(1-P)) \rceil$ in \autoref{alg:QGE} instead, we obtain a success probability that is lower bounded by $P \in (1/2,1)$.
	
	There are many moving parts to the proofs in this section, and it is not directly obvious that all these parts connect tightly. However, despite considerable effort it seems like this argument cannot be improved in the case where $p = \infty$ and $\sigma \in (1/2,1]$ without introducing some fundamentally new ideas. This is a very interesting topic of further research, as we will further explain in \autoref{sec:conclusion}.

	\section{Quantum gradient estimation query complexity lower bounds}
	\label{sec:lower_bound}
	
	In this section, we introduce a new lower bound on the query complexity of the quantum gradient estimation problem as presented in \autoref{def:quantum_gradient_estimation_problem}. First, in \autoref{subsec:lower_bound_proof}, we prove a lower bound on the query complexity under the assumption that $p = 1$, $\sigma = 0$ and $P = 17/18$. Subsequently, in \autoref{subsec:general_lower_bound_proof}, we present reduction arguments with which we can obtain lower bounds on the query complexity for $p \in [1,\infty]$, $\sigma \geq 0$ and $P \in (1/2,1]$.
	
	\subsection{Lower bound for the case $p = 1$, $\sigma = 0$ and $P = 17/18$}
	\label{subsec:lower_bound_proof}
	
	In this section, we will formalize the ideas that were presented in \autoref{subsec:instance_selection} and \autoref{subsec:claw_selection}. We start by introducing some functions that we refer to as \textit{test functions}. These functions can intuitively be thought of as being very close to each other w.r.t.\ the supremum norm, but nonetheless having very different gradients.

	\begin{definition}{Test functions}
		\label{def:test_functions}
		Let $d \in \N$, $c > 0$ and $\varepsilon > 0$. We define, for all $\vec{b} \in \{-1,1\}^d$ and $\vec{x} \in \R^d$:
		\[f_{d,c,\varepsilon,\vec{b}} : \R^d \to \R, \qquad f_{d,c,\varepsilon,\vec{b}}(\vec{x}) = \sum_{j=1}^d \frac{73\varepsilon b_j}{cd} \sin\left(cx_j\right) \cdot \prod_{\underset{k \neq j}{k=1}}^d \cos\left(cx_k\right).\]
		Furthermore, we define the class of test functions as follows:
		\[\F_{d,c,\varepsilon} = \left\{f_{d,c,\varepsilon,\vec{b}} : \vec{b} \in \{-1,1\}^d\right\}.\]
	\end{definition}

	Let's first quantify the smoothness of these functions.

	\begin{lemma}{Test functions are Gevrey functions}
		\label{lem:test_functions_smoothness}
		Let $d \in \N$, $c > 0$ and $\varepsilon \in (0,c/146)$. Then $\F_{d,c,\varepsilon} \subseteq \G_{d,c,0,\R^d}$.
	\end{lemma}

	\begin{proof}
		Let $\ell \in \N_0$, and $\alpha \in [d]^{\ell}$. For all $j \in [d]$, let $a_j$ denote the number of occurrences of $j$ in $\alpha$. Let $\vec{b} \in \{-1,1\}^d$. We find, for all $\vec{x} \in \R^d$:
		\begin{equation}
			\label{eq:partial_derivative_test_functions}
			\partial_{\alpha}f_{d,c,\varepsilon,\vec{b}}(\vec{x}) = \sum_{j=1}^d \frac{73\varepsilon b_j}{cd} c^{a_j}\sin^{(a_j)}(cx_j) \cdot \prod_{\underset{k \neq j}{k=1}}^d c^{a_k} \cos^{(a_k)}(cx_k).
		\end{equation}
		Since all derivatives of cosines and sines are again cosines and sines, they are bounded by $1$ in absolute value. Hence,
		\[|\partial_{\alpha}f_{d,c,\varepsilon,\vec{b}}(\vec{x})| \leq \sum_{j=1}^d \frac{73\varepsilon |b_j|}{cd} c^{a_j} \cdot \prod_{\underset{k \neq j}{k=1}}^d c^{a_k} = \frac{73\varepsilon}{cd} \cdot d \cdot c^{\sum_{j=1}^d a_j} = \frac{73\varepsilon}{c} \cdot c^{\ell} \leq \frac12c^{\ell},\]
		where we used $\varepsilon < c/146$ in the last step. We find that $f_{d,c,\varepsilon,\vec{b}} \in \G_{d,c,0,\R^d}$. As this holds for any $\vec{b} \in \{-1,1\}^d$, $\F_{d,c,\varepsilon} \subseteq \G_{d,c,0,\R^d}$.
	\end{proof}

	Hence the test functions are members of Gevrey classes with parameter $\sigma = 0$. This indicates that we might be able to use these functions to obtain lower bounds on the query complexity of the gradient estimation problem with $\sigma \geq 0$, and that they are useless for proving lower bounds when $\sigma < 0$.
	
	Let us calculate the gradient of each test function.

	\begin{lemma}{Gradient of test functions}
		\label{lem:gradient_test_functions}
		Let $d \in \N$, $c > 0$, $\varepsilon > 0$ and $\vec{b} \in \{-1,1\}^d$. Then
		\[\nabla f_{d,c,\varepsilon,\vec{b}}(\vec{0}) = \frac{73\varepsilon}{d}\vec{b}.\]
	\end{lemma}

	\begin{proof}
		Let $\ell \in [d]$. We will use \autoref{eq:partial_derivative_test_functions} with $\alpha = (\ell)$ and $\vec{x} = \vec{0}$. Observe that whenever $j \neq \ell$, we are evaluating $\sin(cx_j)$ at $x_j = 0$ in the $j$th term of the summation, so this term vanishes. The only term that remains is the $\ell$th term, hence
		\[\partial_{\ell}f_{d,c,\varepsilon,\vec{b}}(\vec{0}) = \frac{73\varepsilon b_j}{cd} c\cos(cx_{\ell}) \cdot \prod_{\underset{k\neq\ell}{k=1}}^d \cos(cx_k) = \frac{73\varepsilon b_j}{cd} \cdot c = \frac{73\varepsilon b_j}{d}.\]
		Putting all these partial derivatives in a vector yields
		\[\nabla f_{d,c,\varepsilon,\vec{b}}(\vec{0}) = \frac{73\varepsilon}{d}\vec{b},\]
		completing the proof.
	\end{proof}

	Note that if we represented each of these gradients in $\R^d$ with a vertex, then we would obtain a Hamming cube with inradius $73\varepsilon/d$. This highlights the connection with \autoref{fig:Hamming_cube}.
	
	Now, we are ready to start the proof of the lower bound. We start from the assumption that we have any $\varepsilon$-precise $\ell^1$-approximate quantum gradient estimation algorithm $\A$ for $\G_{d,c,0,\R^d}$ on $G$ with success probability lower bounded by $17/18$. The following two lemmas deduce some properties that such an $\A$ must inevitably satisfy.

	\begin{lemma}{Vertex approximation}
		\label{lem:vertex_approximation}
		Let $d \in \N$, $c > 0$, $\varepsilon \in (0,c/146)$ and $G \subseteq \R^d$. Suppose $\A$ is an $\varepsilon$-precise $\ell^1$-approximate quantum gradient estimation algorithm for $\G_{d,c,0,\R^d}$ on $G$ with success probability lower bounded by $17/18$. Then there exists a $J \subseteq [d]$, $|J| \geq \frac34 \cdot d$ such that for all $j \in J$, there exists a set $F_j \subseteq \F_{d,c,\varepsilon}$ of size $|F_j| \geq \frac23 \cdot 2^d$ such that for all $f \in F_j$:
		\[\P\left[\left|\A(f)_j - \nabla f(\vec{0})_j\right| \leq \frac{72\varepsilon}{d}\right] \geq \frac23.\]
	\end{lemma}

	\begin{proof}
		Let $f \in \G_{d,c,0,\R^d}$ be arbitrary. From \autoref{def:quantum_gradient_estimation_algorithms} we infer that the quantum gradient estimation algorithm $\A$ has the following property:
		\[\P\left[\norm{\A(f) - \nabla f(\vec{0})}_1 \leq \varepsilon\right] \geq \frac{17}{18}.\]
		For every $f \in \G_{d,c,0,\R^d}$, we denote the event that the algorithm succeeds by $S_f$. Hence, the above line reduces to $\P[S_f] \geq 17/18$. Now, we find that
		\[\E\left[\norm{\A(f) - \nabla f(\vec{0})}_1 \mid S_f\right] = \E\left[\norm{\A(f) - \nabla f(\vec{0})}_1 \mid \norm{\A(f) - \nabla f(\vec{0})}_1 \leq \varepsilon\right] \leq \varepsilon.\]
		Note that we can rewrite the left-hand side as follows:
		\[\sum_{j=1}^d \E\left[\left|\A(f)_j - \nabla f(\vec{0})_j\right| \mid S_f\right] = \E\left[\left.\sum_{j=1}^d \left|\A(f)_j - \nabla f(\vec{0})_j\right| \right| S_f\right] = \E\left[\norm{\A(f) - \nabla f(\vec{0})}_1 \mid S_f\right] \leq \varepsilon.\]
		The above relation holds for all $f \in \G_{d,c,0,\R^d}$, so the result still holds if we average over the elements in $\F_{d,c,\varepsilon}$:
		\[\frac{1}{2^d} \sum_{f \in \F_{d,c,\varepsilon}} \sum_{j=1}^d \E\left[\left|\A(f)_j - \nabla f(\vec{0})_j\right| \mid S_f\right] \leq \varepsilon.\]
		As both summations are finite, we can swap them to obtain
		\[\sum_{j=1}^d \frac{1}{2^d} \sum_{f \in \F_{d,c,\varepsilon}} \E\left[\left|\A(f)_j - \nabla f(\vec{0})_j\right| \mid S_f\right] \leq \varepsilon.\]
		As all terms are non-negative, we now argue using a pigeonhole principle argument that at least $3d/4$ of the terms of the outer summation are upper bounded by $4\varepsilon/d$. Suppose there were more than $d/4$ terms that were not bounded by $4\varepsilon/d$. Then the resulting summation would exceed $\varepsilon$. But this is a contradiction, so there must be a set $J \subseteq [d]$ with $|J| \geq 3d/4$, such that for all $j \in J$,
		\[\frac{1}{2^d} \sum_{f \in \F_{d,c,\varepsilon}} \E\left[\left|\A(f)_j - \nabla f(\vec{0})_j\right| \mid S_f\right] \leq \frac{4\varepsilon}{d}.\]
		Let $j \in J$ arbitrarily. We now rewrite the above relation using Markov's inequality to obtain
		\[\frac{1}{2^d} \sum_{f \in \F_{d,c,\varepsilon}} \P\left[\left.\left|\A(f)_j - \nabla f(\vec{0})_j\right| \geq \frac{72\varepsilon}{d} \right| S_f\right] \leq \frac{1}{2^d} \sum_{f \in \F_{d,c,\varepsilon}} \frac{\E\left[\left|\A(f)_j - \nabla f(\vec{0})_j\right| \mid S_f\right]}{\frac{72\varepsilon}{d}} \leq \frac{\frac{4\varepsilon}{d}}{\frac{72\varepsilon}{d}} = \frac{1}{18}.\]
		And so, using Bayes' rule, we obtain
		\begin{align*}
			&\frac{1}{2^d} \sum_{f \in \F_{d,c,\varepsilon}} \P\left[\left|\A(f)_j - \nabla f(\vec{0})_j\right| \geq \frac{72\varepsilon}{d}\right] \\
			&= \frac{1}{2^d} \sum_{f \in \F_{d,c,\varepsilon}} \left(\P\left[\left.\left|\A(f)_j - \nabla f(\vec{0})_j\right| \geq \frac{72\varepsilon}{d} \right| S_f\right] \cdot \P(S_f) + \P\left[\left.\left|\A(f)_j - \nabla f(\vec{0})_j\right| \geq \frac{72\varepsilon}{d} \right| S_f^c\right] \cdot \P(S_f^c)\right) \\
			&\leq \frac{1}{2^d} \sum_{f \in \F_{d,c,\varepsilon}} \left(\P\left[\left.\left|\A(f)_j - \nabla f(\vec{0})_j\right| \geq \frac{72\varepsilon}{d} \right| S_f\right] \cdot 1 + 1 \cdot \frac{1}{18}\right) \\
			&= \frac{1}{2^d} \sum_{f \in \F_{d,c,\varepsilon}} \P\left[\left. \left|\A(f)_j - \nabla f(\vec{0})_j\right| \geq \frac{72\varepsilon}{d} \right| S_f \right] + \frac{1}{2^d} \cdot 2^d \cdot \frac{1}{18} \leq \frac{1}{18} + \frac{1}{18} = \frac19.
		\end{align*}
		Now, we can employ a similar pigeonhole principle argument to show that there must be a set $F_j \subseteq \F_{d,c,\varepsilon}$ such that $|F_j| \geq \frac23 \cdot 2^d$ and for all $f \in F_j$,
		\[\P\left[\left|\A(f)_j - \nabla f(\vec{0})_j\right| \geq \frac{72\varepsilon}{d}\right] \leq \frac13.\]
		But from here, we infer that for any $f \in F_j$,
		\[\P\left[\left|\A(f)_j - \nabla f(\vec{0})_j\right| \leq \frac{72\varepsilon}{d}\right] \geq \frac23.\]
		As we chose $j \in J$ and $f \in F_j$ arbitrarily, the above relation holds for all $j \in J$ and $f \in F_j$. This completes the proof.
	\end{proof}

	Let us pause here and develop some intuition for what the result of the previous lemma entails. Recall from \autoref{lem:gradient_test_functions} that all functions $f_{d,c,\varepsilon,\vec{b}} \in \F_{d,c,\varepsilon}$ have a gradient evaluated at $\vec{0}$ given by $\nabla f_{d,c,\varepsilon,\vec{b}}(\vec{0}) = (73\varepsilon/d)\vec{b}$. Hence, by taking the gradient at $\vec{0}$, we can relate the function $f_{d,c,\varepsilon,\vec{b}}$ to a vertex of the $d$-dimensional Hamming cube with inradius $73\varepsilon/d$ centered around the origin in $\R^d$. We say that a vertex $(73\varepsilon/d)\vec{b}$ of the Hamming cube is \textit{marked by the $j$th coordinate}, if the gradient of the corresponding test function $f_{d,c,\varepsilon,\vec{b}}$ is well-approximated in the $j$th coordinate by algorithm $\A$, in the sense that
	\begin{equation}
		\label{eq:marked}
		\P\left[\left|\A(f)_j - \nabla f(\vec{0})_j\right| \leq \frac{72\varepsilon}{d}\right] \geq \frac23.
	\end{equation}
	The previous lemma can now be very concisely rephrased. We have learned that at least three quarters of all coordinates mark at least two thirds of all vertices of the Hamming cube. Even more generically, we can say that the majority of all coordinates mark the majority of all vertices.
	
	Now, let $\vec{b}^{(1)},\vec{b}^{(2)} \in \{-1,1\}^d$ be such that $\vec{b}^{(1)} - \vec{b}^{(2)} = 2\vec{e}_j$, i.e., the vectors $\vec{b}^{(1)}$ and $\vec{b}^{(2)}$ differ only in the $j$th coordinate with $b^{(1)}_j = 1$ and $b^{(2)}_j = -1$. Another way to look at this is that $(73\varepsilon/d)\vec{b}^{(1)}$ and $(73\varepsilon/d)\vec{b}^{(2)}$ are adjacent vertices in the Hamming cube, with an adjoining edge pointed in the $j$th direction. Suppose that both vertices $\vec{b}^{(1)}$ and $\vec{b}^{(2)}$ are marked by the $j$th coordinate. Using \autoref{eq:marked} and \autoref{lem:gradient_test_functions}, we find that,
	\begin{equation}
		\label{eq:marked_intervals}
		\P\left[\A(f_{d,c,\varepsilon,\vec{b}^{(1)}})_j \in \left[\frac{\varepsilon}{d}, \frac{145\varepsilon}{d}\right]\right] \geq \frac23 \qquad \text{and} \qquad \P\left[\A(f_{d,c,\varepsilon,\vec{b}^{(2)}})_j \in \left[-\frac{145\varepsilon}{d}, -\frac{\varepsilon}{d}\right]\right] \geq \frac23.
	\end{equation}
	Note that these events are disjoint. In other words, if both $\vec{b}^{(1)}$ and $\vec{b}^{(2)}$ are marked by the $j$th coordinate, then any $\varepsilon$-precise $\ell^1$-approximate quantum gradient estimation algorithm with success probability at least $17/18$ can distinguish between the functions corresponding to these two vertices. If we can find many of these pairs, then we can take the vertex that is part of as many such pairs as possible, and apply the hybrid method with this vertex as the central instance.
	
	Since \autoref{lem:vertex_approximation} shows that there are many coordinates that select many vertices, it is intuitively clear that there must be many of these pairs. The aim of the following lemma is to make this intuition rigorous.

	\begin{lemma}{Edge separation}
		\label{lem:edge_separation}
		Let $d \in \N$, $c > 0$, $\varepsilon \in (0,c/146)$ and $G \subseteq \R^d$. Suppose that $\A$ is an $\varepsilon$-precise $\ell^1$-approximate quantum gradient estimation algorithm for $\G_{d,c,0,\R^d}$ on $G$ with success probability lower bounded by $17/18$. Then, there exists a $\vec{b}^* \in \{-1,1\}^d$ and a set $U \subseteq [d]$ of size $|U| \geq d/4$ such that for all $j \in U$,
		\[\P\left[\left|\A(f_{d,c,\varepsilon,\vec{b}^*})_j - \nabla f_{d,c,\varepsilon,\vec{b}^*}(\vec{0})_j\right| \leq \frac{72\varepsilon}{d}\right] \geq \frac23,\]
		and
		\[\P\left[\left|\A(f_{d,c,\varepsilon,\vec{b}^{\pm j}})_j - \nabla f_{d,c,\varepsilon,\vec{b}^{\pm j}}(\vec{0})_j\right| \leq \frac{72\varepsilon}{d}\right] \geq \frac23 \qquad \text{where} \qquad \vec{b}^{\pm j} = \vec{b}^* - 2b^*_j\vec{e}_j,\]
		i.e., $\vec{b}^{\pm j}$ differs from $\vec{b}^*$ solely in the $j$th entry.
	\end{lemma}

	\begin{proof}
		Let $\vec{b} \in \{-1,1\}^d$ and $j \in [d]$. Recall from \autoref{eq:marked} that vertex $\vec{b}$ of the Hamming cube $\{-1,1\}^d$ is marked by the $j$th coordinate if algorithm $\A$ approximates the $j$th coordinate of $\nabla f_{d,c,\varepsilon,\vec{b}}(\vec{0})$ sufficiently well. We say that edge $(\vec{b}_1,\vec{b}_2) \in (\{-1,1\}^d)^2$ is marked if for some $j \in [d]$, we have $\vec{b}_1 - \vec{b}_2 = \pm2\vec{e}_j$ and both $\vec{b}_1$ and $\vec{b}_2$ are marked by the $j$th coordinate.
		
		Let $J$ and $F_j$ be as in \autoref{lem:vertex_approximation}. Then for every $j \in J$, there are at least $2/3 \cdot 2^d$ vertices $\vec{b} \in \{-1,1\}^d$ that are marked by the $j$th coordinate. Moreover, there are a total of $2^{d-1}$ edges pointing in the $j$th direction in the $d$-dimensional Hamming cube, and the sets of  endpoints of these edges partition the set $\{-1,1\}^d$ into $2^{d-1}$ disjoint subsets of size $2$. As there are at least $2/3 \cdot 2^d = \frac43 \cdot 2^{d-1}$ marked vertices, we find, by the pigeonhole principle, that at least $\frac13 \cdot 2^{d-1}$ of these subsets satisfy the property that both vertices are marked. Thus, at least $\frac13 \cdot 2^{d-1}$ of the edges that point in the $j$th direction are marked.
		
		The above argument holds for all $j \in J$, and as $|J| \geq \frac34 \cdot d$, the total number of marked edges in the $d$-dimensional Hamming cube is at least $\frac34 \cdot d \cdot \frac13 \cdot 2^{d-1} = \frac14 \cdot d \cdot 2^{d-1}$. Moreover, there are a total of $d2^{d-1}$ edges in the $d$-dimensional Hamming cube, hence at least a quarter of them are marked. But this implies, again by the pigeonhole principle and because all vertices have equal degree, that there must be a vertex that has at least $d/4$ adjacent edges that are marked. Call this vertex $\vec{b}^*$, and call the directions in which the adjacent marked edges are pointing $U$. The result follows.
	\end{proof}

	Now, we just have to apply the hybrid method with the vertex $\vec{b}^*$ that we found in the previous lemma corresponding to the central instance, and with the neighboring vertices in the directions specified by $U$ corresponding to the peripheral instances.
	
	\begin{theorem}{Lower bound proof}
		\label{thm:lower_bound_proof}
		Let $d \in \N$, $c > 0$, $\varepsilon \in (0,c/146)$ and $G \subseteq \R^d$. Suppose $\A$ is an $\varepsilon$-precise $\ell^1$-approximate quantum gradient estimation algorithm for $\G_{d,c,0,\R^d}$ on $G$ with success probability lower bounded by $17/18$. Then, on every input $f \in \G_{d,c,0,\R^d}$, the resulting query complexity to (controlled) phase oracles $O_{f,G}$, denoted by $T_{\A}(f)$, satisfies
		\[\max_{f \in \G_{d,c,0,\R^d}} T_{\A}(f) \geq \frac{cd^{\frac32}}{876\varepsilon}.\]
	\end{theorem}

	\begin{proof}
		From \autoref{lem:edge_separation} we know that there exists a $\vec{b}^* \in \{-1,1\}^d$ and a $U \subseteq [d]$ of size $|U| \geq d/4$, such that for all $j \in U$,
		\begin{equation}
			\label{eq:disjoint_events}
			\P\left[\left|\A(f_{d,c,\varepsilon,\vec{b}^*})_j - \nabla f_{d,c,\varepsilon,\vec{b}^*}(\vec{0})_j\right| \leq \frac{72\varepsilon}{d}\right] \geq \frac23 \qquad \text{and} \qquad \P\left[\left|\A(f_{d,c,\varepsilon,\vec{b}^{(j)}})_j - \nabla f_{d,c,\varepsilon,\vec{b}^{(j)}}(\vec{0})_j\right| \leq \frac{72\varepsilon}{d}\right] \geq \frac23,
		\end{equation}
		where $\vec{b}^{(j)} \in \{-1,1\}^d$ differs from $\vec{b}^*$ only in the $j$th entry. From \autoref{lem:gradient_test_functions}, we obtain that
		\[\nabla f_{d,c,\varepsilon,\vec{b}^*}(\vec{0})_j = \frac{73\varepsilon b_j^*}{d} \qquad \text{and} \qquad \nabla f_{d,c,\varepsilon,\vec{b}^*}(\vec{0})_j = \frac{73\varepsilon b^{(j)}_j}{d}.\]
		As $b^*_j$ and $b^{(j)}_j$ are different by construction, $|b^*_j - b^{(j)}_j| = 2$, and hence the two events in \autoref{eq:disjoint_events} are disjoint. Thus, we can employ the hybrid method, as described in \autoref{thm:hybrid_method} to obtain
		\begin{align}
			\max_{f \in \G_{d,c,0,\R^d}} T_{\A}(f) &\geq \max\left(T_{\A}(f_{d,c,\varepsilon,\vec{b}^*}), \max_{j \in U} T_{\A}(f_{d,c,\varepsilon,\vec{b}^{(j)}})\right) \nonumber \\
			&\geq \sqrt{\frac{|U|}{\displaystyle 9 \max_{\underset{\norm{\ket{\psi}} = 1}{\ket{\psi} \in \C^{2^n}}}\sum_{j \in U} \norm{\left(O_{f_{d,c,\varepsilon,\vec{b}^*},G} - O_{f_{d,c,\varepsilon,\vec{b}^{(j)}},G}\right)\ket{\psi}}^2}}.
			\label{eq:hybrid_bound}
		\end{align}
		To bound the summation that appears in the denominator, take $\ket{\psi} \in \C^{2^n}$ such that $\norm{\ket{\psi}} = 1$ arbitrarily. Then,
		\begin{align*}
			\sum_{j \in U} \norm{\left(O_{f_{d,c,\varepsilon,\vec{b}^*},G} - O_{f_{d,c,\varepsilon,\vec{b}^{(j)}},G}\right)\ket{\psi}}^2 &= \sum_{j \in U} \norm{\left(\sum_{\vec{x} \in G} \ket{\vec{x}}\bra{\vec{x}}\right) \left(O_{f_{d,\varepsilon,\vec{b}^*},G} - O_{f_{d,\varepsilon,\vec{b}^{(j)}},G}\right) \left(\sum_{\vec{x'} \in G} \ket{\vec{x'}}\bra{\vec{x'}}\right) \ket{\psi}}^2 \\
			&= \sum_{j \in U} \sum_{\vec{x} \in G} \left|\sum_{\vec{x'} \in G} \bra{\vec{x}} \left(O_{f_{d,\varepsilon,\vec{b}^*},G} - O_{f_{d,\varepsilon,\vec{b}^{(j)}},G}\right) \ket{\vec{x'}}\braket{\vec{x'}}{\psi}\right|^2 \\
		\end{align*}
		As the phase oracles are diagonal operators, the innermost summation in the above expression is only non-zero if $\vec{x} = \vec{x'}$. Thus,
		\begin{align*}
			\sum_{j \in U} \norm{\left(O_{f_{d,c,\varepsilon,\vec{b}^*},G} - O_{f_{d,c,\varepsilon,\vec{b}^{(j)}},G}\right)\ket{\psi}}^2 &= \sum_{j \in U} \sum_{\vec{x} \in G} \left|\bra{\vec{x}} \left(O_{f_{d,\varepsilon,\vec{b}^*},G} - O_{f_{d,\varepsilon,\vec{b}^{(j)}},G}\right) \ket{\vec{x}}\right|^2 \cdot \left|\braket{\vec{x}}{\psi}\right|^2 \\
			&= \sum_{\vec{x} \in G} |\braket{\vec{x}}{\psi}|^2 \cdot \sum_{j \in U} \left|\bra{\vec{x}}\left(O_{f_{d,c,\varepsilon,\vec{b}^*},G} - O_{f_{d,c,\varepsilon,\vec{b}^{(j)}},G}\right)\ket{\vec{x}}\right|^2 \\
			&\leq \sum_{\vec{x} \in G} |\braket{\vec{x}}{\psi}|^2 \cdot \max_{\vec{x} \in G} \sum_{j \in U} \left|\bra{\vec{x}} \left(O_{f_{d,c,\varepsilon,\vec{b}^*},G} - O_{f_{d,c,\varepsilon,\vec{b}^{(j)}},G}\right) \ket{\vec{x}}\right|^2.
		\end{align*}
		The first factor is simply the norm of $\ket{\psi}$, which is $1$. Furthermore, from \autoref{def:phase_oracle} what the action of the phase oracles on the state $\vec{x}$ is. Hence, we obtain
		\begin{align*}
			\sum_{j \in U} \norm{\left(O_{f_{d,c,\varepsilon,\vec{b}^*},G} - O_{f_{d,c,\varepsilon,\vec{b}^{(j)}},G}\right)\ket{\psi}}^2 &\leq
			\max_{\vec{x} \in G} \sum_{j \in U} \left|e^{if_{d,c,\varepsilon,\vec{b}^*}(\vec{x})} - e^{if_{d,c,\varepsilon,\vec{b}^{(j)}}(\vec{x})}\right|^2 \\
			&\leq \max_{\vec{x} \in G} \sum_{j \in U} \left|f_{d,c,\varepsilon,\vec{b}^*}(\vec{x}) - f_{d,c,\varepsilon,\vec{b}^{(j)}}(\vec{x})\right|^2,
		\end{align*}
		where we used $|e^{ix} - e^{iy}| = |e^{i\frac{x+y}{2}}| \cdot |e^{i\frac{x-y}{2}} - e^{-i\frac{x-y}{2}}| = 2\sin\left|\frac{x-y}{2}\right| \leq |x-y|$ for all $x,y \in \R$ in the last line. By filling in the definition of the test functions, see \autoref{def:test_functions}, we obtain
		\begin{align*}
			\sum_{j \in U} \norm{\left(O_{f_{d,c,\varepsilon,\vec{b}^*},G} - O_{f_{d,c,\varepsilon,\vec{b}^{(j)}},G}\right)\ket{\psi}}^2 &\leq \max_{\vec{x} \in G} \sum_{j \in U} \left(\frac{73\varepsilon |b^*_j - b^{(j)}_j|}{cd}\right)^2 \cdot \sin^2(cx_j) \cdot \prod_{\underset{k \neq j}{k=1}} \cos^2(cx_k) \\
			&= \left(\frac{146\varepsilon}{cd}\right)^2 \cdot \max_{\vec{x} \in G} \sum_{j \in U} \sin^2(cx_j) \cdot \prod_{\underset{k\neq j}{k=1}}^d \cos^2(cx_k),
		\end{align*}
		where we used that $|b^*_j - b^{(j)}_j| = 2$. Adding non-negative terms to the right-hand side and relaxing our constraint on the choice of $\vec{x}$ from $G$ to $\R^d$, we obtain
		\begin{align*}
			\sum_{j \in U} \norm{\left(O_{f_{d,c,\varepsilon,\vec{b}^*},G} - O_{f_{d,c,\varepsilon,\vec{b}^{(j)}},G}\right)\ket{\psi}}^2 &\leq \left(\frac{146\varepsilon}{cd}\right)^2 \cdot \sup_{\vec{x} \in \R^d} \sum_{A \subseteq [d]} \prod_{j\in A} \sin^2(cx_j) \cdot \prod_{j \in [d] \setminus A} \cos^2(cx_j) \\
			&= \left(\frac{146\varepsilon}{cd}\right)^2 \cdot \sup_{\vec{x} \in \R^d} \prod_{j=1}^d \left(\sin^2(cx_j) + \cos^2(cx_j)\right) = \left(\frac{146\varepsilon}{cd}\right)^2.
		\end{align*}
		Since $|U| \geq d/4$, we find by plugging the above into \autoref{eq:hybrid_bound}:
		\[\max_{f \in \G_{d,c,0,\R^d}} T_{\A}(f) \geq \sqrt{\frac{d}{36}} \cdot \frac{cd}{146\varepsilon} = \frac{cd^{\frac32}}{876\varepsilon}.\]
		This completes the proof.
	\end{proof}

	\subsection{Lower bound for more general cases}
	\label{subsec:general_lower_bound_proof}
	
	In the previous section, we have proven a lower bound on the query complexity of the quantum gradient estimation problem for $p = 1$, $\sigma = 0$ and $P = 17/18$. The aim of the following theorem is to reduce the cases where $p \in [1,\infty]$, $\sigma \geq 0$ and $P \in (\frac12,1]$ to this single case, so that we can prove similar lower bounds for these cases as well.
	
	\begin{theorem}{Lower bound on gradient computation}
		\label{thm:general_lower_bound_proof}
		Let $d \in \N$, $c > 0$, $\varepsilon \in (0,c/(292d^{1-1/p}))$, $p \in [1,\infty]$, $\sigma \geq 0$, $P \in (\frac12,1]$ and $G \subseteq \R^d$. Suppose $\A$ is an $\varepsilon$-precise $\ell^p$-approximate quantum gradient estimation algorithm for $\G_{d,c,\sigma,\R^d}$ on $G$, with success probability lower bounded by $P$. Then
		\[\max_{f \in \G_{d,c,\sigma,\R^d}} T_{\A}(f) \geq \frac{cd^{\frac12 + \frac1p}}{1752N\varepsilon} \qquad \text{with} \qquad N = \left\lceil \frac{18(1-P)}{(P-\frac12)^2} \right\rceil.\]
	\end{theorem}

	\begin{proof}
		We will construct a new algorithm $\B$ as follows:
		\begin{enumerate}
			\setlength\itemsep{-.6em}
			\item Do $N$ independent runs of $\A$, and call the resulting vectors $\vec{g}_1, \dots, \vec{g}_N \in \R^d$.
			\item Search for a vector $\vec{g} \in \R^d$ such that there exists a $J \subseteq [N]$ of size $|J| > N/2$ and for all $j \in J$:
			\[\norm{\vec{g} - \vec{g}_j}_p \leq \varepsilon.\]
			If such a vector $\vec{g}$ exists, return $\vec{g}$. Otherwise, return $\vec{0}$.
		\end{enumerate}
		Now, we prove that $\B$ is a $2\varepsilon d^{1-1/p}$-precise $\ell^1$-approximate quantum gradient estimation algorithm for $\G_{d,c,\sigma,\R^d}$ on $G$ with success probability lower bounded by $17/18$. To that end, observe that every run of $\A$ with probability at least $P$ yields a vector $\vec{g}_j$ such that
		\[\norm{\vec{g}_j - \nabla f(\vec{0})}_p \leq \varepsilon.\]
		Let $B_j$ be the Bernoulli random variable that equals $1$ if and only if $\vec{g}_j$ satisfies the above property. We define
		\[P^* = \P\left[B_j = 1\right] \geq P > \frac12.\]
		Observe that
		\[\E\left[\sum_{j=1}^N B_j\right] = NP^* \qquad \text{and} \qquad \Var\left(\sum_{j=1}^N B_j\right) = NP^*(1-P^*).\]
		We find, using Chebyshev's inequality, as described in \autoref{thm:statistics}:
		\begin{align*}
			\P\left[\sum_{j=1}^N B_j \leq \frac{N}{2}\right] &= \P\left[\sum_{j=1}^N B_j - NP^* \leq N\left(\frac12 - P^*\right)\right] \leq \P\left[\left|\sum_{j=1}^N B_j - NP^*\right| \geq N\left(P^* -  \frac12\right)\right] \\
			&\leq \frac{\Var\left(\sum_{j=1}^N B_j\right)}{N^2\left(P^* - \frac12\right)^2} = \frac{NP^*(1-P^*)}{N^2(P^*-\frac12)^2} \leq \frac1N \cdot \frac{1-P}{(P-\frac12)^2} \leq \frac{(P-\frac12)^2}{18(1-P)} \cdot \frac{1-P}{(P-\frac12)^2} = \frac{1}{18}.
		\end{align*}
		Hence, with probability at least $17/18$, there exists a set $K \subseteq [N]$ such that $|K| > N/2$ and for all $j \in K$:
		\[\norm{\vec{g}_j - \nabla f(\vec{0})}_p \leq \varepsilon.\]
		Hence, with probability at least $17/18$, we can find at least one vector $\vec{g} \in \R^d$ such that there exists a set $J \subseteq [N]$ such that $|J| > N/2$ and for all $j \in J$ we have $\norm{\vec{g} - \vec{g}_j}_p \leq \varepsilon$, simply because $\nabla f(\vec{0})$ is such a vector. But as $|J| > N/2$ and $|K| > N/2$, we find by the pigeonhole principle that there must exist a $j \in J \cap K$. Hence,
		\[\norm{\nabla f(\vec{0}) - \vec{g}}_p \leq \norm{\nabla f(\vec{0}) - \vec{g}_j}_p + \norm{\vec{g}_j - \vec{g}}_p \leq \varepsilon + \varepsilon = 2\varepsilon.\]
		Thus, with probability at least $17/18$, the resulting vector $\vec{g}$ satisfies
		\[\norm{\nabla f(\vec{0}) - \vec{g}}_1 \leq d^{1-\frac1p} \cdot \norm{\nabla f(\vec{0}) - \vec{g}}_p \leq 2\varepsilon d^{1-\frac1p},\]
		so $\B$ is indeed a $2\varepsilon d^{1-1/p}$-precise $\ell^1$-approximate quantum gradient estimation algorithm for $\G_{d,c,0,\R^d} \subseteq \G_{d,c,\sigma,\R^d}$ on $G$ with success probability lower bounded by $17/18$. Moreover, by our choice of $\varepsilon$,
		\[2\varepsilon d^{1-\frac1p} < \frac{2d^{1-\frac1p}c}{292d^{1-\frac1p}} = \frac{c}{146}.\]
		Thus we can employ \autoref{thm:lower_bound_proof} to find
		\[\max_{f \in \G_{d,c,\sigma,\R^d}} T_{\B}(f) \geq \max_{f \in \G_{d,c,0,\R^d}} T_{\B}(f) \geq \frac{cd^{\frac32}}{876 \cdot 2\varepsilon d^{1-1/p}} = \frac{cd^{\frac12 + \frac1p}}{1752\varepsilon}.\]
		But by analyzing the construction of $\B$, we also find, for all $f \in \G_{d,c,\sigma,\R^d}$
		\[T_{\B}(f) = NT_{\A}(f).\]
		Hence,
		\[\max_{f \in \G_{d,c,\sigma,\R^d}} T_{\A}(f) \geq \frac{cd^{\frac12 + \frac1p}}{1752N\varepsilon}.\]
		This completes the proof.
	\end{proof}

	With this, we have reached the end of the lower bound proofs that are presented in this paper.
	
	\section{Conclusion and outlook}
	\label{sec:conclusion}
	
	In this paper, we have reached two main new results. First, we have generalized Gily\'en et al.'s quantum gradient estimation algorithm, so that it also works on Gevrey classes with parameter $\sigma \in (\frac12,1]$. Secondly, we have proved a lower bound on the query complexity of the quantum gradient estimation problem for $\sigma \in [0,\frac12]$ and $p \in [1,\infty]$.
	
	What seems most interesting is how the query complexity of the quantum gradient estimation problem scales with $d$ when we set some parameters $\sigma \in \R$ and $p \in [1,\infty]$. For the extremal values of $p$, we have drawn the currently best-known bounds on this query complexity in \autoref{fig:overview}. Note that simple coordinate-wise methods give linear in $d$ dependence when $p = \infty$ and quadratic in $d$ dependence when $p = 1$.
	
	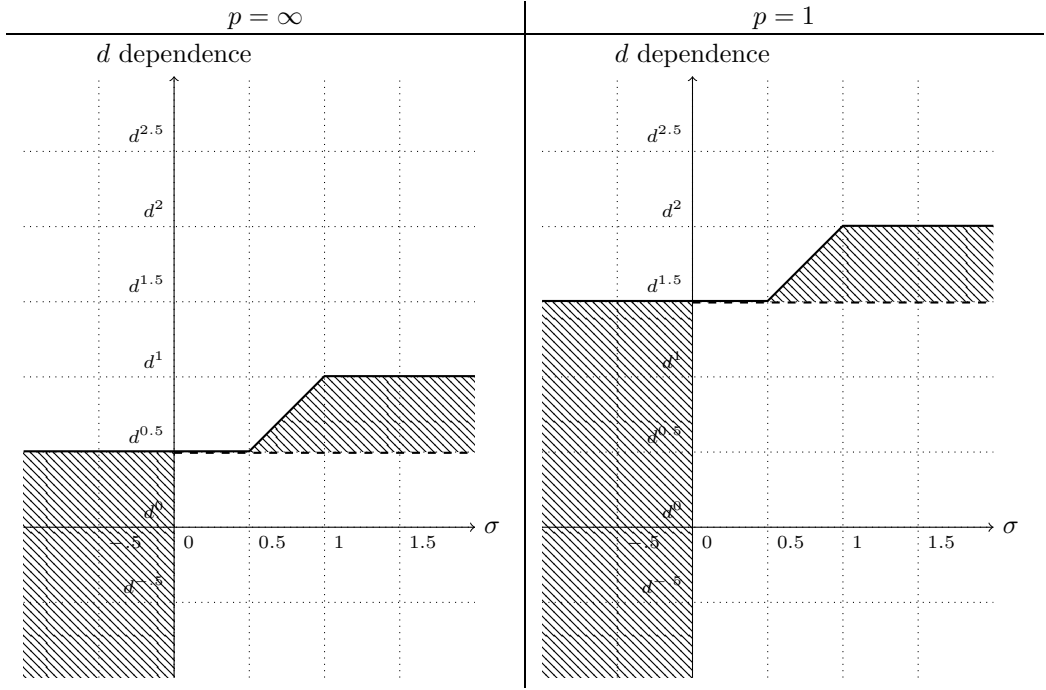
\begin{figure}[h!]
		\centering
		\begin{tabular}{c|c}
			$p = \infty$ & $p = 1$ \\\hline
			\begin{tikzpicture}[scale=2]
				\draw[->] (-1,0) -- (2,0) node[right] {$\sigma$};
				\draw[->] (0,-1) -- (0,3) node[above] {$d$ dependence};
				\fill[pattern=north west lines] (.5,.5) -- (1,1) -- (2,1) -- (2,.5) -- cycle;
				\fill[pattern=north west lines] (-1,.5) -- (0,.5) -- (0,-1) -- (-1,-1) -- cycle;
				\draw[thick,shift={(0,.005)}] (-1,.5) -- (.5,.5) -- (1,1) -- (2,1);
				\draw[thick,dashed,shift={(0,-.005)}] (0,.5) -- (2,.5);
				\foreach \x in {-.5,0,...,1.5}
				{
					\draw[dotted] (\x,-1) -- (\x,3);
					\node[below right] at (\x,0) {\scriptsize $\x$};
				}
				\foreach \y in {-.5,0,...,2.5}
				{
					\draw[dotted] (-1,\y) -- (2,\y);
					\node[above left] at (0,\y) {\scriptsize $d^{\y}$};
				}
			\end{tikzpicture} & 
			\begin{tikzpicture}[scale=2]
				\draw[->] (-1,0) -- (2,0) node[right] {$\sigma$};
				\draw[->] (0,-1) -- (0,3) node[above] {$d$ dependence};
				\fill[pattern=north west lines] (.5,1.5) -- (1,2) -- (2,2) -- (2,1.5) -- cycle;
				\fill[pattern=north west lines] (-1,1.5) -- (0,1.5) -- (0,-1) -- (-1,-1) -- cycle;
				\draw[thick,shift={(0,.005)}] (-1,1.5) -- (.5,1.5) -- (1,2) -- (2,2);
				\draw[thick,dashed,shift={(0,-.005)}] (0,1.5) -- (2,1.5);
				\foreach \x in {-.5,0,...,1.5}
				{
					\draw[dotted] (\x,-1) -- (\x,3);
					\node[below right] at (\x,0) {\scriptsize $\x$};
				}
				\foreach \y in {-.5,0,...,2.5}
				{
					\draw[dotted] (-1,\y) -- (2,\y);
					\node[above left] at (0,\y) {\scriptsize $d^{\y}$};
				}
			\end{tikzpicture}
		\end{tabular}
		\caption{Currently best-known bounds on the query complexity of the $\ell^p$-approximate quantum gradient estimation problem for Gevrey classes with parameter $\sigma \in \R$ for $p = \infty$ (left) and $p = 1$ (right). The solid line is the query complexity of the current best-known algorithm, and the dashed line is the current best-known lower bound on the query complexity. In this figure, we neglect multiplicative logarithmic growth factors.}
		\label{fig:overview}
	\end{figure}

	The hatched regions in the above figure represent the optimality gaps. Whenever the query complexity of the best-known quantum gradient estimation algorithm and the currently best-known lower bound do not match, the gap between them is hatched. While in the interval $[0,\frac12]$ optimality has been reached, there are quite some gaps remaining.
	
	Especially the gap at $\sigma = 1$ would be tempting to close. This is because the Gevrey class with parameter $\sigma = 1$ can be interpreted as the class of analytic functions, which is a very natural class of functions to consider, and pops up naturally in applications such as reinforcement learning.
	
	In this paper, we investigated lifting the gradient estimation problem in which the objective function is treated as a black box to the quantum domain. One could also consider lifting the problem from different settings to the quantum domain, for instance, a setting in which one has access to an oracle that calculates first order partial derivatives, or a setting in which the oracle circuit can be differentiated using techniques that are the quantum equivalent of automatic differentiation. These would all be interesting topics of further research.
	
	\section*{Acknowledgements}
	
	The research written down in this paper was part of the author's master's thesis~\cite{Cornelissen18}, which was part of the Applied Mathematics program at Delft University of Technology, and was conducted in cooperation with CWI. The author would first of all like to thank Ronald de Wolf for daily supervision of this project and many motivating and insightful discussions. Secondly, the author would like to thank Martijn Caspers for the daily supervision from the Delft University of Technology's side. Furthermore, the author would like to thank Andr\'as Gily\'en for very interesting conversations about the subject. Finally, the author would also like to thank Maris Ozols for doing a thorough review of this text and providing very useful comments on how the result is best presented.
	
	\bibliographystyle{alphaurl}
	\bibliography{mybib}
	
	\appendix
	
	\section{Miscellaneous results}
	\label{sec:misc_results}
	
	This appendix lists some results that are mainly included for reference. We also point to where one can find the proofs. The theorems listed here include results from probability theory, statistics, approximation theory, and some statements that are considered \textit{folklore} within the community of quantum computing.
	
	\begin{theorem}{Union bound}
		\label{thm:union_bound}
		Let $n \in \N$, $(\Omega,\Sigma,\P)$ be a probability space and $A_1, \dots, A_n \in \Sigma$ events. Then,
		\[\P\left[\bigcap_{j=1}^n A_j\right] \geq 1 - \sum_{j=1}^n \P\left[\Omega \setminus A_j\right].\]
	\end{theorem}

	\begin{proof}
		The proof is very elementary and can be found in any standard text on probability theory.
	\end{proof}
	
	\begin{theorem}{Results from statistics}
		\label{thm:statistics}
		Let $X$ be a real-valued random variable. We have:
		\[\Var\left[X\right] = \E\left[X^2\right] - \E\left[X\right]^2.\]
		If $X$ is non-negative and $t > 0$, we have:
		\[\P[X \geq t] \leq \frac{\E[X]}{t}. \qquad \text{(Markov's inequality)}\]
		Moreover, for all $t > 0$:
		\[\P\left[\left|X - \E[X]\right| > t\right] \leq \frac{\Var(X)}{t^2}. \qquad \text{(Chebyshev's inequatliy)}\]
		Finally, for all $t \geq 0$ and independently random variables $X_1, \dots, X_n$, where for every $i \in [n]$, $X_i$ is contained in the interval $[a_i,b_i]$ almost surely:
		\[\P\left(X_1 + \cdots + X_n - \E\left[X_1\right] - \cdots - \E\left[X_n\right] \leq -t\right) \leq e^{-\frac{2t^2}{\sum_{i=1}^n (b_i-a_i)^2}}. \qquad \text{(Hoeffding's inequality)}\]
	\end{theorem}
	
	\begin{proof}
		The proofs of these claims can be found in standard texts on statistics. For example, one can find them in \cite{Boucheron12}.
	\end{proof}
	
	\begin{theorem}{Stirling's approximation}
		\label{thm:stirling}
		Let $n \in \N$. We have:
		\[n^{n+\frac12}e^{-n}\sqrt{2\pi} \leq n! \leq n^{n+\frac12}e^{-n}e.\]
	\end{theorem}
	
	\begin{proof}
		The proof of this result can be found in many places, a slightly stronger result can for instance be found in \cite{Weisstein}.
	\end{proof}
	
	\begin{theorem}{Robustness of measurements}
		\label{thm:measurement_robustness}
		Let $n \in \N$ and $\delta,\varepsilon > 0$. Suppose we have two $n$-qubit states $\ket{\psi}$ and $\ket{\phi}$, such that $\norm{\ket{\psi} - \ket{\phi}} \leq \varepsilon$. Furthermore, suppose that a measurement with corresponding outcome set $\Omega$ performed on $\ket{\psi}$ yields an outcome in $S \subseteq \Omega$ with probability lower bounded by $\delta$. Then, performing the same measurement on $\ket{\phi}$ yields an element from $S$ with probability lower bounded by $\delta - \varepsilon$.
	\end{theorem}

	\begin{proof}
		This is a well-known result in quantum information theory. A proof can for instance be found in \cite{Cornelissen18}, Lemma B.1.
	\end{proof}

	\begin{theorem}{Hybrid method}
		\label{thm:hybrid_method}
		Let $n,N \in \N$. Suppose that $O_0, \dots, O_N$ are unitary operators acting on $n$ qubits. Let $\A$ be a quantum algorithm which is given access to the input oracle $O_j$ for one $j \in \{0,1,\dots, N\}$. For all $j \in \{1, \dots, N\}$, let $R_j$ and $R_j^*$ be disjoint sets. Suppose that for all $j \in \{1, \dots, N\}$,
		\[\P\left[\A(j) \in R_j\right] \geq \frac23 \qquad \text{and} \qquad \P\left[\A(0) \in R_j^*\right] \geq \frac23,\]
		where for all $j \in \{0,1,\dots,N\}$, $\A(j)$ is the measurement result of the algorithm run on input oracle $O_j$. Then the worst-case query complexity of $\A$ to the input oracle satisfies:
		\[T_{\A} \geq \sqrt{\frac{N}{\displaystyle 9 \max_{\underset{\norm{\ket{\psi}}=1}{\ket{\psi} \in \C^{2^n}}}\sum_{j=1}^N \norm{\left(O_j - O_0\right)\ket{\psi}}^2}}.\]
	\end{theorem}

	\begin{proof}
		One can find a similar result in the proof in \cite{Gilyen18}, Theorem 2. The full statement with all the details is proven in \cite{Cornelissen18}, Appendix C.
	\end{proof}
\end{document}